\documentclass[final]{siamltexmm}

\usepackage{amsmath}
\usepackage{amssymb} 
\usepackage{amsfonts}
\usepackage{verbatim}
\usepackage{mathrsfs}
\usepackage{subfigure}
\usepackage{dsfont}

\usepackage[usenames,dvipsnames]{pstricks}
\newtheorem{assumption}{Assumption}
\newtheorem{example}{Example}
\newtheorem{remark}{Remark}
\title{Aggregation and Control of Populations of Thermostatically Controlled Loads\\ by Formal Abstractions\thanks{This work is supported by the European Commission MoVeS project FP7-ICT-2009-257005,
by the European Commission Marie Curie grant MANTRAS 249295,
and by the NWO VENI grant 016.103.020.
This paper generalizes the results of \cite{SA13TCL}.}}

\author{Sadegh Esmaeil Zadeh Soudjani and Alessandro Abate\thanks{The authors are with the 
Delft Center for Systems and Control,  
TU Delft -- Delft University of Technology, The Netherlands.
A. Abate is also with the Department of Computer Science, University of Oxford, United Kingdom.
Emails: \email{S.EsmaeilZadehSoudjani@TUDelft.nl,\, A.Abate@TUDelft.nl}}}

\newcommand{\vectr}[1]{\mathbf{#1}}

\begin{document}
\maketitle
\newcommand{\slugmaster}{}

\begin{abstract}
This work discusses a two-step procedure, based on formal abstractions, 
to generate a finite-space stochastic dynamical model as an aggregation of the continuous temperature dynamics of a homogeneous population of Thermostatically Controlled Loads (TCL). 
The temperature of a single TCL is described by a stochastic difference equation and the TCL status (ON, OFF) by a deterministic switching mechanism.   
The procedure is formal as it allows the exact quantification of the error introduced by the abstraction -- as such it builds and improves on a known, 
earlier approximation technique in the literature. 
Further, the contribution discusses the extension to the case of a heterogeneous population of TCL by means of two approaches  
resulting in the notion of approximate abstractions.  
It moreover investigates the problem of global (population-level) regulation and load balancing for the case of TCL that are dependent on a control input.
The procedure is tested on a case study and benchmarked against the mentioned alternative approach in the literature.
\end{abstract}



\pagestyle{myheadings}
\thispagestyle{plain}
\markboth{S. Esmaeil Zadeh Soudjani and A. Abate}{Aggregation and Control of Thermostatically Controlled Loads by Formal Abstractions}

\section{Introduction}
\label{sec:intro}
Models for Thermostatically Controlled Loads (TCL) have shown potential to be employed in applications of load balancing and regulation. 
The shaping of the total power consumption of large populations of TCL, 
with the goal of tracking for instance the uncertain demand over the grid, 
while abiding to strict requirements on the users comfort, 
can lead to economically relevant repercussions for an energy provider.       
With this perspective, 
recent studies have focused on the development of usable models for aggregated populations of TCL.
The goal of the seminal work in \cite{HomTCL} is that of finding a reliable aggregated model under homogeneity assumptions over the population, 
meaning that all TCL are assumed to have the same dynamics and parameters.   
Under this assumption, \cite{HomTCL} puts forward a simple Linear Time-Invariant (LTI) model for a population characterized by an input as the temperature set-point, 
and an output as the total consumed power at the population level.   
The parameters of the LTI model are estimated based on data and the model is used to track fluctuations in electricity generation from wind.  
The work in \cite{HetTCL} proposes an approach, 
based on a partitioning of the TCL temperature range, 
to obtain an aggregate state-space model for a population of TCL that is heterogeneous over the TCL thermal capacitance.  
The full information of the state variables of the model is used to synthesize a control strategy for the output (namely, the total power consumption) tracking via a (deterministic) Model Predictive Control scheme.   
The contributions in \cite{HetTCLCallaway,MKC12} extend the results in \cite{HetTCL}
by considering a population of TCL that are heterogeneous over all their parameters. 
The Extended Kalman Filter is used to estimate the states of the model and to identify its state transition matrix. 
The control of the population is performed by switching ON/OFF a portion of the TCL population.    
Additional recent contributions have targeted extensions of the work in \cite{HomTCL} towards higher-order dynamics \cite{TPS_13} or population control \cite{KSBH11}. 

Although the control strategy in \cite{HomTCL} appears to be implementable over the current infrastructure with negligible costs, 
the model parameters are not directly related to the dynamics of the TCL population.  
On the other hand,
the control methods proposed in \cite{HetTCL,HetTCLCallaway} may be practically limited by implementation costs. 
Furthermore, the derivation of the state-space models in \cite{HetTCL,HetTCLCallaway} is valid under two unrealistic assumptions: 
first, the temperature evolution is assumed to be deterministic, leading to a deterministic state-space model; 
second, after partitioning the temperatures range in separate bins, the TCL temperatures are assumed to be uniformly distributed within each state bin. 
Moreover, from a practical standpoint there seems to be no clear connection between the precision of the aggregation and the performance of the model: 
increasing the number of state bins (that is, decreasing the width of the temperature intervals) does not necessarily improve the performance of the aggregated model.

This article proposes a two-step abstraction procedure to generate a finite stochastic dynamical model as an aggregation of the dynamics of a population of TCL.  
The approach relaxes the assumptions in \cite{HetTCL,HetTCLCallaway} by providing a model based on a probabilistic evolution of the TCL temperatures.  
The abstraction is made up of two separate parts: 
(1) going from continuous-space models for a TCL population to finite state-space models, which obtains a population of Markov chains; and 
(2) taking the cross product of the population of Markov chain and lumping the obtained model, 
by finding its coarsest probabilistically bisimilar Markov chain \cite{BK08}: as such the reduced-order Markov chain is an exact representation of the larger model.
The approach is fully developed for the case of a homogeneous population of TCL, 
and extended to a heterogeneous population -- 
however, in the latter case the aggregation (second step) employs an \emph{approximate} probabilistic bisimulation, 
which introduces an error.  
In both cases, it is possible to quantify the abstraction error of the first step, 
and in the homogeneous instance the error of the overall abstraction procedure can be quantified -- this is unlike the approaches in \cite{HomTCL,HetTCL,HetTCLCallaway}.  

The article also describes a dynamical model for the time evolution of the abstraction, 
and shows convergence result as the population size grows:  
increasing number of state bins always improves the accuracy, 
leading to a convergence of the error to zero.
This result is aligned with the work in \cite{BF11} on the aggregation of continuous-time deterministic thermostatic loads. 
The analytic relation between model and population parameters
enables the development of a set-point control strategy for reference tracking over the total load power
(cf. Figure \ref{fig:control_MPC}). 
A modified version of the Kalman Filter is employed to estimate the states 
and the power consumption of the population is regulated via a simple one-step regulation approach. 
As such the control architecture does not require knowledge of the single TCL states, but leverages directly the total power consumption.  
Alternatively, a stochastic model predictive control scheme is proposed. 
Both procedures are tested on a case study and 
the abstraction technique is benchmarked against the approach from \cite{HetTCL,HetTCLCallaway}.   

\smallskip

The article is organized as follows. 
Section \ref{sec:formal_abst}, after introducing the model of the single TCL dynamics, 
describes its abstraction as a Markov Chain, 
and further discusses the aggregation of a homogeneous population of TCL -- the errors introduced by both steps are exactly quantified. 
Section \ref{sec:exten_heter} focuses on heterogeneous populations of TCL and elucidates two techniques to aggregate their dynamics: 
one based on averaging, and a second based on clustering the uncertain parameters. 
The latter approach allows for a general quantification of the error. 
Section \ref{sec:controlled} discusses TCL models endowed with a control input, 
and the synthesis of global (acting at the population level
 -- cf. Figure \ref{fig:control_MPC})
controllers to achieve regulation of the total consumed power -- this is achieved by two alternative schemes.  
Finally, all the discussed techniques are tested on a case study described in Section \ref{sec:benchmarks}. 
%
Tables \ref{tab:discrete_distributions}, \ref{tab:discrete_distributions1} recapitulate quantities and \ref{tab:out} discusses some nomenclature, 
introduced in this work. 
%
%
\section{Formal Abstraction of a Homogeneous Population of TCL}
\label{sec:formal_abst}

\subsection{Continuous Model of the Temperature of a Single TCL}

Throughout this article we use the notation $\mathbb{N}$ to denote the natural numbers,
$\mathbb{Z} = \mathbb{N}\cup\{0\}$,
$\mathbb{N}_n = \{1,2,3,\cdots,n\}$,
and $\mathbb{Z}_n =\mathbb{N}_n\cup \{0\}$.
We denote vectors with bold typeset and a letter corresponding to that of its elements. 

The evolution of the temperature in a single TCL can be characterized by the following stochastic difference equation \cite{HomTCL,MC85} 
\begin{equation}
\label{eq:tcl_dyn}
\theta(t+1) = a\,\theta(t)+(1-a)(\theta_a \pm m(t) R P_{rate}) + w(t),
\end{equation}
where
$\theta_a$ is the ambient temperature, 
$C$ and $R$ indicate the thermal capacitance and resistance respectively, 
$P_{rate}$ is the rate of energy transfer,
and $a = e^{-h/RC}$, 
with a discretization step $h$. 
The process noise $w(t), t\in\mathbb Z$, is made up by i.i.d. random variables characterized by a density function $t_w(\cdot)$.
We denote with $m(t) = 0$ a TCL in the OFF mode at time $t$, 
and with $m(t) = 1$ a TCL in the ON mode. 
In equation \eqref{eq:tcl_dyn} a plus sign is used for a heating TCL, 
whereas a minus sign for a cooling TCL. 
In this work we focus on a population of cooling TCL, 
with the understanding that the case of heating TCL can be similarly obtained. 
The distributions of the initial temperature and mode are denoted by $\pi_0(m,\theta)$, respectively.  
The temperature dynamics for the cooling TCL is regulated by the discrete switching control 
$m(t+1) = f(m(t),\theta(t))$, 
where
\begin{equation}
\label{eq:switch}
f(m,\theta)=
\left\{
\begin{array}{ll}
0, & \theta < \theta_s - \delta/2 \doteq \theta_{-}\\ 
1,  & \theta > \theta_s + \delta/2 \doteq \theta_{+}\\
m, & \text{else,}
\end{array}
\right. 
\end{equation}
where $\theta_s$ denotes a temperature set-point and $\delta$ a dead-band, 
and together characterizing a temperature range.  
The power consumption of the single TCL at time $t$ is equal to $\frac{1}{\eta} m(t) P_{rate}$,
which is equal to zero in the OFF mode and positive in the ON mode, 
and where the parameter $\eta$ is the Coefficient Of Performance (COP).
The constant $\frac{1}{\eta}P_{rate}$, namely the power consumption of TCL in the ON mode, 
will be shortened as $P_{rate,ON}$ in the sequel. 

\subsection{Finite Abstraction of a Single TCL by State-Space Partitioning}
\label{sec:state_partition}
The composition of the dynamical equation in \eqref{eq:tcl_dyn} with the algebraic relation in \eqref{eq:switch} 
allows us to consider a single TCL as a Stochastic Hybrid System \cite{APLS08},
namely as a discrete-time Markov process evolving over a hybrid (namely, discrete/continuous) state-space. 
A hybrid state-space is characterized by a variable $s = (m,\theta) \in \mathbb Z_1\times\mathbb R$ with two components, 
a discrete ($m$) and a continuous ($\theta$) one.  
The one-step transition density function of the stochastic process, $t_s(\cdot|s)$, 
made up of the dynamical equations in \eqref{eq:tcl_dyn}, \eqref{eq:switch}, 
and conditional on point $s$,
can be computed as
\begin{align*}
t_s\left((m',\theta')|(m,\theta)\right) = 
\delta[m'-f(m,\theta)]t_w(\theta'-a\,\theta-(1-a)(\theta_a- m R P_{rate})),
\end{align*}
where $\delta[\cdot]$ denotes the discrete unit impulse function.  
This interpretation allows leveraging an abstraction technique,
proposed in \cite{APKL10} and extended in \cite{SA13,SAH12}, 
aimed at reducing a discrete-time, uncountable state-space Markov process into a 
(discrete-time) finite-state Markov chain. 
This abstraction is based on a state-space partitioning procedure as follows. 
Consider an arbitrary, finite partition of the continuous domain $\mathbb{R} = \cup_{i=1}^{n}\Theta_i$, 
and arbitrary representative points within the partitioning regions denoted by $\{\bar\theta_i\in\Theta_i,i\in\mathbb{N}_n\}$. 
Introduce a finite-state Markov chain $\mathcal M$, 
characterized by $2n$ states $s_{im} = (m,\bar\theta_i),m\in\mathbb{Z}_1,i\in\mathbb{N}_n$. 
The transition probability matrix related to $\mathcal M$ is made up of the following elements 
\begin{align}
\label{eq:trans_prob}
\mathsf P(s_{im},s_{i' m'})& = \int_{\Theta_{i'}}t_s\left((m',\theta')|m,\bar\theta_i\right)d\theta',\quad
\forall m'\in\mathbb{Z}_1,i'\in\mathbb{N}_n.
\end{align} 
The initial probability mass for $\mathcal M$ is obtained as
$
p_0(s_{im}) = \int_{\Theta_{i}}\pi_0(m,\theta)d\theta.
$
For ease of notation we rename the states of $\mathcal{M}$ by the bijective map 
$\ell(s_{im}) = m n + i, m\in\mathbb{Z}_1, i\in\mathbb{N}_n$,
and accordingly we introduce the new notation
\begin{equation*}
P_{ij} =\mathsf P(\ell^{-1}(i),\ell^{-1}(j)),\quad p_{0i} = p_0(\ell^{-1}(i)),\quad \forall i,j\in\mathbb{N}_{2n}.
\end{equation*}

Notice that the conditional density function of the stochastic system capturing the dynamics of a single TCL is discontinuous,  
due to the presence of equation \eqref{eq:switch}. 
This can be emphasized by the following alternative representation of the discontinuity in the discrete conditional distribution,
for all $m,m' \in\mathbb{Z}_1,\theta\in\mathbb{R}$:
\begin{align*}
\delta\left[m'-f(m,\theta)\right] & =
m'\mathbb{I}_{(\theta_+,\infty)}(\theta) + (1-m')\mathbb{I}_{(-\infty,\theta_-)}(\theta)
+ (1-|m-m'|)\mathbb{I}_{[\theta_-,\theta_+]}(\theta),
\end{align*}
where $\mathbb{I}_{\mathcal A} (\cdot)$ denotes the indicator function of a general set $\mathcal A$. 
The selection of the partitioning sets then requires special attention: 
a convenient way to obtain that is to select a partition for the dead-band $[\theta_-,\theta_+]$, 
thereafter extending it to a partition covering the whole real line $\mathbb R$ (cf. Figure \ref{fig:part}). 
Let us select two constants $\mathsf l,\mathsf m\in\mathbb{N}$, $\mathsf l<\mathsf m$,
compute the partition size $\upsilon = \delta/2\mathsf l$ and quantity $\mathcal L = 2\mathsf m\upsilon$. 
Now construct the boundary points of the partition sets $\{\theta_i\}_{i=-\mathsf m}^{i=\mathsf m}$ for the temperature axis as follows: 
\begin{align}
& \theta_{\pm \mathsf l} = \theta_{s}\pm\delta/2,\quad \theta_{\pm\mathsf m} = \theta_{s}\pm\mathcal L/2,\quad \theta_{i+1} = \theta_i+\upsilon,\nonumber\\
& \mathbb{R} = \cup_{i=1}^{n}\Theta_i,\quad\Theta_1 = (-\infty,\theta_{-\mathsf m}),\quad \Theta_n = [\theta_{\mathsf m},\infty),\label{eq:partition}\\
& \Theta_{i+1} = [\theta_{-\mathsf m+i-1},\theta_{-\mathsf m+i}),\quad i\in\mathbb N_{n-2}, \quad n = 2\mathsf m+2,\nonumber
\end{align}
and let us render the Markov states of the infinite-length intervals $\Theta_1, \Theta_n$ absorbing.

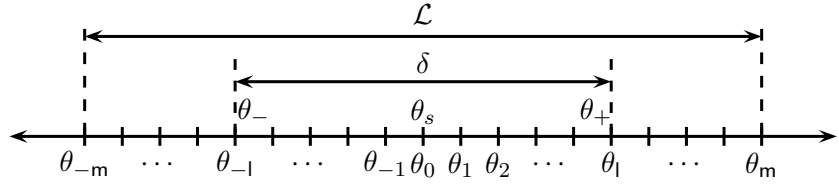
\begin{figure}
\centering
\scalebox{1}
{
\begin{pspicture}(-1.5,-1)(11,1)
\psline[linewidth=0.04cm,arrowsize=0.05291667cm 3.0,arrowlength=1.4,arrowinset=0.4]{<->}(-1,-0.725)(10,-0.725)
\psline[linewidth=0.04cm](0,-0.6)(0,-0.85)
\psline[linewidth=0.04cm](0.5,-0.6)(0.5,-0.85)
\psline[linewidth=0.04cm](1,-0.6)(1,-0.85)
\psline[linewidth=0.04cm](1.5,-0.6)(1.5,-0.85)
\psline[linewidth=0.04cm](2,-0.6)(2,-0.85)
\psline[linewidth=0.04cm](2.5,-0.6)(2.5,-0.85)
\psline[linewidth=0.04cm](3,-0.6)(3,-0.85)
\psline[linewidth=0.04cm](3.5,-0.6)(3.5,-0.85)
\psline[linewidth=0.04cm](4,-0.6)(4,-0.85)
\psline[linewidth=0.04cm](4.5,-0.6)(4.5,-0.85)
\psline[linewidth=0.04cm](5,-0.6)(5,-0.85)
\psline[linewidth=0.04cm](5.5,-0.6)(5.5,-0.85)
\psline[linewidth=0.04cm](6,-0.6)(6,-0.85)
\psline[linewidth=0.04cm](6.5,-0.6)(6.5,-0.85)
\psline[linewidth=0.04cm](7,-0.6)(7,-0.85)
\psline[linewidth=0.04cm](7.5,-0.6)(7.5,-0.85)
\psline[linewidth=0.04cm](8,-0.6)(8,-0.85)
\psline[linewidth=0.04cm](8.5,-0.6)(8.5,-0.85)
\psline[linewidth=0.04cm](9,-0.6)(9,-0.85)
\psline[linewidth=0.04cm,arrowsize=0.05291667cm 3.0,arrowlength=1.4,arrowinset=0.4]{<->}(2,0)(7,0)
\psline[linewidth=0.04cm,arrowsize=0.05291667cm 3.0,arrowlength=1.4,arrowinset=0.4]{<->}(0,0.6)(9,0.6)
\psline[linewidth=0.04cm,linestyle=dashed,dash=0.16cm 0.16cm](2,0.3)(2,-0.7)
\psline[linewidth=0.04cm,linestyle=dashed,dash=0.16cm 0.16cm](7,0.3)(7,-0.7)
\psline[linewidth=0.04cm,linestyle=dashed,dash=0.16cm 0.16cm](0,0.7)(0,-0.7)
\psline[linewidth=0.04cm,linestyle=dashed,dash=0.16cm 0.16cm](9,0.7)(9,-0.7)
\rput(4.5,0.9){$\mathcal L$}
\rput(4.5,0.25){$\delta$}
\rput(2.25,-0.4){$\theta_-$}
\rput(6.8,-0.4){$\theta_+$}
\rput(4.5,-0.4){$\theta_s$}
\rput(3,-1.1){$\cdots$}
\rput(4,-1.1){$\theta_{-1}$}
\rput(4.5,-1.1){$\theta_0$}
\rput(5,-1.1){$\theta_1$}
\rput(5.5,-1.1){$\theta_2$}
\rput(6.2,-1.1){$\cdots$}
\rput(7,-1.1){$\theta_{\mathsf l}$}
\rput(8,-1.1){$\cdots$}
\rput(1,-1.1){$\cdots$}
\rput(2,-1.1){$\theta_{-\mathsf l}$}
\rput(9,-1.1){$\theta_{\mathsf m}$}
\rput(0,-1.1){$\theta_{-\mathsf m}$}
\end{pspicture} 
}
\caption{Partitioning of the temperature axis for the abstraction of the dynamics of a single TCL.}
\label{fig:part}
\end{figure}

Let us emphasize that the discontinuity in the discrete transition kernel $\delta\left[m'-f(m,\theta)\right]$ 
and the above partition induce the following structure on the transition probability matrix of the chain $\mathcal M$:
\begin{equation}
\label{eq:structure_P}
P =
\left[
\begin{array}{cccc}
Q_{11} & 0\\
0 & Q_{22}\\
Q_{31} & 0\\
0 & Q_{42}
\end{array}
 \right],
\end{equation}
where $Q_{11},Q_{42}\in\mathbb R^{(\mathsf m+\mathsf l+1)\times n}$, 
whereas $Q_{22},Q_{31}\in\mathbb R^{(\mathsf m-\mathsf l+1)\times n}$, 
which leads to $P \in\mathbb R^{2n \times 2n}$. 

Clearly, the abstraction of the dynamics in \eqref{eq:tcl_dyn}-\eqref{eq:switch} over this partition of the state-space leads to a discretization error: 
in the next section we formally derive bounds on this error as a function of the partition size $\upsilon$ and of the quantity $\mathcal L$. 
This guarantees the convergence (in expected value) of the power consumption of the model to that of the entire population \cite{APKL10,SA13,SAH12}.

\subsection{Aggregation of a Population of TCL by Bisimulation Relation}

Consider now a population of $n_p$ homogeneous TCL, 
that is a population of TCL which, 
after possible rescaling of \eqref{eq:tcl_dyn}-\eqref{eq:switch}, 
share the same set of parameters $\theta_s,\delta,\theta_a,C,R, P_{rate},P_{rate,ON}$ (and thus $\eta$), $h$, and noise terms $t_w(\cdot)$.
Each TCL can be abstracted as a Markov chain $\mathcal M$ with the same transition probability matrix $P = [P_{ij}]_{i,j}$, 
where $i,j\in\mathbb{N}_{2n}$, 
which leads to a population of $n_p$ homogeneous Markov chains.   
The initial probability mass vector $p_0 = [p_{0i}]_i$ might vary over the population.

The homogeneous population of TCL can be represented by a single Markov chain $\varXi$, 
built as the cross product of the $n_p$ homogeneous Markov chains.  
The state of the Markov chain $\varXi$ is 
\begin{equation*}
\vectr z = [z_1, z_2, \cdots, z_{n_p}]^T\in\mathcal Z = \mathbb N_{2n}^{n_p},
\end{equation*}
where $z_j\in\mathbb N_{2n}$ represents the state of the $j^{\text{th}}$ Markov chain. 
We denote by $P_{\varXi}$ the transition probability matrix of $\varXi$. 

It is understood that $\varXi$, having exactly $(2n)^{n_p}$ states, can in general be quite large, 
and thus cumbersome to manipulate computationally. 
As the second step of the abstraction procedure,  
we are interested in aggregating this model and employ the notion of probabilistic bisimulation to achieve this \cite{BK08}.   
Let us introduce a finite set of atomic propositions as 
a constrained vector with a dimension corresponding to the number of bins of the single $\mathcal M$: 
\begin{equation*}
AP = \left\{\vectr x = [x_1,x_2\cdots,x_{2n}]^T\in\mathbb Z_{n_p}^{2n}\bigg|\sum_{r=1}^{2n}x_r = n_p\right\}. 
\end{equation*}
The labeling function $L:\mathcal Z\rightarrow AP$ associates to a configuration $\vectr z$ of $\varXi$ 
a vector $\vectr x = L(\vectr z)$, which elements $x_i\in\mathbb Z_{n_p}$ count the number of thermostats in bin $i, i\in\mathbb N_{2n}$.  
Notice that
the set $AP$ is finite with cardinality $|AP| = (n_p+2n-1)!/(n_p!(2n-1)!)$,
which for $n_p\ge 2$ is (much) less than the cardinality $(2n)^{n_p}$ of $\varXi$.
%
%

Let us define an equivalence relation $\mathcal R$ \cite{BK08} on the state-space of $\mathcal Z$, 
such that 
\begin{equation*}
(\vectr z,\vectr z')\in \mathcal R \Leftrightarrow L(\vectr z) = L(\vectr z'). 
\end{equation*}
A pair of elements of $\mathcal Z$ is in the relation whenever the corresponding number of TCL in each introduced bin are the same (recall that the TCL are assumed to be homogeneous).   
Such an equivalence relation provides a partition of the state-space of $\mathcal Z$ into equivalence classes belonging to the quotient set $\mathcal Z / \mathcal R$,
where each class is uniquely specified by the label associated to its elements.
We plan to show that $\mathcal R$ is an exact probabilistic bisimulation relation on $\varXi$ \cite{BK08}, 
which requires proving that, 
for any set $\mathscr T \in \mathcal Z / \mathcal R$ and any pair $(\vectr z,\vectr z')\in \mathcal R$ 
\begin{align}
\label{tp4pb}
P_{\varXi} (\vectr z,\mathscr T) = P_{\varXi} (\vectr z',\mathscr T), 
\end{align}
This is achieved by Corollary \ref{cor:pbs} in the next Section. 
We now focus on the stochastic properties of $\varXi$, 
which we study through its quotient Markov chain obtained via $\mathcal R$. 

\subsection{Properties of the Aggregated Quotient Markov Chain}  
\label{sec:stoch_proper}

\begin{table}
\centering
\caption{Construction of probability distributions of interest in this work}
\scalebox{0.85}{%
\begin{tabular}{| l l |}
\hline
Name of the distribution & Interpretation\\
\hline
Bernoulli trials &  the result of a random event that takes on one of $\mathsf n=2$ possible outcomes\\
binomial & sum of independent Bernoulli trials with the same success probability\\
Poisson-binomial & sum of independent Bernoulli trials with success probabilities $p_1,p_2,\cdots,p_{n_p}$\\ 
\hline 
categorical & result of a random event that takes on one of $\mathsf n>2$ possible outcomes\\
multinomial & sum of categorical random variables with the same parameters\\
generalized multinomial & sum of categorical random variables with different parameters\\
\hline
\end{tabular}
}
\label{tab:discrete_distributions}
\end{table}

\begin{table}
\centering
\caption{Properties of probability distributions of interest in this work}
\scalebox{0.82}{%
\begin{tabular}{| l l c c c|}
\hline
Name of the distribution & Support & Parameters & Mean & Variance and covariance\\
\hline
Bernoulli & $\mathbb Z_1$ & $p$ success probability & $p$ & $p(1-p)$\\ 
binomial &$\mathbb Z_{n_p}$ & $p,n_p$ & $n_pp$ & $n_pp(1-p)$\\
Poisson-binomial &$\mathbb Z_{n_p}$ & $p_1,\cdots,p_{n_p}$ & $\sum\limits_{r=1}^{n_p} p_r$ & $\sum\limits_{r=1}^{n_p} p_r(1-p_r)$\\
\hline 
categorical & $\mathbb{N}_{\mathsf n}$ & $\vectr{p} = [p_1,p_2,\cdots,p_{\mathsf n}]$ & $\sum\limits_{r=1}^{\mathsf n} rp_r$& $\sum\limits_{r=1}^{\mathsf n}r^2p_r - \left(\sum\limits_{r=1}^{\mathsf n} rp_r\right)^2$ \\
multinomial & $\mathbb{N}_{\mathsf n}^{n_p}$& $\vectr{p}_r = [p_1,p_2,\cdots,p_{\mathsf n}], r\in\mathbb N_{n_p}$ & $n_pp_i$ &
$n_p p_i(1-p_i)$, 
$-n_p p_i p_j$ \\
generalized multinomial & $\mathbb{N}_{\mathsf n}^{n_p}$ & $\vectr{p}_r = [p_{r1},p_{r2},\cdots,p_{r\mathsf n}],r\in\mathbb N_{n_p}$ &
$\sum\limits_{r=1}^{n_p} p_{ri}$ &
$\sum\limits_{r=1}^{n_p} p_{ri}(1-p_{ri})$, 
$-\sum\limits_{r=1}^{n_p} p_{ri}p_{rj}$ \\
\hline
\end{tabular}
}
\label{tab:discrete_distributions1}
\end{table}

Let us recall the definition of known discrete random variables that are to be used for describing quantities obtained by the abstraction procedure.   
The sum of independent Bernoulli trials characterized by the same success probability follows a \emph{binomial distribution} \cite{D05}.
If a random variable $Y$ is instead defined as the sum of $n_p$ independent Bernoulli trials with different success probabilities ($p_1,p_2,\cdots,p_{n_p}$), 
then $Y$ has a \emph{Poisson-binomial distribution} \cite{PoissBin93} with the sample space $\mathbb Z_{n_p}$ and the following mean and variance:
\begin{equation*}
\mathbb E[Y] = \sum_{r=1}^{n_p}p_r,\quad var(Y) = \sum_{r=1}^{n_p}p_r(1-p_r).
\end{equation*}
As a generalization of the Bernoulli trials, 
a \emph{categorical distribution} describes the result of a random event that takes on one of $\mathsf n>2$ possible outcomes. 
%
Its sample space is taken to be 
$\mathbb{N}_{\mathsf n}$ and its 
probability mass function $\vectr{p} = [p_1,p_2,\cdots,p_{\mathsf n}]$, 
such that $\sum_{i=1}^{\mathsf n} p_i = 1$. 
A \emph{multinomial distribution} is a generalization of the binomial distribution as the sum of categorical random variables with the same parameters. 
The sum of categorical random variables with different parameters follows instead the \emph{generalized multinomial distribution}, defined as follows \cite{JKB97}.
%
%
%
Consider $n_p$ independent categorical random variables defined over the same sample space $\mathbb N_{\mathsf n}$ but with different outcome probabilities  
$\vectr p_r = [p_{r 1},p_{r 2},\cdots,p_{r \mathsf n}]$, $r\in\mathbb N_{n_p}$. 
Let the random variable $Y_i$ indicate the number of times the $i^{\text{th}}$ outcome is observed over $n_p$ samples.
Then vector $\vectr Y = [Y_1, ..., Y_{\mathsf n}]^T$ has a generalized multinomial distribution characterized by 
\begin{align*}
& \mathbb E[Y_i] = \sum_{r=1}^{n_p} p_{ri},\quad
var(Y_i) = \sum_{r=1}^{n_p} p_{ri}(1-p_{ri}),\quad
cov(Y_i,Y_j) = -\sum_{r=1}^{n_p} p_{ri}p_{rj}\quad (i\ne j).
\end{align*}
%
%

\medskip

We now study the one-step probability mass function associated to the codomain of the labeling function (that is, any of the labels), 
conditional on the state of the chain $\varXi$. 
\begin{theorem}
\label{thm:pois_bino}
The conditional random variable $(x_i(t+1)|\vectr z(t))$, $i\in\mathbb N_{2n}$, 
has a Poisson-binomial distribution over the sample space $\mathbb Z_{n_p}$, 
with the following mean and variance:
\begin{align}
\label{eq:mean_var_Poi}
\mathbb E[x_i(t+1)|\vectr z(t)] = \sum\limits_{r=1}^{n_p}P_{z_r(t)i},\quad
var(x_i(t+1)|\vectr z(t)) = \sum\limits_{r=1}^{n_p}P_{z_r(t)i}(1-P_{z_r(t)i}).
\end{align}
\end{theorem}


Conditional on an observation $\vectr x = [x_1,x_2,\cdots,x_{2n}]^T$ at time $t$ 
over the Markov chain $\varXi$, 
it is of interest to compute the probability mass function of the conditional random variable $(x_i(t+1)|\vectr x(t))$ as 
$\mathsf P(x_i(t+1) = j| \vectr x(t))$, for any $j\in\mathbb Z_{n_p}$ --- notice the difference with the quantity discussed in \eqref{eq:mean_var_Poi}.
%
For any label $\vectr x = [x_1,\cdots,x_{2n}]^T$ there are exactly
$n_p!/(x_1!x_2!\cdots x_{2n}!)$
states of $\varXi$ such that $L(\vectr z) = \vectr x$.
We use the notation $\vectr z\rightarrow\vectr x$ to indicate the states in $\varXi$ associated to label $\vectr x$, that is $\vectr z: L(\vectr z) = \vectr x$.

Based on the law of total probability for conditional probabilities, we can write
\begin{align}
\mathsf P(x_i(t+1)  = j|\vectr x(t))
&= \frac{\sum_{\vectr z(t)\rightarrow\vectr x(t)}
\mathsf P(x_i(t+1) = j|\vectr z(t))\mathsf P(\vectr z(t))}{\mathsf P(\vectr x(t))}\label{eq:total_prob}\\  
& = \mathsf P(x_i(t+1) = j|\vectr z(t))\frac{\sum_{\vectr z(t)\rightarrow\vectr x(t)}\mathsf P(\vectr z(t))}
{\mathsf P(\vectr x(t))}
= \mathsf P(x_i(t+1) = j|\vectr z(t)), \nonumber
\end{align}
where the sum is over all states $\vectr z(t)$ of $\varXi$ such that $L(\vectr z(t)) = \vectr x(t)$: 
in these states we have 
$x_1(t)$ Markov chains in state 1 with probability $P_{1i}$,
$x_2(t)$ Markov chains in state 2 with probability $P_{2i}$,
and so on.
The simplification above is legitimate since the probability of having a label $\vectr x = (x_1,x_2,\cdots,x_{2n})$
is exactly the sum of the probabilities associated to the states $\vectr z$ generating such a label. 
This further allows expressing the quantities in \eqref{eq:mean_var_Poi} as  
\begin{align*}
\mathbb E[x_i(t+1)|\vectr z(t)] = \sum_{r=1}^{n_p}P_{z_r(t)i} = \sum_{r=1}^{2n}x_r(t)P_{ri}. 
\end{align*}


The generalization of the previous results to vector labels leads to the following statement. 
\begin{theorem}
\label{thm:poiss_gen_mult}
The conditional random variables $(x_i(t+1)|\vectr x(t))$ are characterized by Poisson-binomial distributions, 
whereas the conditional random vector $(\vectr x(t+1)|\vectr x(t))$ by a generalized multinomial distribution. 
Their mean, variance, and covariance are described by 
\begin{align*}
\begin{array}{l}
\mathbb E[x_i(t+1)|\vectr x(t)] = \sum_{r=1}^{2n}x_r(t)P_{ri},\\[1ex]
var(x_i(t+1)|\vectr x(t)) =\sum_{r=1}^{2n} x_r(t)P_{ri}(1-P_{ri}),\\[1ex]
cov(x_i(t+1),x_j(t+1)|\vectr x(t)) = -\sum_{r=1}^{2n}x_r(t) P_{ri}P_{rj},
\end{array}
\end{align*}
for all $i,j\in \mathbb N_{2n},i\ne j$.
\end{theorem}

\smallskip

Theorem \ref{thm:poiss_gen_mult} indicates that the distribution of the conditional random variable $(\vectr x(t+1)|\vectr x(t))$ is independent of the underlying state $\vectr z(t)\rightarrow \vectr x(t)$ of $\varXi$. 
With focus on equation \eqref{tp4pb}, this result allows us to claim the following. 
\begin{corollary}
\label{cor:pbs}
The equivalence relation $\mathcal R$ is an exact probabilistic bisimulation over the Markov chain $\varXi$. 
The resulting quotient Markov chain is the coarsest probabilistic bisimulation of $\varXi$. 
\end{corollary}

Without loss of generality, 
let us normalize the values of the labels $\vectr x$ by the total population size $n_p$, thus obtaining a new variable $\vectr X$. 
The conditional variable $(\vectr X(t+1)|\vectr X(t))$ is characterized
by the following parameters, 
for all $i,j\in \mathbb N_{2n},i\ne j$:  
\begin{align}
\label{eq:cov_mat}
& \mathbb E[X_i(t+1)|\vectr X(t)] =\sum_{r=1}^{2n} X_r(t)P_{ri},\nonumber\\
& var(X_i(t+1)|\vectr X(t)) = \frac{1}{n_p}\sum_{r=1}^{2n}X_r(t)P_{ri}(1-P_{ri}),\\
& cov(X_i(t+1),X_j(t+1)|\vectr X(t)) = -\frac{1}{n_p}\sum_{r=1}^{2n}X_r(t) P_{ri}P_{rj}.\nonumber
\end{align}
Based on the expression of the first two moments of $(\vectr X(t+1)|\vectr X(t))$,
we apply a translation (shift) on this conditional random vector as 
\begin{align*}
\left\{
\begin{array}{l}
\omega_1(t) = X_1(t+1) - \sum_{r=1}^{2n}X_r(t)P_{r1}\\
\omega_2(t) = X_2(t+1) - \sum_{r=1}^{2n}X_r(t)P_{r2}\\
\vdots\\
\omega_{2n}(t) = X_{2n}(t+1) - \sum_{r=1}^{2n}X_r(t)P_{r 2n},
\end{array}
\right.
\end{align*}
where $\omega_i(t)$ are guaranteed to be (dependent) random variables with zero mean and covariance described by \eqref{eq:cov_mat}.
Such a translation allows expressing the following dynamical model for the variable $\vectr X$:  
\begin{align}
\label{eq:dyn_mod}
\vectr X(t+1) = P^T \vectr X(t)+ \vectr W(t),
\end{align}
where the distribution of $\vectr W(t)$ depends only on the state $\vectr X(t)$.
%

\begin{remark}\label{rem:conn}
We have modeled the evolution of the TCL population with the abstract model \eqref{eq:dyn_mod}, a linear stochastic difference equation.   
The dynamics in \eqref{eq:dyn_mod} represent a direct generalization of the model abstraction provided in \cite{HetTCL,HetTCLCallaway}, 
which is deterministic since its transitions are computed based on the trajectories of a deterministic version of \eqref{eq:tcl_dyn}. 
\end{remark}

\smallskip


In the following we study the limiting behavior when the TCL population size grows to infinity.
The first observation about the covariance in \eqref{eq:cov_mat} is that
the covariance matrix converges to zero as $n_p$ grows, 
with a rate of $1/n_p$.
This observation leads to the next result.
\begin{theorem}
\label{thm:shif_hevi}
The cumulative distribution function of the conditional random variable $(X_i(t+1)|\vectr X(t))$, $i\in\mathbb N_{2n}$,
converges to the shifted Heaviside step function pointed at $\mu = \sum_{r=1}^{2n} X_r(t)P_{ri}$, 
as the number of homogeneous Markov chains goes to infinity.
\end{theorem}

\smallskip

In other words, the random variables $(X_i(t+1)|\vectr X(t)), i = 1, \ldots, 2n$, 
converge in distribution to deterministic random variables.  
This result relates again the model in \eqref{eq:dyn_mod} to that in \cite{HetTCL,HetTCLCallaway}, as discussed in Remark \ref{rem:conn}. 

%
Above we have characterized the random variable $(X_i(t+1)|\vectr X(t))$ with a Poisson-binomial distribution. 
We use Lyapunov Central Limit Theorem
(cf. Lemma 7.1)
to show that this distribution converges to a Gaussian one.

\begin{theorem}
\label{thm:Gaus_appr}
The random variable $(X_i(t+1)|\vectr X(t))$ can be explicitly expressed as 
\begin{align*}
X_i(t+1) = \sum_{r=1}^{2n} X_r(t)P_{ri} + \omega_i(t),
\end{align*}
where the random vector $\vectr W(t) = [\omega_1(t),\cdots,\omega_{2n}(t)]^T$
has a covariance matrix $\varSigma(\vectr X(t))$ as in \eqref{eq:cov_mat},
and converges (in distribution) to a multivariate Gaussian random vector
$\mathcal{N}(0,\varSigma(\vectr X(t)))$,
as $n_p \uparrow \infty$. 
\end{theorem}

\medskip

Theorem \ref{thm:Gaus_appr} practically states that the conditional distribution of the random vector $\vectr W(t)$ 
for a relatively large population size 
can be effectively replaced by a multivariate Gaussian distribution with known moments.  
We shall exploit this result in the state estimation of the model using Kalman Filter in Section \ref{sec:controlled}.  
Notice that the above conclusion can be applied to any population of homogeneous Markov chains (TCL), 
as long as all TCL Markov chains have the same transition probability matrix. 
The initial distributions of the single Markov chain can instead be selected freely. 

In the previous theorem we have developed a linear model for the evolution of $X_i(t)$, 
which in the limit encompasses a Gaussian noise $\omega_i(t)$. 
As discussed in \eqref{eq:cov_mat}, 
these Gaussian random variables are not independent in general.  
The covariance matrix in \eqref{eq:cov_mat} is guaranteed to be positive semi-definite for all
$X_r\in \{0,\frac{1}{n_p},\frac{2}{n_p},\cdots,\frac{n_p-1}{n_p},1\}$, 
provided that $\sum_{r=1}^{2n}X_r = 1$. 
In view of a general use in \eqref{eq:dyn_mod}, 
we next show that the covariance matrix remains positive semi-definite when the model is extended over the variables $X_r \in [0,1]$.  

\begin{theorem}
\label{thm:posit_def}
Suppose we model the behavior of the population by the dynamical system \eqref{eq:dyn_mod}, where 
\begin{align}
\label{eq:dynamical_approx}
\mathbb E[\vectr W(t)] = 0,\quad cov(\vectr W(t)) = \varSigma(\vectr X(t)).
\end{align}
Then the covariance matrix $\varSigma(\vectr X)$ is positive semi-definite for all $X_r\ge 0$.
The entries of the random vector $\vectr W$ are dependent on each other, since 
$\sum_{r=1}^{2n}\omega_r = 0$ whenever $\sum_{r=1}^{2n}X_r = 1$. 
Finally, the condition $\sum_{r=1}^{2n}X_r(0) = 1$ implies that $\sum_{r=1}^{2n}X_r(t) = 1$, for all $t\in \mathbb N$.
\end{theorem}

\subsection{Explicit Quantification of the Errors of the Abstraction and of the Aggregation Procedures}  

Let us now quantify the power consumption of the aggregate model, 
as an extension of the quantity discussed after equation \eqref{eq:switch}.  
The total power consumption obtained from the aggregation of the original models in \eqref{eq:tcl_dyn}-\eqref{eq:switch}, 
with variables $(m_j, \theta_j)(t), j\in\mathbb N_{n_p},$ is 
\begin{equation}
\label{eq:tot_power}
y(t) = \sum_{j=1}^{n_p} m_j (t) P_{rate,ON}.
\end{equation}
With focus on the abstract model (with the normalized variable $\vectr X$), 
the power consumption is equal to 
\begin{equation*}
y_{abs}(t) = H\vectr X(t),\quad H = n_p P_{rate,ON}[0_n,\mathfrak 1_n],
\end{equation*}
where $0_n, \mathfrak 1_n$ are row vectors with all the entries equal to zero and one, respectively.

For the error quantification
we consider a homogeneous population of TCL with dynamics affected by Gaussian process noise $w(\cdot)\sim \mathcal N(0,\sigma^2)$, 
and the abstracted model constructed based on the partition introduced in \eqref{eq:partition}.
The result of this section hinges on two features of the Gaussian distribution, its continuity and its decay at infinity.
In order to keep the discussion focused we proceed considering Gaussian distributions, 
however the result can be extended to any distribution with these two features. 

Since the covariance matrix in \eqref{eq:cov_mat} is small for large population sizes, 
the first moment of the random variable $y(t)$ provides sufficient information on its behavior over a finite time horizon. 
The total power consumption in \eqref{eq:tot_power} is
the sum of $n_p$ independent Bernoulli trials over the sample space $\{0,P_{rate,ON}\}$, 
each with different success probability.   
Then for the quantification of the modeling error we study the error produced by the abstraction over the expected value of the TCL mode.

Consider a single TCL, with the initial state $s_0 = (m_0,\theta_0)$.
Also select the desired final time $T_d$ and time horizon $N = T_d/h$,
where $h$ is the discretization step.
The expected value of its mode at time $N$, $m(N)$, can be computed as
\begin{align}
\label{eq:reach_def}
\mathbb E[m(N)|m_0,\theta_0] & = \mathsf P\left(m(N) = 1|m_0,\theta_0 \right)
= \mathsf P\left(s(N-1)\in \mathcal A|m_0,\theta_0 \right),
\end{align}
where
$
\mathcal A = \{1\}\times[\theta_-,+\infty)\cup\{0\}\times[\theta_+,+\infty). 
$
%
This quantity can be characterized via value functions $\mathcal V_k:\mathcal S\rightarrow [0,1]$, $k\in\mathbb N_N$,
which are computed recursively as follows:
\begin{equation}
\label{eq:recur_prob}
\mathcal V_k(s_k) =
\int_{\mathcal S}\mathcal V_{k+1}(s_{k+1})t_s(s_{k+1}|s_k)ds_{k+1},
\quad \forall k\in\mathbb N_{N-1},
\quad\mathcal V_N(s) = \mathfrak 1_{\mathcal A}(s).
\end{equation}
Knowing these value functions, we have that $\mathbb E[m(N)|m_0,\theta_0] = \mathcal V_1(m_0,\theta_0)$.
Computationally, the calculation of these quantities can leverage the results in \cite{APKL10,SA11,SA13}, 
which however require extensions 1) to conditional density functions of the process that are discontinuous, 
and 2) to an unbounded state-space.
The first issue is addressed by the following theorem.
\begin{theorem}
\label{thm:reach_continuity}
The density function $t_s(s'|\cdot)$ is piecewise-continuous within the continuity regions
\begin{align*}
\{0\}\times\mathbb (-\infty,\theta_+],\quad
\{0\}\times\mathbb (\theta_+,+\infty),\quad
\{1\}\times\mathbb (-\infty,\theta_-),\quad
\{1\}\times\mathbb [\theta_-,+\infty).
\end{align*} 
The value functions $\mathcal V_k(s)$ are piecewise-Lipschitz continuous, namely:
\begin{equation*}
|\mathcal V_k(m,\theta)-\mathcal V_k(m,\theta')|\le \frac{2a}{\sigma\sqrt{2\pi}}|\theta-\theta'|,
\end{equation*} 
where $a,\sigma$ represent respectively the TCL parameter vector and the variance of the process noise, and
where $(m,\theta),(m,\theta')$ is any pair of points belonging to one of the four continuity regions of the density $t_s$. 
\end{theorem}

\smallskip

To cope with the second issue, we study the limiting behavior of the value functions at infinity.  
\begin{theorem}
\label{thm:reach_asymp}
The value function $\mathcal V_k(\cdot)$ has the following asymptotic properties:
\begin{align*}
& \lim_{\theta\rightarrow+\infty}\mathcal V_k(m,\theta) = 1,\quad
\lim_{\theta\rightarrow-\infty}\mathcal V_k(m,\theta) = 0,
\quad m\in\mathbb Z_1.
\end{align*}
\end{theorem}

Given the above asymptotic properties of the value functions, 
we leverage the truncation over the state-space proposed in Section \ref{sec:state_partition}, 
and properly select the value of the functions outside this region.  
The following theorem quantifies the error we incur with this state-space truncation.
\begin{theorem}
\label{thm:VF_limit}
For the partitioning procedure in \eqref{eq:partition} we have that 
\begin{align*}
& \mathcal V_k(m,\theta)\ge 1-(N-k)\epsilon,\quad \forall \theta\ge \theta_{\mathsf m} = \theta_s+\mathcal L/2,
\quad m\in\mathbb Z_1,\\
& \mathcal V_k(m,\theta)\le (N-k)\epsilon,\quad \forall \theta\le \theta_{-\mathsf m} = \theta_s-\mathcal L/2,
\quad m\in\mathbb Z_1,
\end{align*}
where $\epsilon = \dfrac{e^{-\gamma^2/2}}{\gamma\sqrt{2\pi}}$, and where 
\begin{equation*}
\gamma = \frac{1-a}{2\sigma}\left[\frac{a^N\mathcal L+\delta}{1-a^N}-\lambda\right],
\quad \lambda = RP_{rate}+|2(\theta_s-\theta_a)+RP_{rate}|.
\end{equation*}
\end{theorem}
Notice in particular that the previous theorem draws a linear dependence of $\gamma$ on $\mathcal L$.   

\begin{theorem}
\label{thm:error_1}
If we abstract a single TCL to a Markov chain based on the procedure of Section \ref{sec:state_partition},
and compute the solution of problem \eqref{eq:recur_prob} over the Markov chain -- call it $\mathcal W_1(m_0,\theta_0)$ --  
then the approximation error can be upper-bounded as follows: 
\begin{equation*}
|\mathcal V_{1}(m_0,\theta_0)-\mathcal W_1(m_0,\theta_0)|\le
(N-1)\left[\frac{N-2}{2}\epsilon + \frac{2a}{\sigma\sqrt{2\pi}}\upsilon\right],\quad
\forall (m_0,\theta_0)\in\mathbb Z_1\times[\theta_{-\mathsf m},\theta_{\mathsf m}].
\end{equation*}
The error has two terms: one term accounts for the error of the approximation over infinite-length intervals $\epsilon$, 
whereas the second is related to the choice of the partition size $\upsilon$. 
\end{theorem}

\smallskip

Collecting the results above, 
the following theorem quantifies the abstraction error over the total power consumption.
\begin{theorem}
\label{thm:error_hom}
The difference in the expected value of the total power consumption of the population $y(N)$,
and that of the abstracted model $y_{abs}(N)$, 
both conditional on the corresponding initial conditions, 
is upper bounded by
\begin{align}
\label{eq:error_hom}
\big|\mathbb E[ & y(N)|\vectr s_0] - \mathbb E[y_{abs}(N)|\vectr X_0]\big|
\le n_pP_{rate,ON}(N-1)\left[\frac{(N-2)}{2}\epsilon+\frac{2a}{\sigma\sqrt{2\pi}}\upsilon\right],
\end{align}
for all $\vectr s_0\in(\mathbb Z_1\times[\theta_{-\mathsf m},\theta_{\mathsf m}])^{n_p}$.
The initial state $\vectr X_0$ is a function of the initial states in the TCL population $\vectr s_0$, 
as from the definition of the state vector $\vectr X$.
\end{theorem}

\smallskip

Notice that this result allows tuning the error in the total power consumption of the population estimated from the abstraction -- effectively reducing it to a desired level by increasing the abstraction precision. 

\begin{remark}
\label{rem:local_error}
The above upper bound on the error can be tightened by local computation of the errors, 
as suggested in \cite{SA11,SA13}.  
Suppose $\vectr E_k$ is a $2n\times 1$ vector where each element specifies the error in the related partition set.
It is possible to derive the following recursion for the local error:
$\vectr E_k = \vectr E + P \vectr E_{k+1}$, 
where $\vectr E_N = 0_{2n}^T$,
and where $\vectr E$ is a constant vector with elements equal to $\frac{2a\upsilon}{\sigma\sqrt{2\pi}}$, 
except those related to the absorbing states, which are equal to $\epsilon$.  
Then the elements of $\vectr E_1$ are upper-bounds for the quantity $|\mathcal V_{1}(m_0,\theta_0)-\mathcal W_1(m_0,\theta_0)|$ in each partition set.
%
Moreover, it is possible to reduce the upper bound \eqref{eq:error_hom} to the quantity $n_p\vectr E_1^T \vectr X_0$. 
\end{remark}

\subsection{Further State-Space Reduction of the Asymptotic Model}
The result in Theorem \ref{thm:error_hom} suggests that 
in order to decrease the abstraction error we have to decrease the size of the partitioning bins, which consequently leads to an increase on their number and to a large-dimensional linear model in \eqref{eq:dynamical_approx}. 
This section discusses how to mitigate this issue by means of application of model-order reduction techniques over the large dimensional linear model obtained by the abstraction. 
This known technique follows the observation that the dynamics of the linear model are mostly determined by the largest eigenvalues of the transition probability matrix. 
The following statement helps reframing the linear model within the framework of model-order reduction by eliminating the dependency of state variables in \eqref{eq:dynamical_approx}.

\begin{theorem}
\label{thm:elim_redun}
The dynamical system in \eqref{eq:dynamical_approx} can be modeled by the following stable input/output model 
\begin{align*}
& \bar X(t+1) = A \bar X(t) + B u(t) + \bar W(t)\\
& y_{red}(t) = C \bar X(t) + D u(t),
\end{align*}
where the state vector is $\bar X = [X_1,\cdots,X_{2n-1}]^T$ and the input is taken as the step function. 
The process noise $\bar W(t)$ contains the first $(2n-1)$ elements of $W(t)$.
Suppose we partition the transition matrix:
\begin{align*}
P = \left[
\begin{array}{cc}
\Omega_{11} & \Omega_{12}\\
\Omega_{21} & \Omega_{22}
\end{array}
\right],
\end{align*}
where
$\Omega_{11}\in \mathbb R^{(2n-1)\times(2n-1)},\Omega_{12}\in \mathbb R^{(2n-1)\times 1},\Omega_{21}\in \mathbb R^{1\times(2n-1)}$, and $\Omega_{22}\in \mathbb R$.
Then $A^T = \Omega_{11}-\mathfrak 1_{2n-1}^T\Omega_{21}$,
$B^T = \Omega_{21}$,
$C = [-\mathfrak 1_{n}, 0]$,
and $D = 1$.
Finally, $\lambda(P) = \lambda(A)\cup\{1\}$.
\end{theorem}

\smallskip

This Theorem allows using model-order reduction techniques like \emph{balanced realization and truncation}
or \emph{Hankel singular values} \cite{A05} to obtain a low dimensional model describing the dynamics of the population power consumption.
Notice that one major difference between the reduced-order model and the model of \cite{HetTCL, HetTCLCallaway} is that here the state matrix $A$ is no longer a transition probability matrix. 
Theorem \ref{thm:elim_redun} and the related model reduction technique are not exclusively applicable to homogeneous populations of TCL, 
but can as well be employed for the heterogeneous populations discussed in the following section.    
The technique is applied on the case study in Section \ref{sec:benchmarks}. 

\section{Formal Abstraction of a Heterogeneous Population of TCL}
\label{sec:exten_heter}
Consider a heterogeneous population of $n_p$ TCL, 
where heterogeneity is characterized by a parameter $\alpha$ that takes $n_p$ values in $\{\alpha_1,\alpha_2,...,\alpha_{n_p}\}$.
Each instance of $\alpha$ specifies a set of model parameters
$(\theta_s,\delta,\theta_a,C,R,\sigma,P_{rate},P_{rate,ON})$ for the dynamics of a single TCL.
Notice that all the parameters in the set influence the temperature evolution, 
except $P_{rate,ON}$, which affects exclusively the output equation. 
Each dynamical model can be abstracted as a Markov chain $\mathcal{M}_\alpha$ with a transition matrix $P_\alpha = [P_{ij}(\alpha)]_{i,j}$, 
according to the procedure in Section \ref{sec:formal_abst}.    
As expected, the transition probability matrix $P_\alpha$ obtained for a TCL depends on its own set of parameters $\alpha$. 

With focus on an aggregated Markov chain model for a population of $n_p$ TCL, 
the goal is again that of abstracting it as a reduced-order (lumped) model.  
The apparent difficulty is that the heterogeneity in the transition probability matrix $P_\alpha$ of the single TCL renders the quantity
$\mathsf P(x_i(t+1) = j|\vectr z(t))$ 
dependent not only on the label $\vectr x(t) = L(\vectr z(t))$,
but effectively on the current state $\vectr z(t)$, 
namely the present distribution of temperatures of each TCL. 
This leads to the impossibility to simplify equation \eqref{eq:total_prob}, 
as done in the homogeneous case. 
Recall that computations on $\mathsf P(\vectr z(t))$ require manipulations over the large dimensional matrix $P_\Xi$, 
which can become practically infeasible.  

In contrast to the homogeneous case, 
which allows us to quantify the probabilities $\mathsf P(x_i(t+1) = j|\vectr x(t))$ 
over a Markov chain obtained as an exact probabilistic bisimulation of the product chain $\varXi$, 
in the heterogeneous case we resort to an \emph{approximate} probabilistic bisimulation \cite{DLT08} of the Markov chain $\varXi$.
The approximation enters in equation \eqref{eq:total_prob} with the replacement of the weighted average in the expression of the law of total probability 
with a normalized (equally weighted) average, as follows: 
\begin{align}
\label{eq:total_prob2} 
\mathsf P(x_i(t+1) = j| \vectr x(t)) = \frac{\sum_{\vectr z(t)\rightarrow\vectr x(t)}
\mathsf P(x_i(t+1) = j|\vectr z(t))}{\#\left\{\vectr z(t)\rightarrow\vectr x(t)\right\}}.
\end{align}
In other words we have assumed that the probability for the Markov chain $\varXi$ to be in each labeled state is the same.  
Similarly, 
the average of the random variables $x_i(t+1)$, conditioned over $\vectr x(t)$, 
can be obtained from \eqref{eq:total_prob2} as
$
\mathbb E[x_i(t+1)|\vectr x(t)]
= \frac{\sum_{\vectr z(t)\rightarrow\vectr x(t)}\mathbb E[x_i(t+1)|\vectr z(t)]}
{\#\left\{\vectr z(t)\rightarrow\vectr x(t)\right\}}. 
$
Unlike in the exact bisimulation instance, 
the error introduced by the approximate probabilistic bisimulation relation can only be quantified empirically over matrix $P_\Xi$.   

Next, we put forward two alternative approaches to characterize the properties of the abstraction of the TCL population: 
by an averaging argument  in Section \ref{sec:Het_avg}, 
and by a clustering assumption in Section \ref{sec:Het:cluster}. 

\subsection{Abstraction of a Heterogeneous Population of TCL via Averaging}
\label{sec:Het_avg}
We characterize quantitatively the population heterogeneity by constructing 
the empirical density function $f_{\alpha}(\cdot)$ from the finite set of values taken by the parameter $\alpha$. 
This allows the characterization of the statistics of the conditional variable $(\vectr X(t+1)|\vectr X(t))$ 
(recall that $\vectr X$ is a normalized version of $\vectr x$) as follows. 
\begin{theorem}
\label{thm:hetr_mean}
Consider a TCL population with heterogeneity that is encompassed by a parameter $\alpha$ with empirical density function $f_{\alpha}(\cdot)$. 
Introducing an approximate probabilistic bisimulation of the Markov chain $\varXi$ as in \eqref{eq:total_prob2},   
the conditional random variable $(\vectr X(t+1)|\vectr X(t))$ has the following statistics: 
\begin{align*}
\mathbb E[X&_i(t+1)|\vectr X(t)] =\sum_{r=1}^{2n} X_r(t)\overline{P_{ri}},\\
var( & X_i(t+1)|\vectr X(t))
= \frac{1}{n_p} \sum\limits_{r=1}^{2n}X_r\overline{P_{ri}(1-P_{ri})}
+\frac{1}{n_p-1} \left(\sum\limits_{r=1}^{2n}X_r\overline{P_{ri}}\right)^2
-\frac{1}{n_p-1} \sum\limits_{r=1}^{2n}X_r \overline{P_{ri}}^2,\\
cov( & X_i(t+1),X_j(t+1) |\vectr X(t))
= \frac{1}{n_p-1}\left(\sum\limits_{r=1}^{2n}X_r \overline{P_{r i}}\right)\left(\sum\limits_{s=1}^{2n}X_s \overline{P_{s j}}\right)\\
& \hspace{1.7in} -\frac{1}{n_p-1}\sum\limits_{r=1}^{2n}X_r \overline{P_{r i}}\overline{P_{r j}}-
\frac{1}{n_p}\sum_{r=1}^{2n}X_r\overline{P_{r i}P_{r j}},
\end{align*}
where the barred quantities indicate an expected value respect to the parameters set $\alpha$, for instance 
$\overline{P_{r i}P_{r j}} =\mathbb E_\alpha[P_{r i}(\alpha)P_{r j}(\alpha)] = \int P_{r i}(v)P_{r j}(v)f_{\alpha}(v)dv$.
\end{theorem}

Further, let us mention that the asymptotic properties obtained as the population size grows, 
as discussed in Section \ref{sec:stoch_proper}, still hold as long as the distribution of the parameters set $f_\alpha(\cdot)$ is given and fixed.  
%

With focus on the heterogeneity in the output equation, 
we can similarly replace the ensemble of parameter instances $P_{rate,ON}$ by the average quantity $\bar P_{rate,ON}$, 
namely the mean rated power of the TCL population in the ON mode, 
which is computed as the expected value of $P_{rate,ON}$ with respect to the parameter set:
$\bar P_{rate,ON} = \mathbb E_{\alpha}\left[ P_{rate,ON}(\alpha)\right]$.
While (as discussed above) we cannot analytically quantify the error introduced by the approximate bisimulation used for the abstraction of the temperature evolution in the population, 
we can still quantify the error related to the heterogeneity in the output equation: 
this will be done shortly in Theorem \ref{thm:cluster_error}.

\subsection{Abstraction of a Heterogeneous Population of TCL via Clustering}
\label{sec:Het:cluster}
We propose an alternative method to reduce a heterogeneous population of TCL into a finite number of homogeneous populations. 
While more elaborate than the preceding approach, it allows for the quantification of the error under the following Assumption.
\begin{assumption}
Assume that the heterogeneity parameter $\alpha = (\theta_s,\delta,\theta_a,C,R,\sigma,P_{rate},P_{rate,ON})$ belongs to a bounded set $\Gamma_a$, 
and that the parametrized transition probability matrix $P_\alpha$ satisfies the following inequality expressing a condition on its continuity w.r.t. $\alpha$:
\begin{align}
\label{eq:Mat_pert}
\|P_\alpha-P_{\alpha'}\|_\infty\le h_a\|\alpha-\alpha'\|\quad\forall \alpha,\alpha'\in\Gamma_a.
\end{align}
\end{assumption}
Consider an heterogeneous range for a given parameter: the approach is to partition the uncertainty range and ``cluster together'' the TCL in the given population, 
according to the partition they belong to, and further considering them as homogeneous within their cluster.   
More precisely, 
Select a finite partition of the set $\Gamma_a = \cup_i\varGamma_i$, characterized by a diameter $\upsilon_a$, namely 
$
\|\alpha-\alpha'\|\le \upsilon_a, \forall \alpha,\alpha'\in\varGamma_i,\forall i. 
$
Associate arbitrary representative points $\alpha_i\in\varGamma_i$ to the partition sets.
Finally, replace the transition matrix $P_\alpha$ and $P_{rate,ON}$ by $\sum_iP_{\alpha_i}\mathbb{I}_{\varGamma_i}(\alpha)$ and $\sum_i P_{rate,ON}(\alpha_i)\mathbb{I}_{\varGamma_i}(\alpha)$, respectively. 
The error made by this procedure is quantified in the following statement. 

\begin{theorem}
\label{thm:cluster_error}
Given a heterogeneous population of TCL, 
suppose we cluster the heterogeneity parameter $\alpha\in\varGamma_i$, 
assume homogeneity within the introduced clusters, 
and model each cluster based on the results of Section \ref{sec:formal_abst} with outputs $y_{abs,i}(N)$.
Let us define the approximate power consumption of the heterogeneous population as the sum of clusters outputs, as follows: 
$y_{abs}(N) = \sum_{i}y_{abs,i}(N)$.
The abstraction error can be upper-bounded by
\begin{align}
\big|\mathbb E[y(N)|\vectr s_0] - \mathbb E[y_{abs}(N)]\big|
& \le \max_{\alpha} n_p(N-1)P_{rate,ON}(\alpha)\left[\frac{(N-2)}{2}\epsilon(\alpha)+\frac{2a(\alpha)}{\sigma(\alpha)\sqrt{2\pi}}\upsilon\right]\nonumber\\
& + n_p\left[ \bar P_{rate,ON}(N-1)h_a+1\right]\upsilon_a,\label{eq:cluster_error}
\end{align}
for all $\vectr s_0\in(\mathbb Z_1\times[\theta_{-\mathsf m},\theta_{\mathsf m}])^{n_p}$. 
The parameters $\epsilon(\cdot),\gamma(\cdot)$, and $\lambda(\cdot)$ are computed as in Theorem \ref{thm:error_hom} and depend on the value of $\alpha$. 
Finally, let us introduce the quantity $\bar P_{rate,ON} = \sum_i\frac{n_i}{n_p}P_{rate,ON}(\alpha_i) = \mathbb E_\alpha\left[P_{rate,ON}(\alpha) \right]$,
where $n_i$ is the population size of the $i^{th}$ cluster, so that $\sum_{i}n_i = n_p$.
\end{theorem}

\smallskip

Notice that the first part of the error in \eqref{eq:cluster_error} is due to the abstraction of a single TCL by state-space partitioning,
while the second part is related to the clustering procedure described above. 
Further, notice that all terms in the bound above can be reduced by selecting finer temperature partitions (smaller bins) or smaller clusters diameter for the parameter sets.

The second part of the error in \eqref{eq:cluster_error} is computed based on the Lipschitz continuity of the transition probability matrix $P_\alpha$ as per Assumption \ref{eq:Mat_pert}. 
This can be evaluated over the transition probability matrices obtained by abstracting the heterogeneous TCL dynamics (characterized by the conditional density functions $t_s$) as Markov chains.  
Alternatively, 
we could formulate this error bound based on Lipschitz continuity of the conditional density function $t_s$ with respect to the parameters set $\alpha$ by using the explicit relation \eqref{eq:trans_prob} for the transition probabilities.
Then the constant $h_a$ is computable as a function of the Lipschitz constant of the conditional density function of the process. 
As an example, the constant $h_a$ for the case of a Gaussian process noise and heterogeneity term residing exclusively in thermal capacitance (that is, in the parameter $a$) is computed as follows: $h_a = \dfrac{\mathcal L+\lambda}{\sigma\sqrt{2\pi}}$.  

As a final note, the result in Theorem \ref{thm:cluster_error} is applicable to the setup in Section \ref{sec:Het_avg} when the heterogeneity lies in the parameter $P_{rate,ON}$, 
by considering a single cluster. 
 
\section{Abstraction and Control of a Population of Non-Autonomous TCL}
\label{sec:controlled} 

One can imagine a number of different strategies for controlling the total power consumption of a population of TCL. 
With focus on the dynamics of a single TCL, 
one strategy could be to vary the rate of the energy transfer $P_{rate}$, for instance by circulating cold/hot water through the load with higher or lower speed.
Another approach could be to act on the thermal resistance $R$, for instance opening or closing doors and windows at the load.
Yet another strategy could be to apply changes to the set-point $\theta_s$, as suggested in \cite{HomTCL}.  

Let us observe that the first two actions would modify the dynamics of \eqref{eq:tcl_dyn}, 
whereas the third control action would affect the relation in \eqref{eq:switch}. 
Upon abstracting the TCL model as a finite-state Markov chain, 
a control action would result in a modification of the elements of the transition probability matrix.
With reference to \eqref{eq:structure_P}, 
the entries of the matrices $Q_{11},Q_{22},Q_{31},Q_{42}$ are computed based on \eqref{eq:tcl_dyn},
while the size of these matrices are determined based on \eqref{eq:switch}.
Since the set-point $\theta_s$ affects only equation \eqref{eq:switch},
the set-point regulation changes the structure of the probability matrix in \eqref{eq:structure_P}
while other approaches affect the value of its non-zero elements. 
It follows that the set-point regulation has the advantage of a single computation of marginals,
while the other discussed methods would require such a computation as a function of the allowed control inputs.  
%

In order to focus the discussion, 
we consider more challenging case where the control input is taken to be the set-point $\theta_s$ of the TCL. 
We intend to apply the control input to all TCL uniformly,
(cf. Figure \ref{fig:control_MPC}) 
in order to retain the population homogeneity and since this does not require to differentiate among the states of different TCL.  
This is unlike in \cite{HetTCL},
which consider the control signal as an external input that is applied based on the knowledge of states of single TCL, 
and which practically requires adding thermometers (with relatively high accuracy) to each TCL. 
More precisely, 
\cite{HetTCL} assumes full knowledge of the state vector $\vectr X(t)$ and employs a Model Predictive Control architecture to design the control signal. 
Moving forward, 
\cite{HetTCLCallaway} considers different options for the configuration of the closed loop control:  
either states are completely measured, 
or a portion of state information is required,
or else states are estimated by using an Extended Kalman Filter (EKF). 
The minimum required infrastructure in \cite{HetTCLCallaway} ranges from a 
TCL temperature sensor and a two-way data connection for transmitting the state information and control signal, 
to a one-way data connection for sending the control signal to the TCL. 
The presence of a local decision maker is essential in all the scenarios: 
each TCL receives a control signal at each time step, 
determines its current state, and generates a local control action (for instance, by randomization).
The performance of EKF for 
state estimation seems to be bound to computational limitations when the number of states becomes large.  
%

In the following we attempt to mitigate the above limitations by showing that the knowledge of the actual values of the TCL states or of vector $\vectr X(t)$ in the aggregated model are not necessary.
Given the model parameters, all is needed is an online measurement of the total power consumption of the TCL population, 
which allows estimating the states in $\vectr X(t)$ and using the set-point $\theta_s$ to track a given reference signal.
The control action comprises a simple signal for the set-point that is applied to all TCL uniformly: 
no local decision maker is required.  

\subsection{State Estimation and One-Step Regulation}
\label{sec:est_Regul}
Suppose we have a homogeneous population of TCL with known parameters.
We assume that the control input is discrete and take values from a finite set, 
$ \theta_s(t) \in \{\theta_{-\mathsf l},\theta_{-\mathsf l+1},\cdots,\theta_{\mathsf l-1},\theta_{\mathsf l}\}, \forall t \in \mathbb Z$. 
The parameter $\mathsf l$ is arbitrary and has here been chosen to align with the abstraction parameter in Figure \ref{fig:part} 
and with the scheme in \eqref{eq:partition}.  
Based on \eqref{eq:dyn_mod},
we set up the following discrete-time stochastic switched system: 
\begin{align*}
\vectr X(t+1) = F_{\sigma(t)}\vectr X(t)+\vectr W(t), 
\end{align*}
where $\forall t \in \mathbb Z$ the state matrix 
$F_{\sigma(t)}\in\{P^T(\theta_{-\mathsf l}),P^T(\theta_{-\mathsf l+1}),\cdots,P^T(\theta_{\mathsf l-1}),P^T(\theta_{\mathsf l})\}$ (cf. \eqref{eq:dyn_mod}),
and the switching signal $\sigma(\cdot):\mathbb Z\rightarrow \mathbb Z_{2\mathsf l}$ is a map specifying the set-point $\theta_s$, 
and hence the TCL dynamics, as a function of time. 
The process noise $\vectr W(t)$ is normal with zero mean and a state-dependent covariance matrix $\varSigma(\vectr X(t))$ in \eqref{eq:cov_mat}. 
The total power consumption of the TCL population is measured as 
\begin{equation*}
y_{meas} (t) = H\vectr X(t)+ v(t),
\end{equation*}
where $v(t)\sim\mathcal N(0,R_v)$ is a measurement noise characterized by $\sqrt{R_v}$, 
the standard deviation of the real-time measurement in the power meter instrument. 

Since the process noise $\vectr W$ is state-dependent, 
the state of the system can be estimated by modifying the classical Kalman Filter with the following time update: 
\begin{align*}
& \hat{\vectr X}^{-}(t+1) = F_{\sigma(t)}\hat{\vectr X}(t),\\
& P^-(t+1) = F_{\sigma(t)}\mathds P(t) F_{\sigma(t)}^T+\varSigma(\hat{\vectr X}(t)),
\end{align*}
and the following measurement update: 
\begin{align*}
& K_{t+1} = P^-(t+1)H^T\left[H P^-(t+1)H^T+R_v \right]^{-1},\\
& \mathds P(t+1) = [I-K_{t+1}H]P^-(t+1),\\
& \hat{\vectr X}(t+1) = \hat{\vectr X}^{-}(t+1)+K_{t+1}[y_{meas}(t+1)-H \hat{\vectr X}^{-}(t+1)].
\end{align*}

When the state estimates $\hat{\vectr X}$ are available,
we formulate the following optimization problem based on a one-step output prediction, 
in order to synthesize the control input at the next step:  
\begin{align*}
& \min_{\sigma(t+1)\in\mathbb Z_{2\mathsf l}} |y_{est}(t+2)-y_{des}(t+2)|, \text{   s.t.}\\
& \hat{\vectr X}(t+2) = F_{\sigma(t+1)}\hat{\vectr X}(t+1)\\
& y_{est}(t+2) = H\hat{\vectr X}(t+2),
\end{align*}
where $y_{des}(\cdot)$ is a desired reference signal and $\hat{\vectr X}(t+1)$ is provided by the Kalman Filter above.   
The obtained optimal value for $\sigma(t+1)$ provides the set-point $\theta_s(t+1)$, 
which is applied to the entire TCL population at the following $(t+1)^{\text{th}}$ iteration. 
Figure \ref{fig:control_MPC} illustrates the closed-loop configuration of the above scheme for state estimation and one-step regulation of the power consumption.
For clarity we have summarized the interpretation of the different notations used for power consumption in Table \ref{tab:out}.  

\begin{figure}
\scalebox{0.7}
{
\begin{pspicture}(-0.5,-4.5)(21.5,4.5)
\psframe[linewidth=0.04,dimen=outer](5.8,3.3)(2,0.7)
\psframe[linewidth=0.04,dimen=outer](18.5,4.2)(8.6,-0.3)
\psframe[linewidth=0.04,dimen=outer](17.1,-1.3)(10.2,-4)
\psframe[linewidth=0.04,dimen=outer](4.6,-2)(3.2,-3)

\psline[linewidth=0.04cm,arrowsize=0.08cm 2.0,arrowlength=1.4,arrowinset=0.4]{->}(5.75,2)(8.6,2)
\psline[linewidth=0.04cm,arrowsize=0.05291667cm 2.0,arrowlength=1.4,arrowinset=0.4]{->}(10.18,-2.5)(4.6,-2.5)

\psline[linewidth=0.04,arrowsize=0.05291667cm 2.0,arrowlength=1.4,arrowinset=0.4]{->}(18.5,2)(21,2)(21,-2.5)(17.1,-2.5)
\psline[linewidth=0.04,arrowsize=0.05291667cm 2.0,arrowlength=1.4,arrowinset=0.4]{->}(3.2,-2.5)(0,-2.5)(0,2)(2,2)

\rput(3.7,3.6){Population}
\rput(13.4,4.5){Kalman Filter with state dependent process noise}
\rput(13.5,-1){One-step regulation}
\rput(3.95,-2.5){$Z^{-1}$}
\rput(19.8,2.3){$\hat{\vectr X}(t+1)$}
\rput(0.9,2.3){$\theta_s(t)$}
\rput(8,-2.2){$\theta_s(t+1)$}

\psline[linewidth=0.04cm,arrowsize=0.05291667cm 2.0,arrowlength=1.4,arrowinset=0.4]{->}(7.7,0.185)(8.6,0.185)
\psline[linewidth=0.04cm,arrowsize=0.05291667cm 2.0,arrowlength=1.4,arrowinset=0.4]{->}(18,-3.1)(17.1,-3.1)

\rput(7.3,2.3){$y_{meas}(t+1)$}
\rput(19,-3){$y_{des}(t+2)$}
\rput(7.1,1.4){$\hat{\vectr X}(0)$}
\rput(7.1,0.8){$R_v$}
\rput(6.7,0.3){population}
\rput(6.7,-0.1){parameters}

\psline[linewidth=0.04cm,arrowsize=0.05291667cm 2.0,arrowlength=1.4,arrowinset=0.4]{->}(7.7,0.765)(8.6,0.765)
\psline[linewidth=0.04cm,arrowsize=0.05291667cm 2.0,arrowlength=1.4,arrowinset=0.4]{->}(7.7,1.325)(8.6,1.325)
\psline[linewidth=0.04cm,arrowsize=0.05291667cm 2.0,arrowlength=1.4,arrowinset=0.4]{->}(18,-3.7)(17.1,-3.7)

\rput(3.8,2.8){$TCL_1(\theta_1,m_1)$}
\rput(3.8,2.3){$TCL_2(\theta_2,m_2)$}
\rput(3.4,1.95){$\vdots$}
\rput(4,1.3){$TCL_{n_p}(\theta_{n_p},m_{n_p})$}

\rput(19,-3.6){population}
\rput(19,-4){parameters}

\rput(9.9,3.7){Time update:}
\rput(11.3,3.2){$\hat{\vectr X}^{-}(t+1) = F_{\sigma(t)}\hat{\vectr X}(t)$}
\rput(12.6,2.6){$P^-(t+1) = F_{\sigma(t)}\mathds P(t) F_{\sigma(t)}^T+\varSigma(\hat{\vectr X}(t))$}

\rput(10.55,2){Measurement update:}
\rput(12.7,1.5){$K_{t+1} = P^-(t+1)H^T\left[H P^-(t+1)H^T+R_v \right]^{-1}$}
\rput(11.6,0.9){$\mathds P(t+1) = [I-K_{t+1}H]P^-(t+1)$}
\rput(13.6,0.3){$\hat{\vectr X}(t+1) = \hat{\vectr X}^{-}(t+1)+K_{t+1}[y_{meas}(t+1)-H \hat{\vectr X}^{-}(t+1)]$}

\rput(13.7,-1.8){$ \min_{\sigma(t+1)\in\mathbb Z_{2\mathsf l}} |y_{est}(t+2)-y_{des}(t+2)|$}
\rput(11.5,-2.4){subject to:}
\rput(13.3,-3.1){$\hat{\vectr X}(t+2) = F_{\sigma(t+1)}\hat{\vectr X}(t+1)$}
\rput(13.05,-3.6){$y_{est}(t+2) = H\hat{\vectr X}(t+2)$}
\end{pspicture}
}
\caption{State estimation and one-step regulation for the closed-loop control of the power consumption.}
\label{fig:control_MPC}
\end{figure}
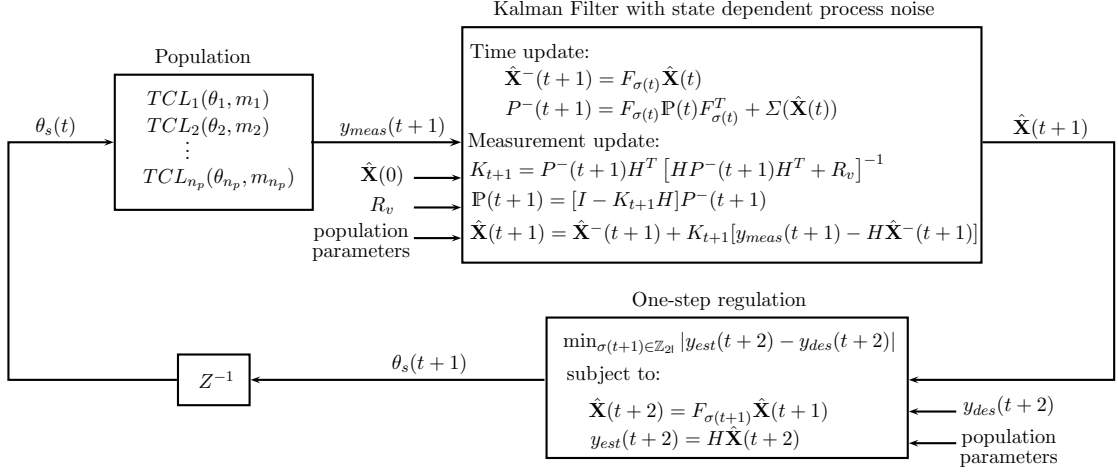

\begin{table}
\centering
\caption{Notations introduced for the output signals (total power consumption of TCL population)}
\scalebox{1}{ %
\begin{tabular}{| l l |}
\hline
signal & Interpretation\\
\hline
$y(t)$ & actual power consumption of the TCL population\\
$y_{meas}(t)$ & measured power consumption of the population, input to the KF \\
$y_{est}(t)$ & estimated power consumption, output of KF\\
$y_{des}(t)$ &  desired power consumption, given reference signal for power tracking\\
$y_{abs}(t)$ & output of the linear stochastic model, abstraction of TCL population\\
$y_{red}(t)$ & output of reduced-order model of abstraction of TCL population\\
\hline
\end{tabular}
}
\label{tab:out}
\end{table}

\subsection{Regulation via Stochastic Model Predictive Control (SMPC)}
\label{sec:SMPC}
We can perform power tracking by formulating and solving the following SMPC problem \cite{HCL09}:
\begin{align*}
& \min_{\sigma(\tau)} J_t = \mathbb E\left[\sum_{\tau=t+1}^{T}[y_{abs}(\tau)-y_{des}(\tau)]^2+\kappa^T\vectr X(T)\bigg\vert \vectr X(t)\right], \text{   s.t.}\\
&\vectr X(\tau+1) = F_{\sigma(\tau)}\vectr X(\tau)+\vectr W(\tau),\quad
y_{abs}(\tau) = H\vectr X(\tau),\\
& \sigma(\tau)\in \mathbb Z_{2\mathsf l},\quad\forall \tau\in\{t,t+1,\cdots,T-1\}.
\end{align*}

The cost function comprises a running cost for tracking and a terminal cost. 
The terminal cost is assumed to be a linear combination (with weighting vector $\kappa$) of the model states at final time $T$,
and practically accounts for possible penalty weights over the number of TCL within the temperature intervals. 
The expectation is taken over the underlying probability space for the trajectories of the process over the time interval $[t+1,T]$. 

The dynamics are nonlinear due to the switching nature of the control signal. 
The average evolution of the states and output of the system can be expressed by the following deterministic difference equation:  
\begin{align*}
\mathbb E[\vectr X(\tau+1)] = F_{\sigma(\tau)}\mathbb E[\vectr X(\tau)],\quad
\mathbb E[y_{abs}(\tau)] = H\mathbb E[\vectr X(\tau)].
\end{align*}
The associated state transition matrix $\Phi_{\sigma}(T,t) = F_{\sigma(T-1)}F_{\sigma(T-2)}\cdots F_{\sigma(t)}$
provides a closed form for the average evolution over the interval $[t,T]$:
\begin{align*}
\mathbb E[\vectr X(T)] = \Phi_{\sigma}(T,t)\mathbb E[\vectr X(t)],\quad
\mathbb E[y_{abs}(T)] = H\Phi_{\sigma}(T,t)\mathbb E[\vectr X(t)].
\end{align*}
%
Thanks to the linearly state-dependent covariance matrix, 
we can establish the following result. 
\begin{theorem}
\label{thm:cost_function}
The cost function of the SMPC problem can be computed explicitly as
\begin{align}
\label{eq:cost_function}
J_t = \sum_{\tau=t+1}^{T}\left[H\Phi_{\sigma}(\tau,t)\vectr X(t)-y_{des}(\tau)\right]^2 + \Psi_{\sigma}(T,t)\vectr X(t),
\end{align}
where the matrix
\begin{align*}
\Psi_{\sigma}(T,t) = \kappa^T\Phi_{\sigma}(T,t)+\frac{1}{n_p}\sum_{\tau_1=t}^{T}\sum_{\tau_2=\tau_1+1}^{T}\mathscr R(H\Phi_{\sigma}(\tau_2,\tau_1+1),F_{\sigma(\tau_1)})\Phi_{\sigma}(\tau_1,t), 
\end{align*}
and where 
$\mathscr R:\mathbb R^{1\times 2n}\times\mathbb R^{2n\times 2n}\rightarrow \mathbb R^{1\times 2n}$
is a matrix-valued map with
$\mathscr R(C,D) = C^{\circ 2}D - (C D)^{\circ 2}$,
where the operator $\circ 2$ is the \emph{Hadamard} square of the matrix (element-wise square).
\end{theorem}

\smallskip

The obtained explicit cost function is the sum of
a quadratic cost for the deterministic average evolution of the system state 
and of a linear cost related to the covariance of the process noise.  
\begin{example}
The SMPC formulation can accommodate problems where the population participates in the energy and ancillary services market to minimize its own energy costs.
In the real-time energy market the Locational Marginal Pricing algorithms result in the profile of energy price for time intervals of 5-minutes \cite{WMBG12}. 
Given that profile,
the population can save money by minimizing the total cost of its energy usage within the given time frame,
i.e. consuming less energy when the price is high and more energy when the price is low, under some constraints,
in the next 24 hours.
Suppose the final time $T$ is selected such that $T = 24/h$, where $h$ is the length of the sampling time ($5$ minutes),
and let the sequence $\{\lambda_\tau,\tau=t+1,t+2,\cdots,T\}$ be the profile of the energy price provided by the energy market.
The total energy consumption of the population is then $\sum_{\tau=t+1}^{T}\lambda_\tau y_{abs}(\tau)h$.
The following optimization problem can be solved, 
given the model dynamics, 
in order to minimize the expected value of the energy consumption as follows:  
\begin{align*}
\min_{\sigma(\tau)} \mathbb E\left[\sum_{\tau=t+1}^{T}\lambda_\tau y_{abs}(\tau)h\bigg\vert \vectr X(t)\right]
 = \min_{\sigma(\tau)}\sum_{\tau=t+1}^{T}\lambda_\tau hH\Phi_{\sigma}(\tau,t)\vectr X(t).
\end{align*}
\end{example}
\begin{remark}
Notice that for both formulations of the power tracking problem, the reference signal $y_{des}(\cdot)$ is assumed to be given.   
This can be in practice obtained when the TCL population is controlled by the energy market:  
one can think that a power utility company observes the power demand of the network and generates a predicted reference signal for the population, 
in order to obtain a total flat power consumption of the network together with the population.
\end{remark}
%

%

\section{Numerical Case Study and Benchmarks}
\label{sec:benchmarks}
In this section we compare the performance of our abstraction with that developed in \cite{HetTCL}, 
which as discussed obtains an aggregated model with dynamics that are deterministic, 
and in fact shown to be a special (limiting) case of the model in this work (cf. Remark \ref{rem:conn} and Theorem \ref{thm:shif_hevi}). 
We further elucidate the extension to the case of heterogeneous populations (with a comparison of the two proposed approaches), 
and the application of model-order reduction to the aggregated model. 
Finally, we synthesize global controls over the temperature set-point to perform tracking of a the total power consumption of the population. 

For all simulations we consider a population size of $n_p = 500$ TCL,    
however recall that our abstraction is proved to work as desired for any value $n_p$ of the population size. 
We have run $50$ Monte Carlo simulations for the TCL population based on the dynamics in \eqref{eq:tcl_dyn}-\eqref{eq:switch} aggregated explicitly,  
and computed the average total power consumption. 

\subsection{Aggregation of an Homogeneous Population of TCL}
Each TCL is characterized by parameters that take value in Table \ref{tab:parametrs}.
All TCL are initialized in the OFF mode ($m(0) = 0$) and with a temperature at the set-point ($\theta(0) = \theta_s$).
Unlike the deterministic dynamics considered in \cite{HetTCL},
the model in \eqref{eq:tcl_dyn} includes the process noise: 
we select initially a small value for the standard deviation as $\sigma = 0.001\sqrt{h} = 0.0032$.
\begin{table}
\centering
\caption{Parameters for the case study of a homogeneous population of TCL, as from \cite{HomTCL}}
\scalebox{1}{ %
\begin{tabular}{| l l l |}
\hline
Parameter & Interpretation & Value\\
\hline
$\theta_s$ & temperature set-point & $20\,[^\circ C]$\\
$\delta$ & dead-band width & $0.5\,[^\circ C]$ \\
$\theta_a$ & ambient temperature & $32\,[^\circ C]$ \\
$R$ & thermal resistance & $2\,[^\circ C/kW]$ \\
$C$ & thermal capacitance & $10\,[kWh/^\circ C]$\\
$P_{rate}$ & power & $14\,[kW]$\\
$\eta$ & coefficient of performance & $2.5$\\
$h$& time step & $10\,[sec]$\\
\hline
\end{tabular}
}
\label{tab:parametrs}
\end{table}

The abstraction in \cite{HetTCL} is obtained by partitioning the dead-band exclusively 
and by ``moving the probability mass'' outside of this interval to the next bin in the opposite mode. 
Recall that in the new approach put forward in this work we need to provide a partition not only for the dead-band but for the allowed range of temperatures (cf. Fig. \ref{fig:part}).     
Sample trajectories of the TCL population are presented in Figure \ref{fig:hom_sim_popul}: 
the second plot, obtained for a larger value of noise level, 
confirms that we need to partition the whole temperature range, 
rather than exclusively the dead-band. 
\begin{figure}
\centering
\subfigure{
\includegraphics[scale=0.5]{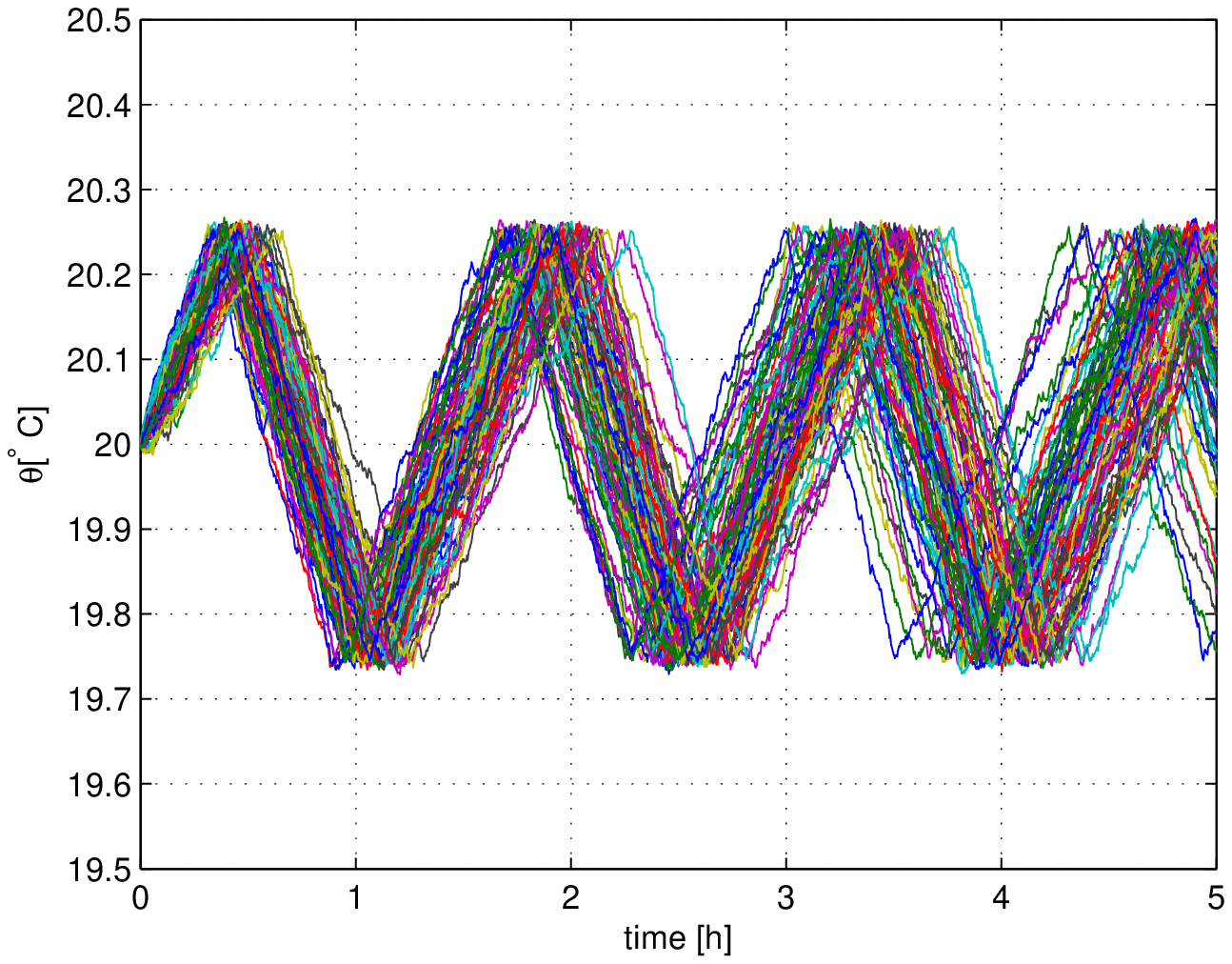}
}
\subfigure{
\includegraphics[scale=0.5]{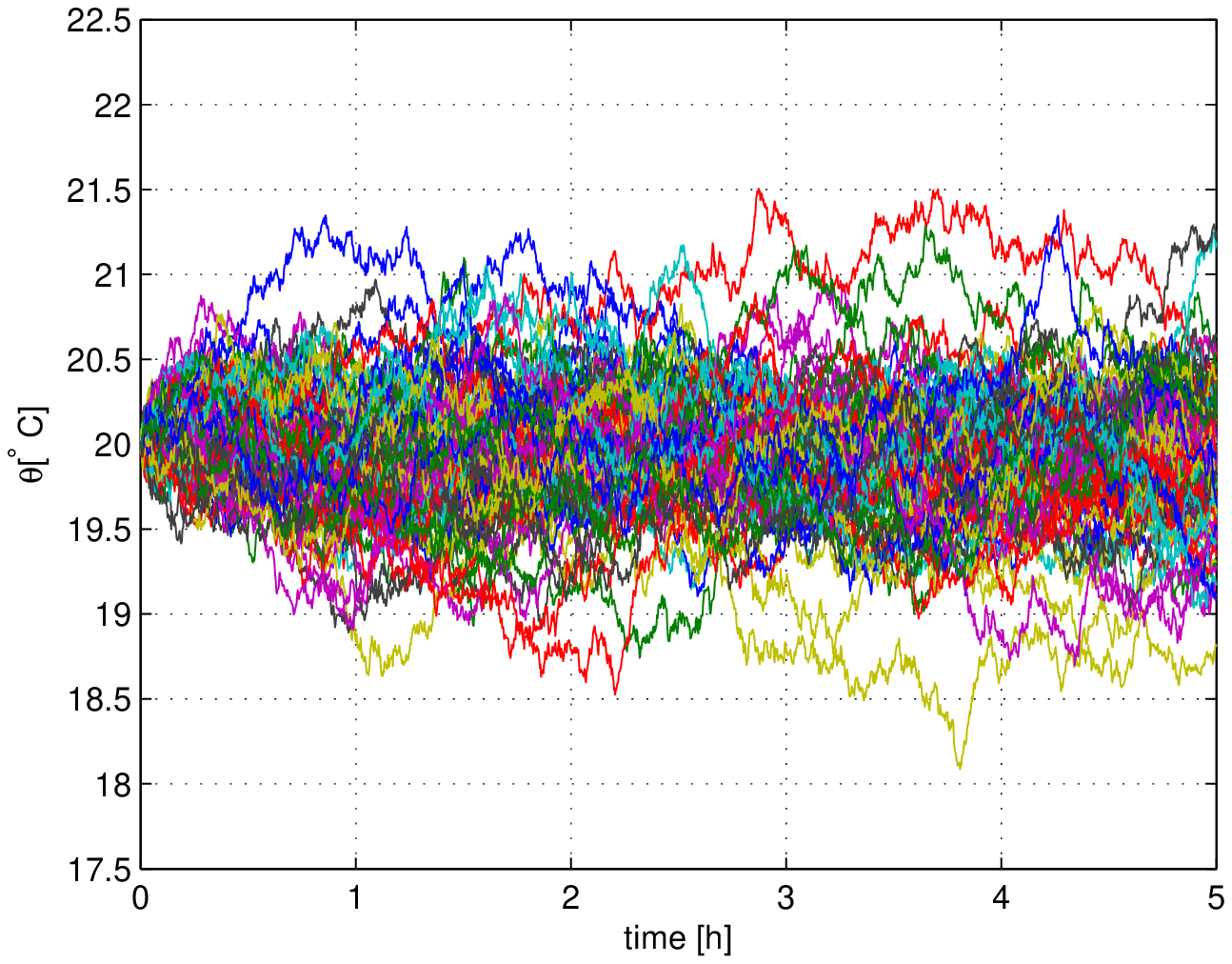}
}
\caption{Sample trajectories of the TCL population for two different values of the standard deviation of the process noise ($\sigma = 0.0032$ and $\sigma = 0.032$).}
\label{fig:hom_sim_popul}
\end{figure}

The abstraction in \cite{HetTCL} depends on a parameter $n_d$, 
denoting the number of bins:
we select $n_d = 70$, which leads to a total of $140$ states. 
The selection of $n_d$ has followed empirical tuning targeted toward optimal performance -- however, 
in general there seems to be no clear correspondence between the choice of $n_d$ and the overall precision of the abstraction procedure in \cite{HetTCL}.  

For the formal abstraction proposed in this work, 
we construct the partition as in \eqref{eq:partition} with parameters $\mathsf l=70, \mathsf m=350$, which leads to $2n = 1404$ abstract states.  
Here notice that the presence of a small standard deviation $\sigma$ for the process noise (not included in the dynamics of  \cite{HetTCL}) requires a smaller partition size to finely resolve the probability of jumps between adjacent bins. Let us emphasize again that an increase in $n_d$ for the method in \cite{HetTCL} does not lead to an improvement of the outcomes. 

The results obtained for a small noise level $\sigma = 0.0032$ are presented in Figure \ref{fig:determ_homog}. 
The aggregate power consumption has an oscillatory decay since all thermostats are started in a single state bin (they share the same initial condition).
This outcome matches that presented in \cite{HetTCL}: 
the deterministic abstraction in \cite{HetTCL} produces precise results for the first few (2-3) oscillations, 
after which the disagreement over the aggregate power between the models increases.

\begin{figure}
\centering
\subfigure{
\includegraphics[scale=0.5]{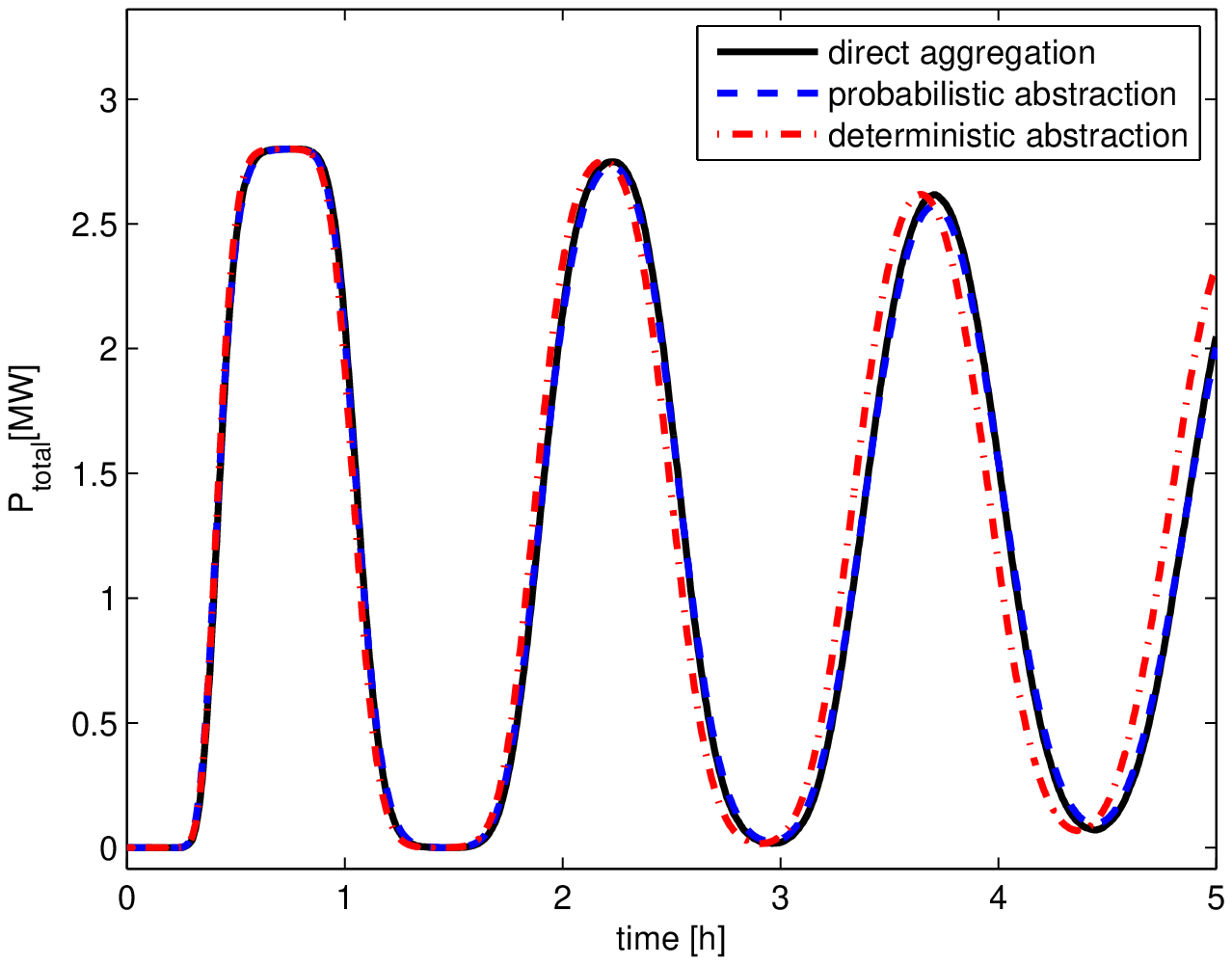}
\label{fig:subfig1}
}
\subfigure{
\includegraphics[scale=0.5]{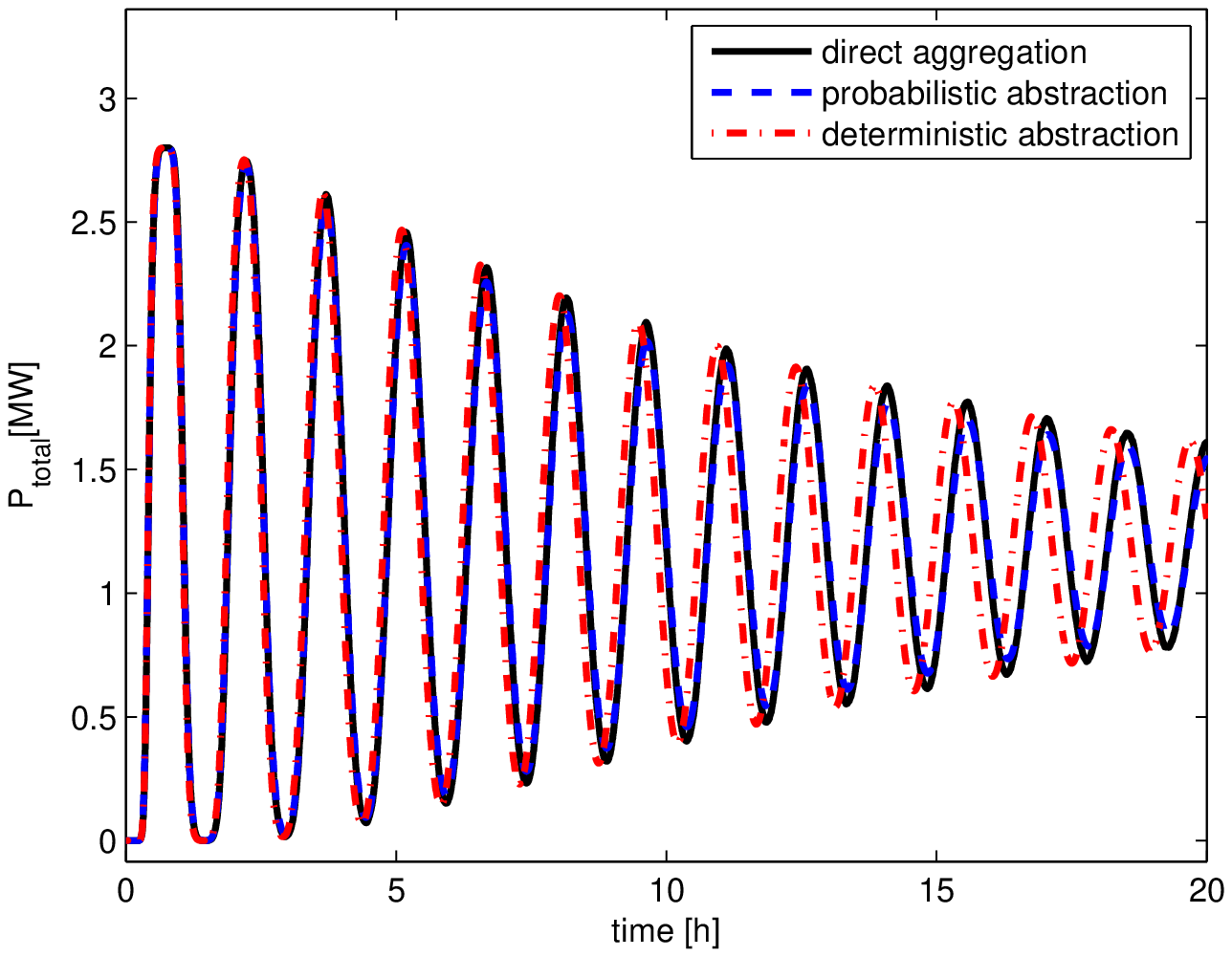}
\label{fig:subfig2}
}
\caption{Homogeneous population of TCL. Comparison of the deterministic abstraction from \cite{HetTCL} with the formal stochastic abstraction, 
for a small process noise $\sigma = 0.0032$.}
\label{fig:determ_homog}
\end{figure}

Let us now select a standard deviation for the process noise to take a larger value $\sigma = 0.01\sqrt{h} = 0.032$, 
all other parameters being the same as before.    
We now employ $n_d = 5$ (by empirical optimal tuning), and $\mathsf l = 7$, and $\mathsf m = 35$,
which leads to $10$ and $144$ abstract states, respectively. 
Figure \ref{fig:stoch_homog} presents the results of the experiment.
It is clear that the model abstraction in \cite{HetTCL} is not able to generate a good trajectory for the aggregate power,
whereas the output of the formal abstraction proposed in this work nicely matches that of the average aggregated power consumption.  
Let us again remark that increasing number of bins $n_d$ does not seem to improve the performance of the deterministic abstraction in \cite{HetTCL}, 
but rather renders the oscillations more evident.
On the contrary, our approach allows a quantification of an explicit bound on the error made: 
for instance, the error on the normalized power consumption with parameters $N = 2$ and $\mathsf l = 70$ is equal to $0.226$.
As a final remark, 
let us emphasize that the outputs of both the abstract models converge to steady-state values that may be slightly different from those obtained as the average of the Monte Carlo simulations for the model aggregated directly.    
%
%
This discrepancy is due to the intrinsic errors introduced by both the abstraction procedures, 
which approximate a concrete continuous-space model (discontinuous stochastic difference equation) with discrete-space abstractions (finite-state Markov chains).  
However, whereas the abstraction in \cite{HetTCL} does not offer an explicit quantification of the error, 
the formal abstraction proposed in this work does, 
and further allows the tuning (decrease) of such error bound, 
by choice of a larger cardinality for the partitions set.  
However as a tradeoff, 
recall that increasing the number of partitions demands handling a Markov chain abstraction with a larger size. 

\begin{figure}
\centering
\includegraphics[scale=0.5]{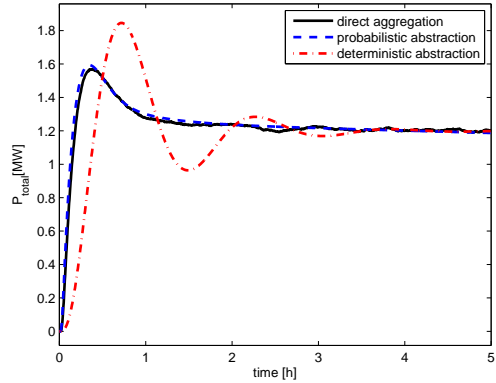}
\caption{Homogeneous population of TCL. Comparison of the deterministic abstraction from \cite{HetTCL} with the formal stochastic abstraction, 
for a larger process noise $\sigma = 0.032$.}
\label{fig:stoch_homog}
\end{figure}

\subsection{Aggregation of an Heterogeneous Population of TCL}
Let us assume that heterogeneity enters the TCL population over the thermal capacitance $C$ of each single TCL, 
which is taken to be $C \sim \mathcal U([8,12])$, 
that is described by a uniform distribution over a compact interval. 

The Monte Carlo simulations are performed with a noise level $\sigma = 0.032$,
and we have selected discretization parameters $n_d = 6$ (deterministic abstraction), 
and $\mathsf l = 10, \mathsf m = 50$ (probabilistic abstraction via averaging).
Figure \ref{fig:stoch_heter} (left) compares the results of the two abstraction methods: 
the plots are quite similar to those for the homogeneous case, 
since the allowed range for the parameter is small. 
However, let us now increase the level of heterogeneity by enlarging the domain of definition of the thermal capacitance, 
so that $C \sim \mathcal U([2,18])$.
The (empirically) best possible deterministic abstraction is obtained by selecting $n_d = 7$, 
whereas we again select $\mathsf l = 10$, $\mathsf m = 50$ for the probabilistic abstraction based on averaging.
The outcomes are presented in Figure \ref{fig:stoch_heter} (right).

\begin{figure}
\centering
\subfigure{
\includegraphics[scale=0.5]{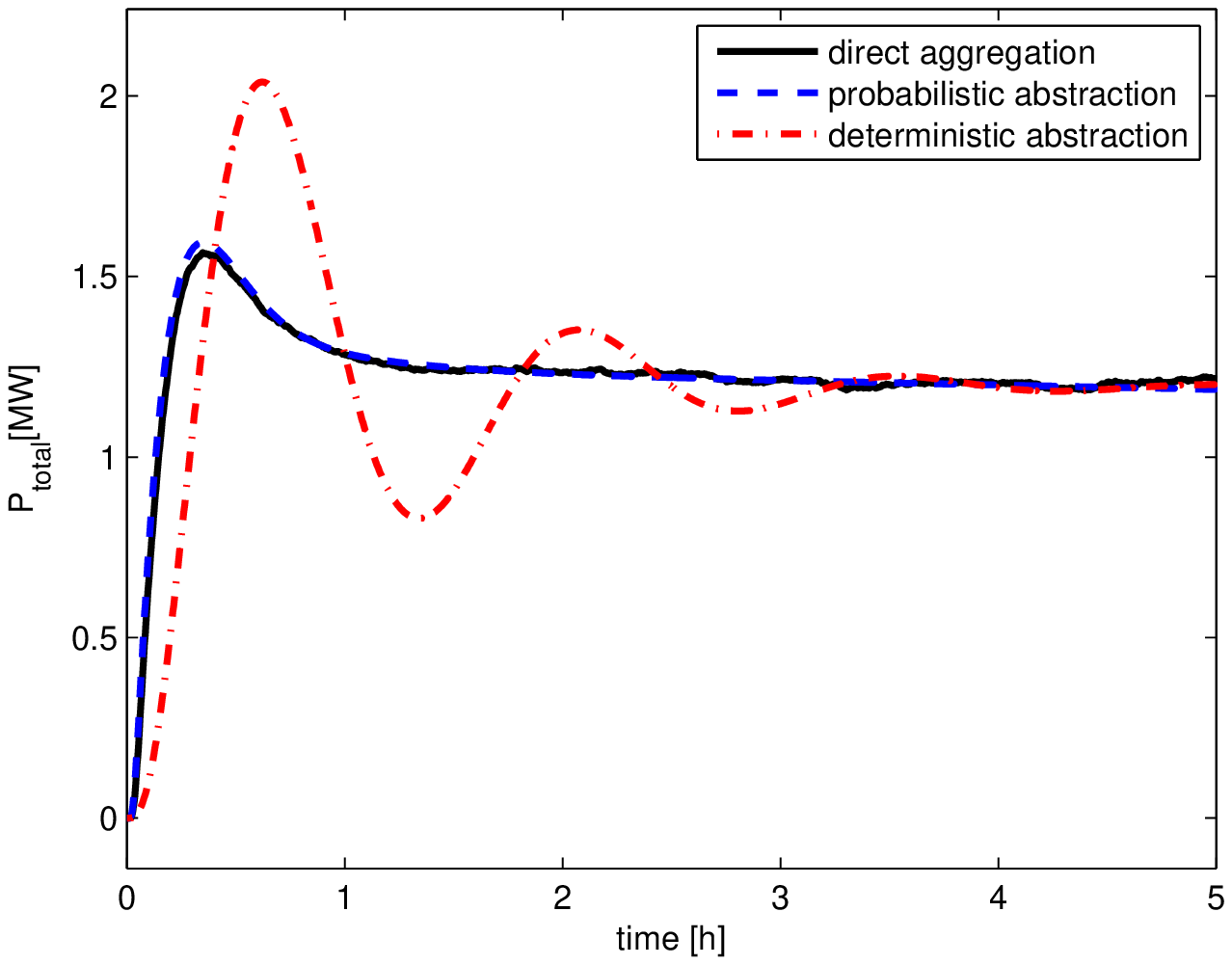}
}
\subfigure{
\includegraphics[scale=0.5]{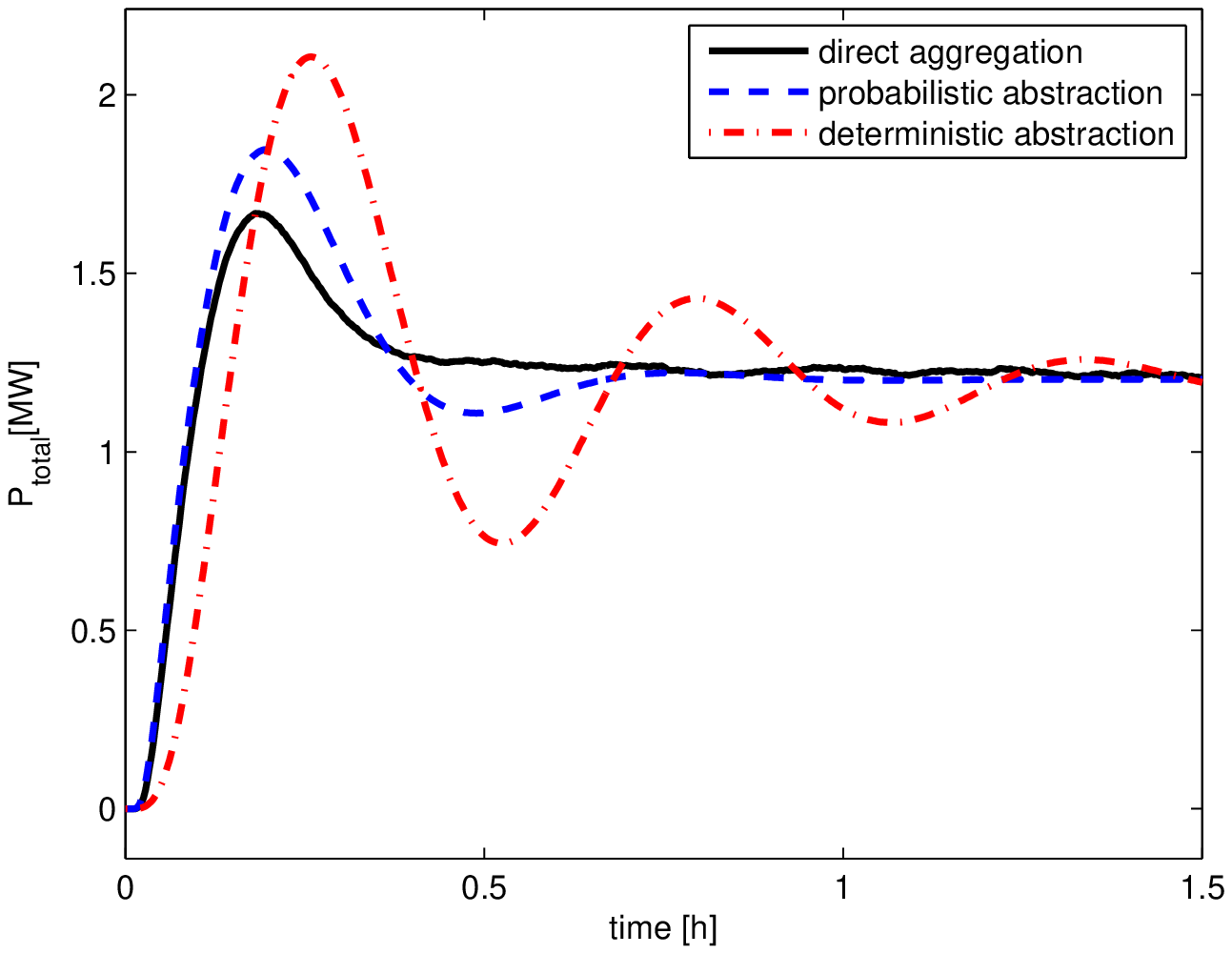}
}
\caption{Heterogeneous population of TCL. Comparison of the deterministic abstraction from \cite{HetTCL} with the formal probabilistic 
abstraction based on averaging,
for two different ranges of the thermal capacitance: $[8,12]$ (left) and $[2,18]$ (right).}
\label{fig:stoch_heter}
\end{figure}

Notice that the number of states in the linear model obtained by the probabilistic abstraction is equal to $204$, which is relatively large.
Computing the Hankel singular values of the model, we can reduce its order down to $6$ variables without corrupting its output performance.
The power consumptions provided by both the abstracted LTI model and the reduced-order models are plotted in Figure \ref{fig:model_reduc} for two different ranges of heterogeneity.
Despite a mismatch at the start of the response of the two models, 
the simulation indicates that the reduced-order model can mimic the behavior of the original one. 
The mismatch at the start of the responses is due to the initial states of the two models, which is non-essential because we use the models to estimate the states in a closed-loop configuration.

%

\begin{figure}
\centering
\subfigure{
\includegraphics[scale=0.5]{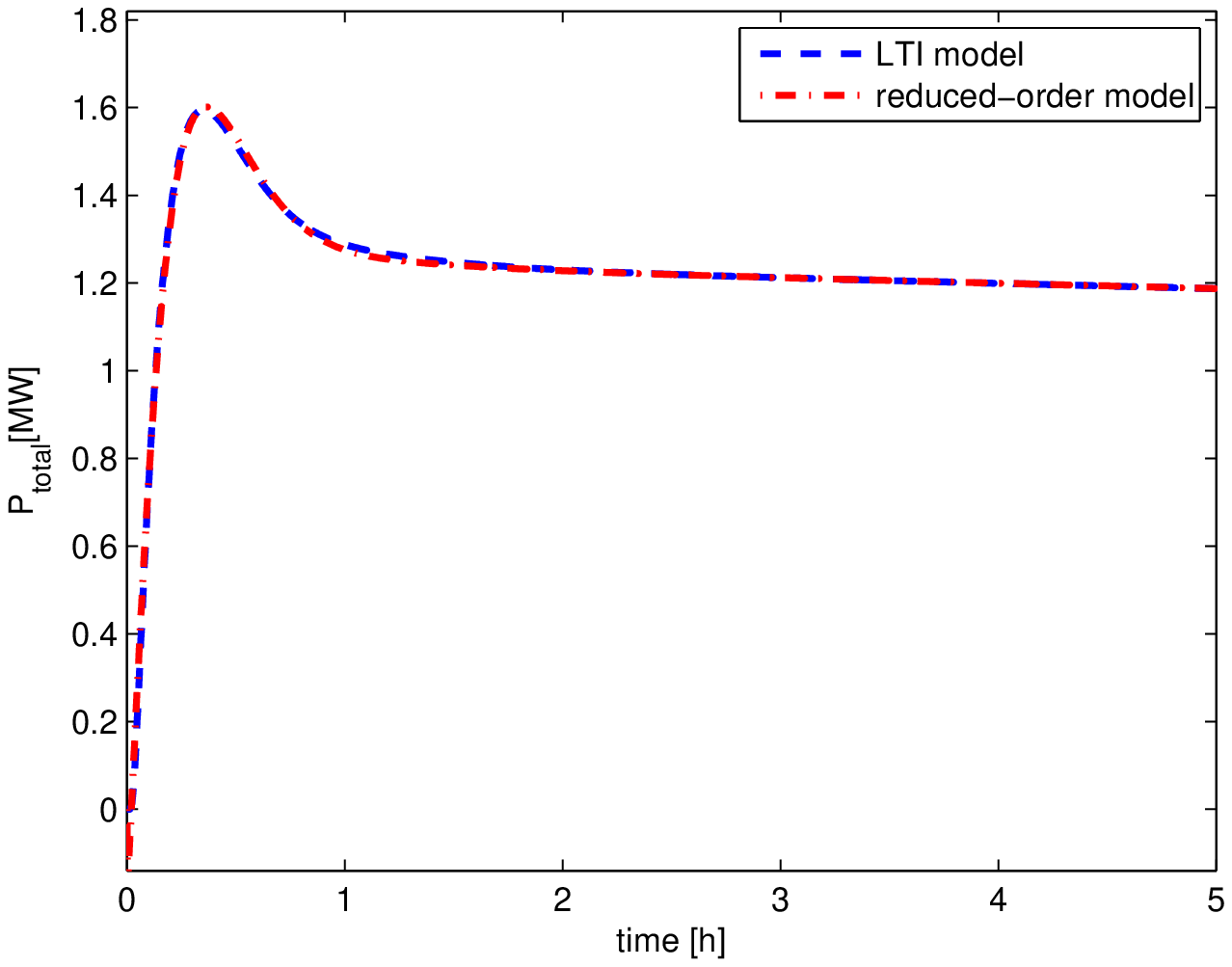}
}
\subfigure{
\includegraphics[scale=0.5]{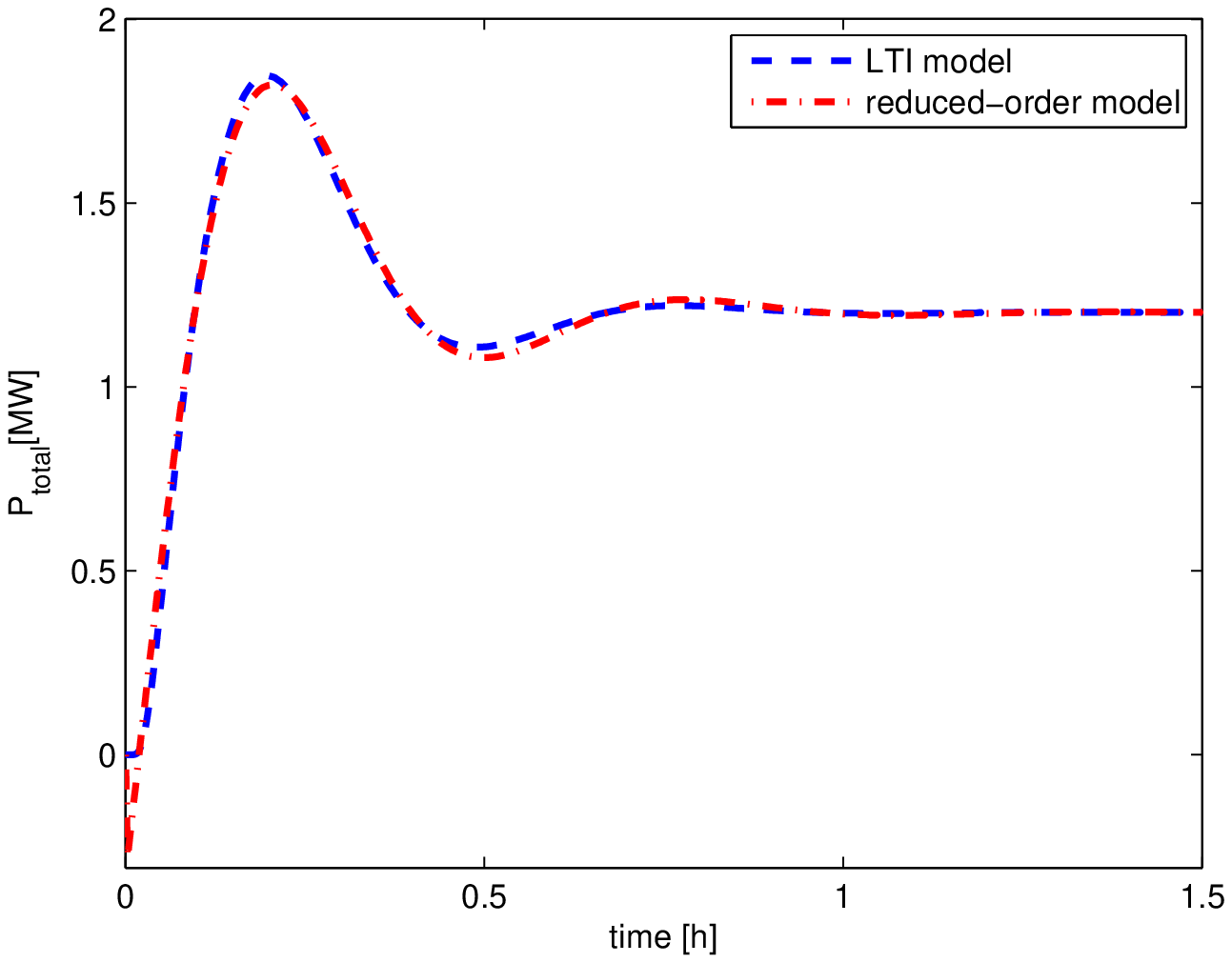}
}
\caption{
Heterogeneous population of TCL: 
comparison of the trajectories of the LTI model (probabilistic abstraction obtained via averaging)
and of the reduced-order system for two different ranges of thermal capacitance: $[8,12]$ (left) and $[2,18]$ (right).}
\label{fig:model_reduc}
\end{figure}

Figure \ref{fig:stoch_heter_clust} compares the performance of the two abstraction approaches described in Section \ref{sec:Het_avg} (via averaging) and  in Section  \ref{sec:Het:cluster} (via clustering). 
Two ranges of thermal capacitance ($[8,12]$ and $[2,18]$ respectively) characterize the heterogeneity in the population. 
For the approach of Section \ref{sec:Het:cluster} the population is clustered into $5$ and $20$ clusters, respectively.
Figure \ref{fig:stoch_heter_clust} indicates that the performance of clustering approach surpasses that of the averaging approach:
while the latter can be suitable for small heterogeneity, the former is essential for large heterogeneity in the population.  

\begin{figure}
\centering
\subfigure{
\includegraphics[scale=0.5]{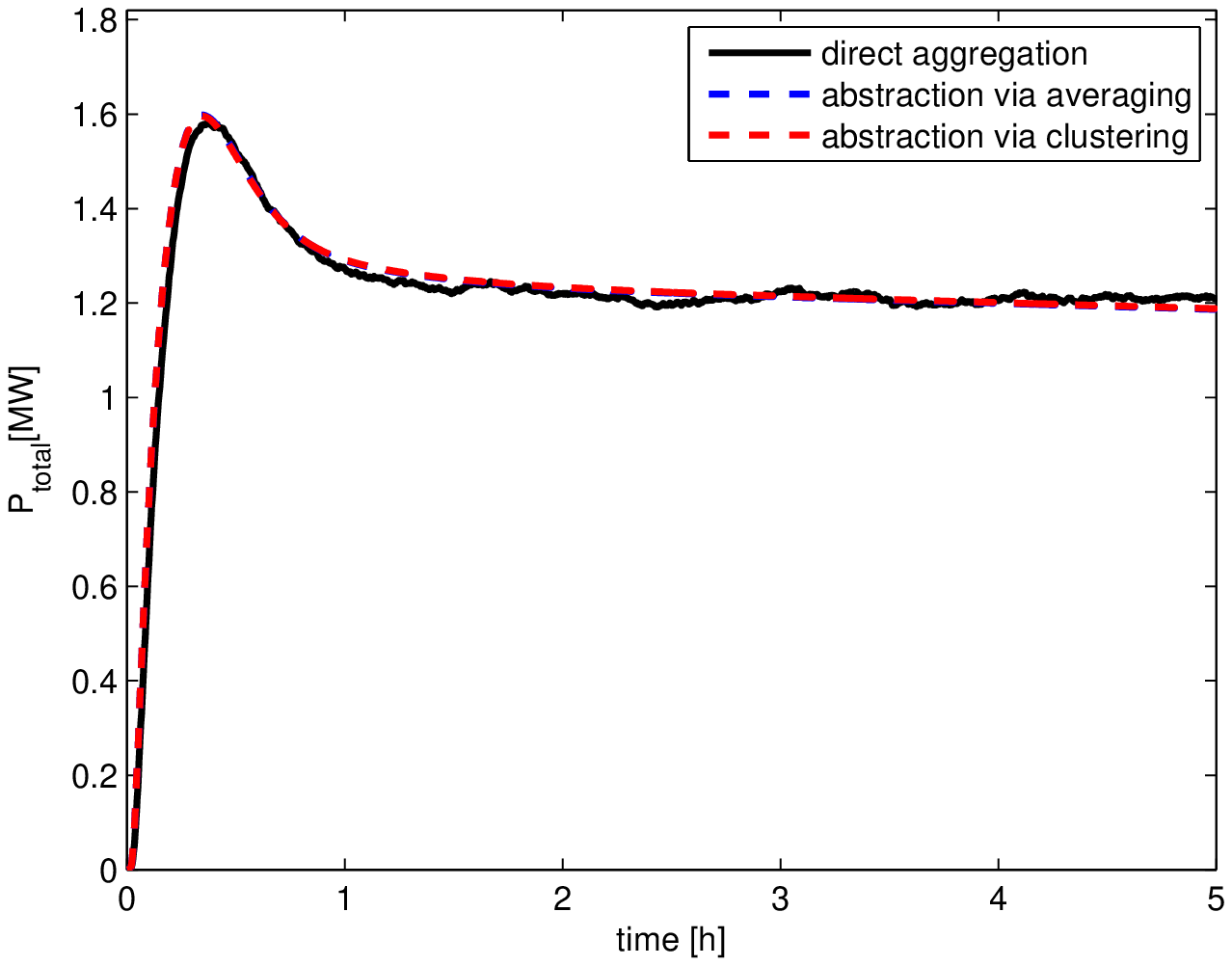}
}
\subfigure{
\includegraphics[scale=0.5]{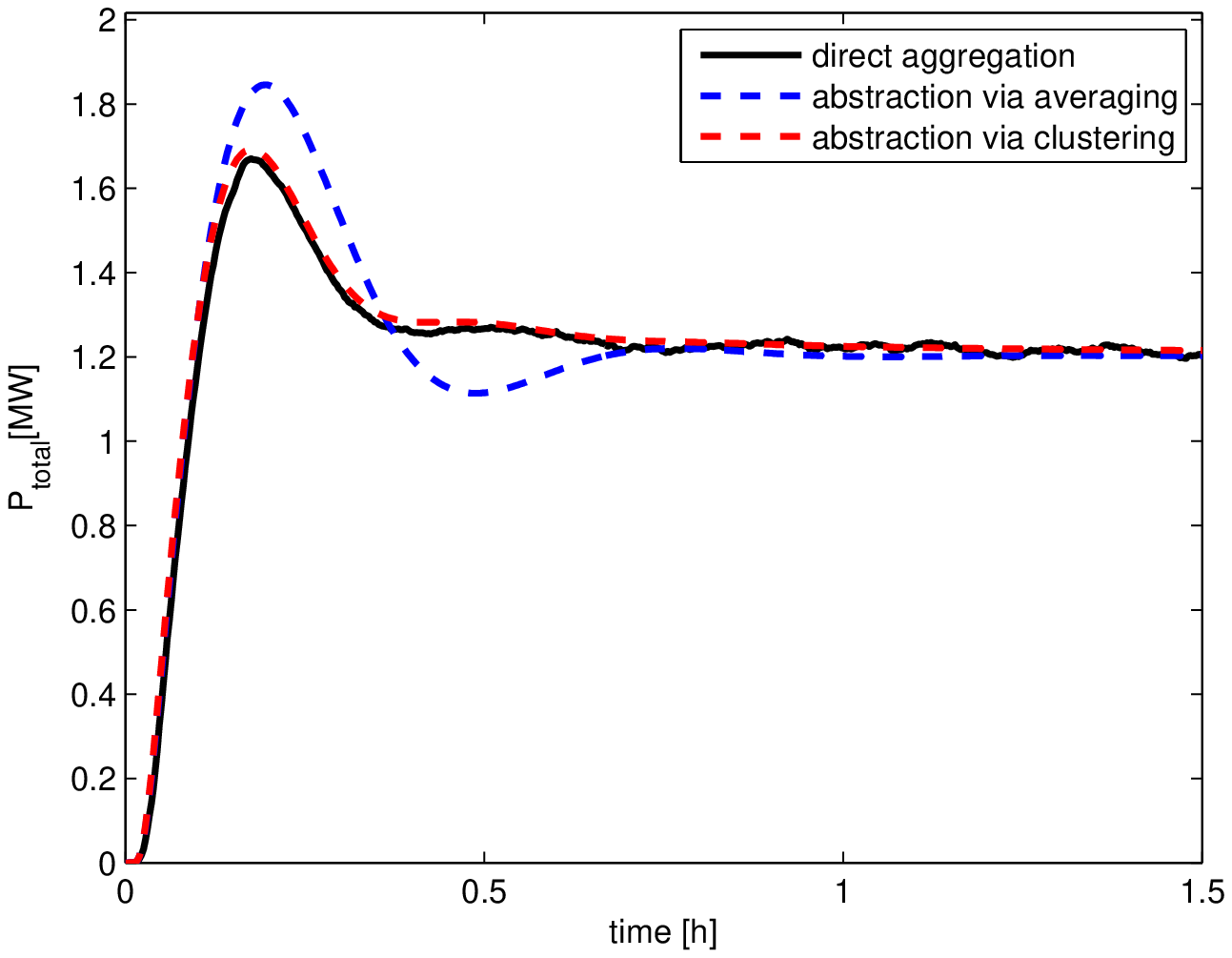}
}
\caption{
Heterogeneous population of TCL. Comparison of formal stochastic abstraction based on averaging (Section \ref{sec:Het_avg}) and clustering (Section \ref{sec:Het:cluster}),
for two different ranges of thermal capacitance: $[8,12]$ (left) and $[2,18]$ (right). The number of clusters are $5$ and $20$, respectively. 
}
\label{fig:stoch_heter_clust}
\end{figure}

\subsection{Abstraction and Control of a Population of TCL}

With focus on the abstraction proposed in this work for a homogeneous population (again of $n_p=500$ TCL),   
the one-step output prediction and regulation scheme of Section \ref{sec:est_Regul} is applied 
with the objective of tracking a randomly generated piece-wise constant reference signal. 
We have used a discretization parameters $\mathsf l = 8$, $\mathsf m=40$, 
and the standard deviation of the measurement ($\sqrt{R_v}$) has been chosen to be $0.5\%$ of the total initial power consumption. 
Figure \ref{fig:signal_track_3} displays the tracking outcome (left),
as well as the required set-point signal synthesized by the above optimization problem (right). 
Notice that the set-point variation is bounded to within a small interval, 
which practically means that the users of the TCL are physiologically unaffected by that.    


\begin{figure}
\centering
\subfigure{
\includegraphics[scale=0.5]{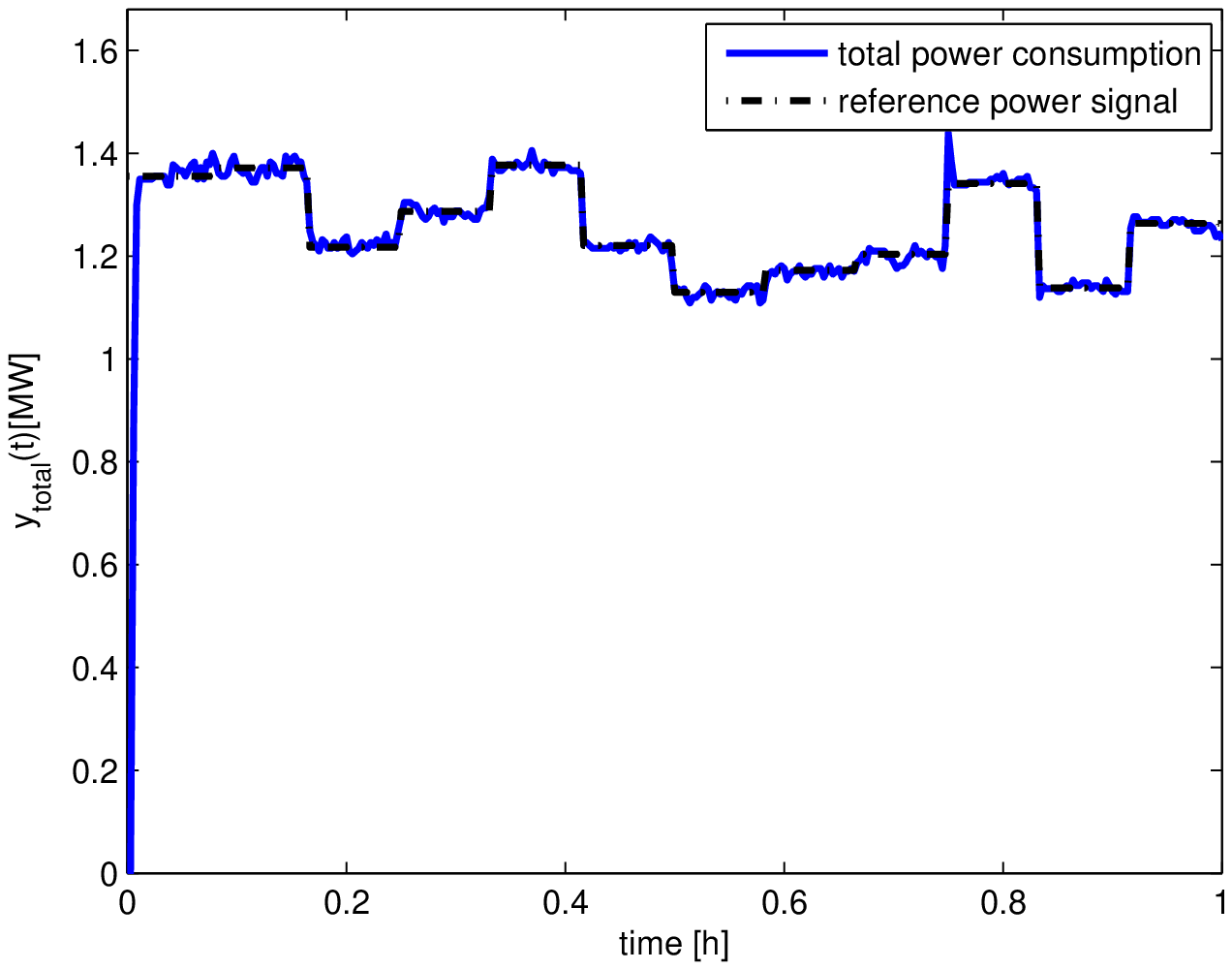}
}
\subfigure{
\includegraphics[scale=0.5]{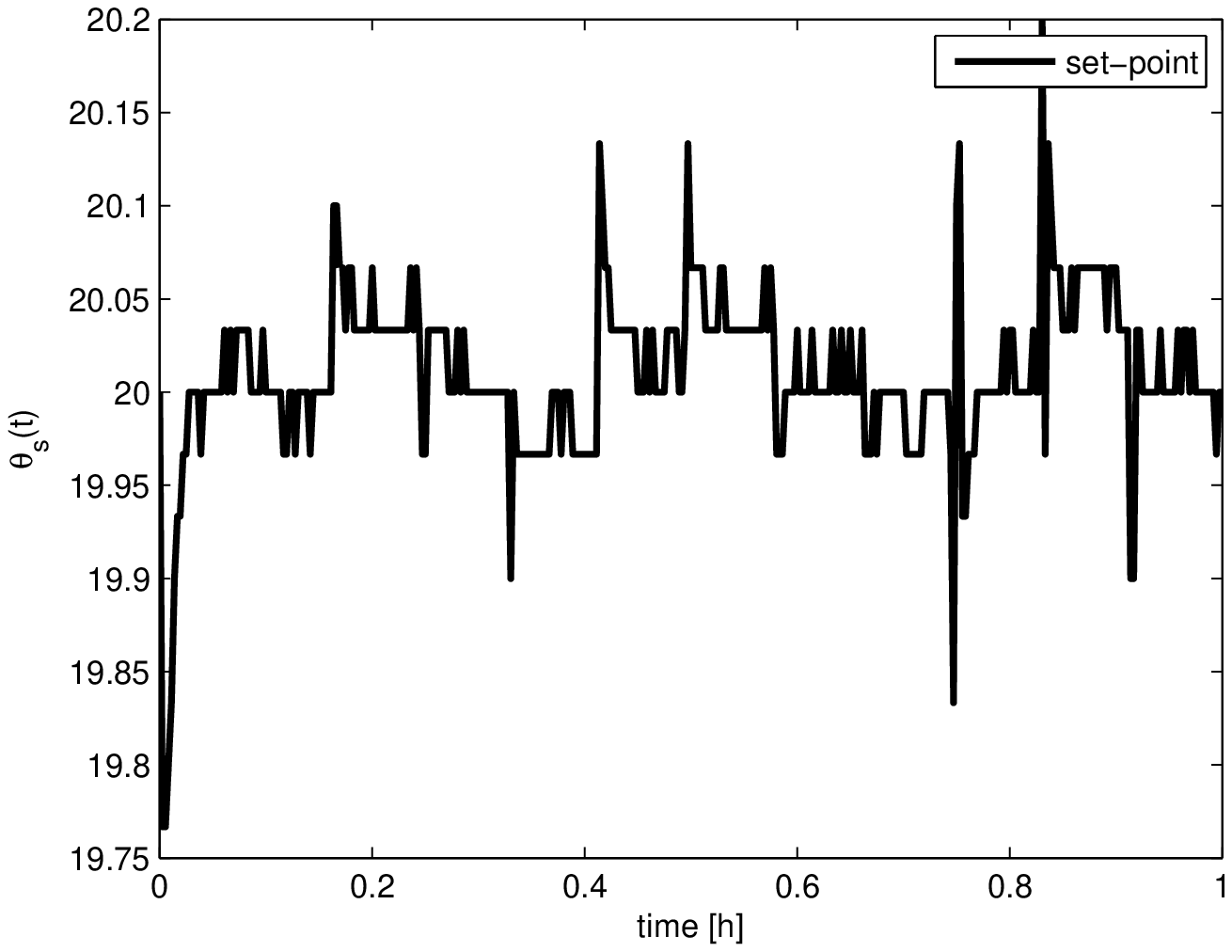}
}
\caption{Tracking of a piece-wise constant reference signal (left) by set-point control (right) in a homogeneous population of TCL abstracted by the formal probabilistic approach.
}
\label{fig:signal_track_3}
\end{figure}


A similar performance, 
as displayed in Figure \ref{fig:tracking_het_1},
is obtained in the case of a heterogeneous population (again of $500$ TCL),   
where heterogeneity is characterized by the parameter $C\in\mathcal U([2,18])$. 
The averaging approach of Section \ref{sec:Het_avg} is employed for the abstraction of the population. 
Figure \ref{fig:tracking_het_2} displays similar outcomes for a heterogeneous population abstracted by the approach of Section \ref{sec:Het:cluster} using $20$ clusters. 

\begin{figure}
\centering
\subfigure{
\includegraphics[scale=0.5]{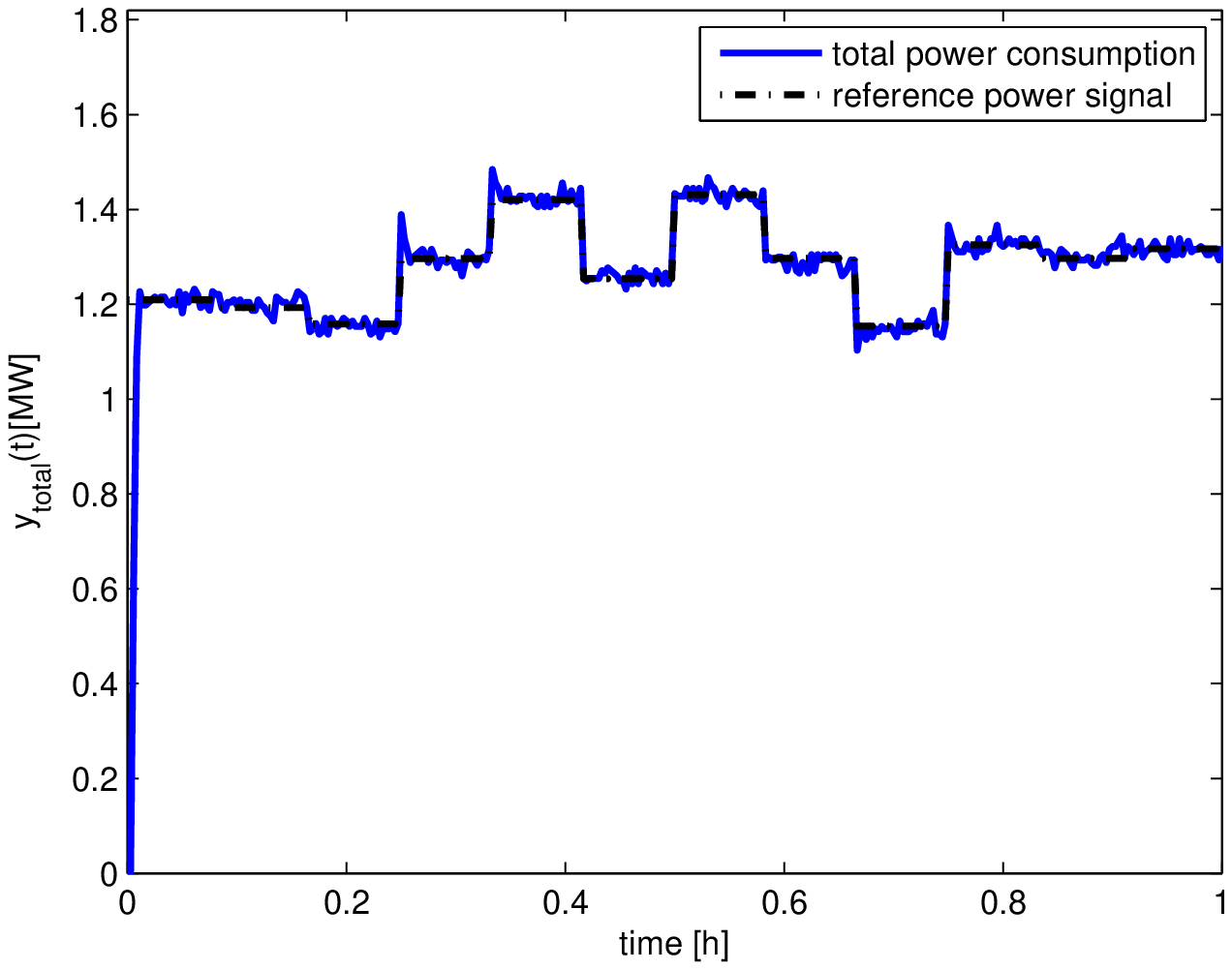}
}
\subfigure{
\includegraphics[scale=0.5]{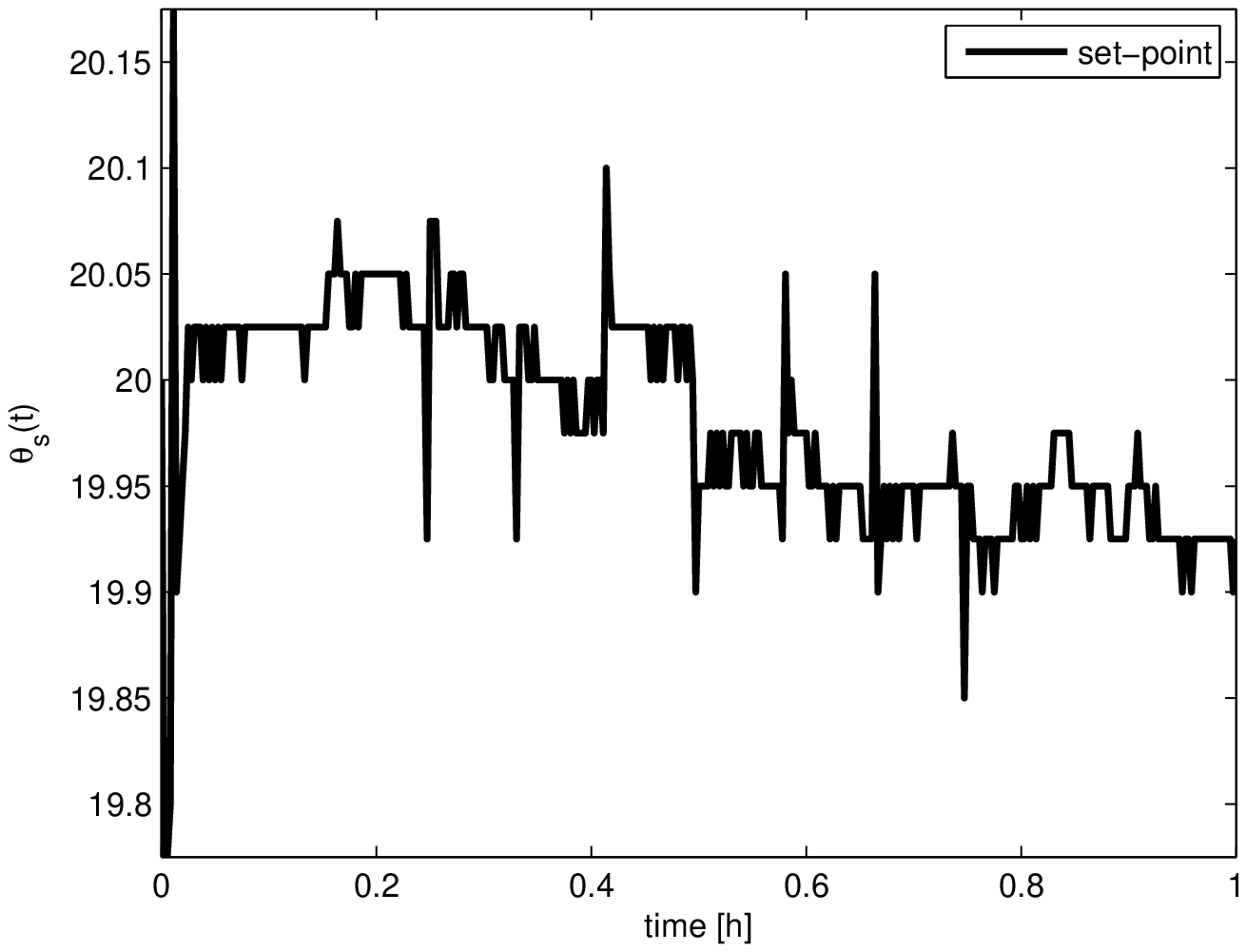}
}
\caption{Tracking of a piece-wise constant reference signal (left) by set-point control (right) for a heterogeneous population of TCL with $C\in \mathcal U([2,18])$, 
abstracted via averaging.}
\label{fig:tracking_het_1}
\end{figure}

\begin{figure}
\centering
\subfigure{
\includegraphics[scale=0.5]{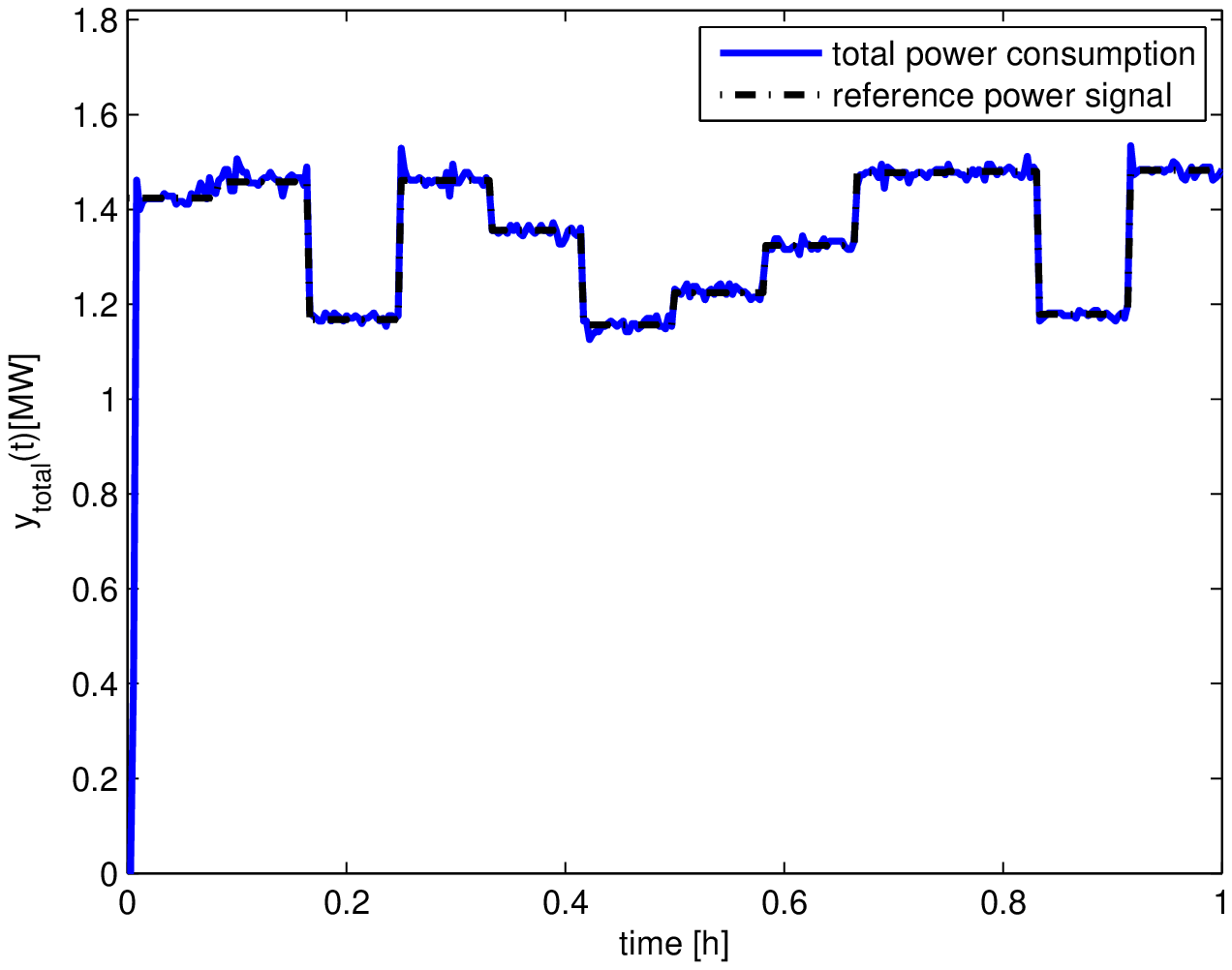}
}
\subfigure{
\includegraphics[scale=0.5]{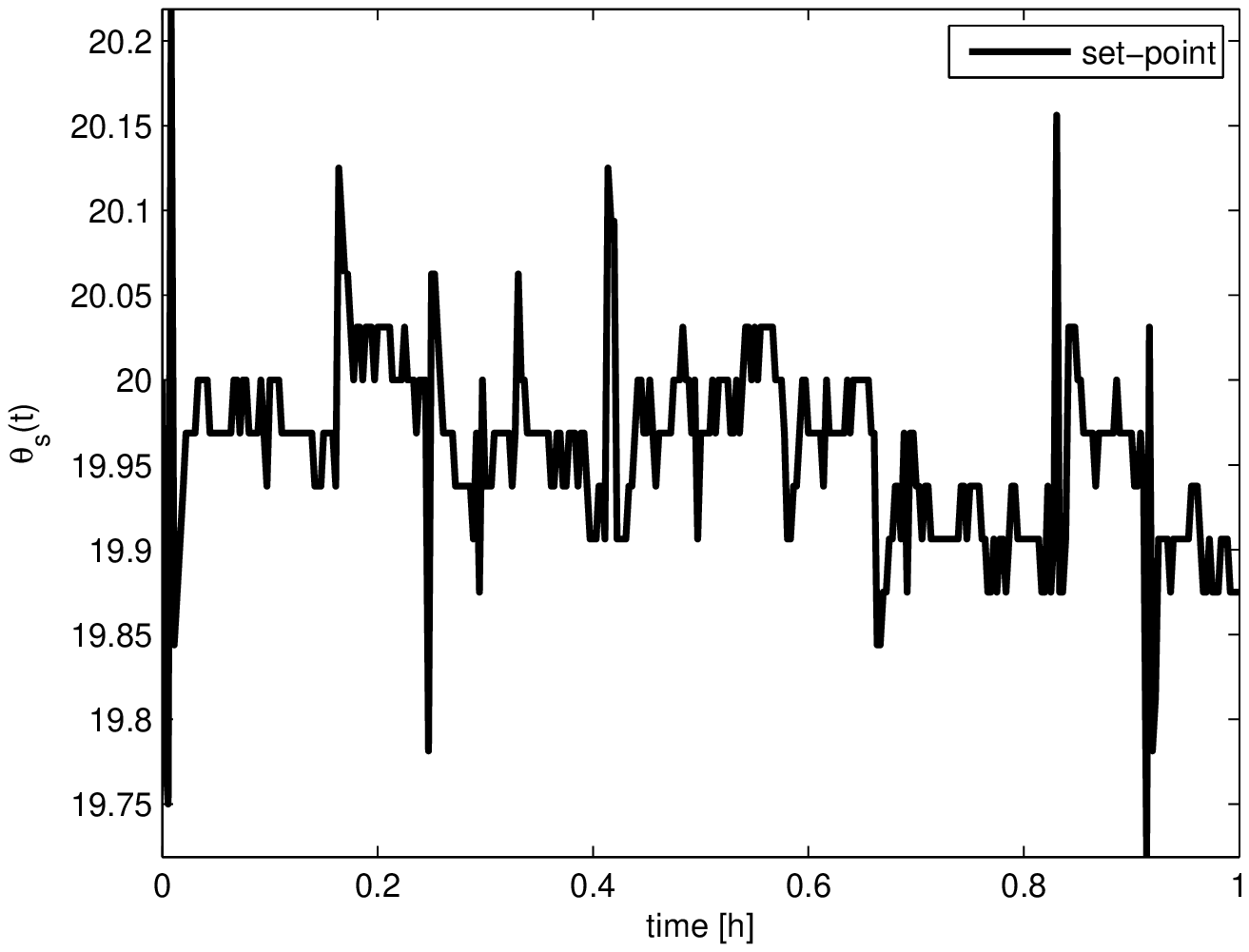}
}
\caption{Tracking of a piece-wise constant reference signal (left) by set-point control (right) for a heterogeneous population of TCL with $C\in \mathcal U([2,18])$,
abstracted via clustering ($20$ clusters).}
\label{fig:tracking_het_2}
\end{figure}


Finally, we have employed the SMPC scheme described in Section \ref{sec:SMPC}
combined with the Kalman state estimator of Section \ref{sec:est_Regul}
to track a constant reference signal over a homogeneous population of TCL.  
A prediction horizon of $T-t = 5$ steps has been selected, 
while the following constraint on the variation of the set-point has been considered:
$
\left|\frac{d\theta_s}{dt}\right|\simeq \left|\frac{\theta_s(t+1)-\theta_s(t)}{h}\right| \le \upsilon = 0.025.
$
Figure \ref{fig:tracking_hom_3} presents the power consumption of the population (left) and required set-point (right). 
The displayed response consists of a transient and of a steady-state phases. 
It takes $3$ minutes to reach the steady-state phase because of the limitations on the rate of set-point changes. 
This can be seen from the plot of the set-point control signal, 
which first decreases and then increases within the transient phase with a constant rate. 
In order to obtain a faster transient phase, 
the upper-bound for the set-point changes could be increased. 

\begin{figure}
\centering
\subfigure{
\includegraphics[scale=0.5]{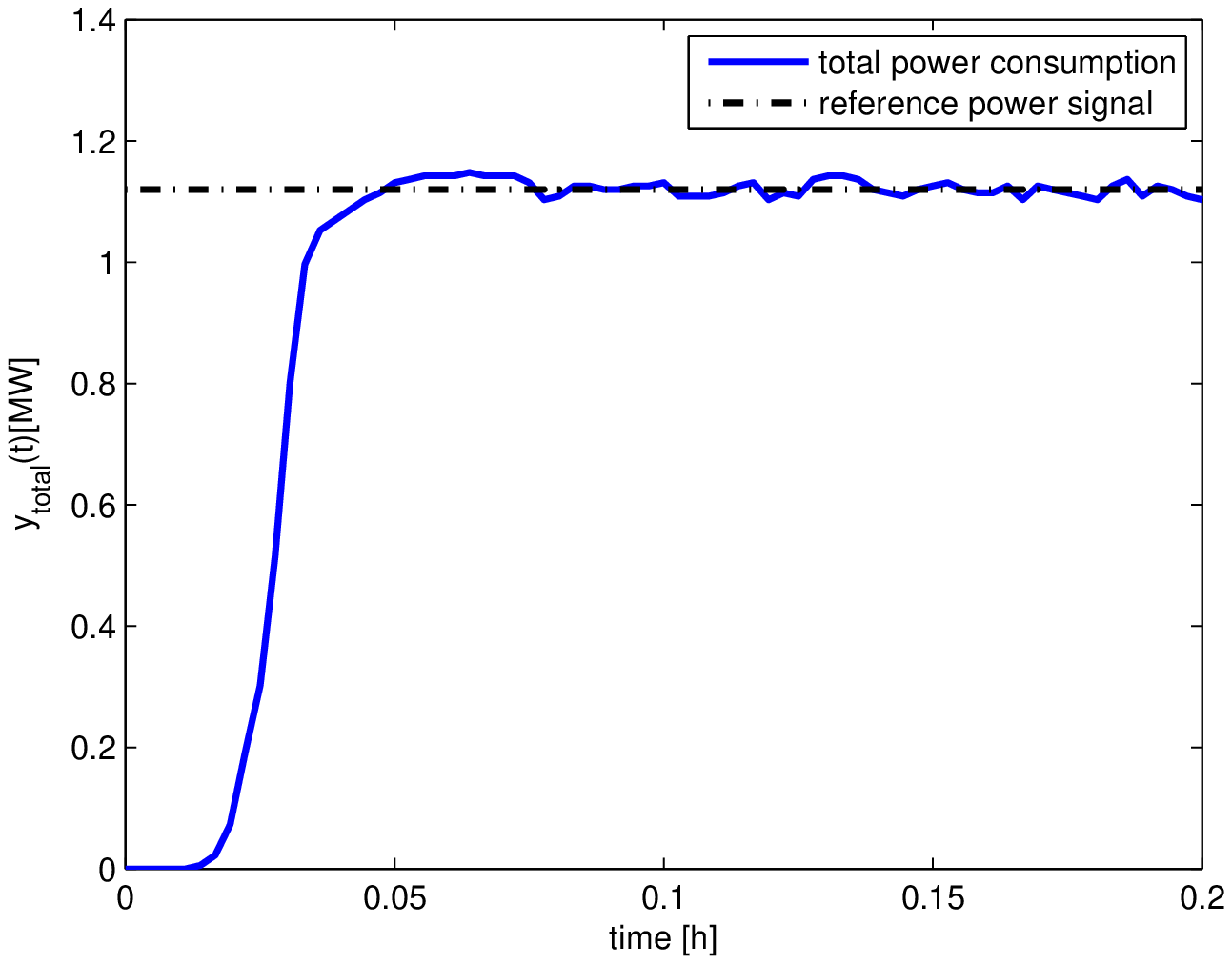}
}
\subfigure{
\includegraphics[scale=0.5]{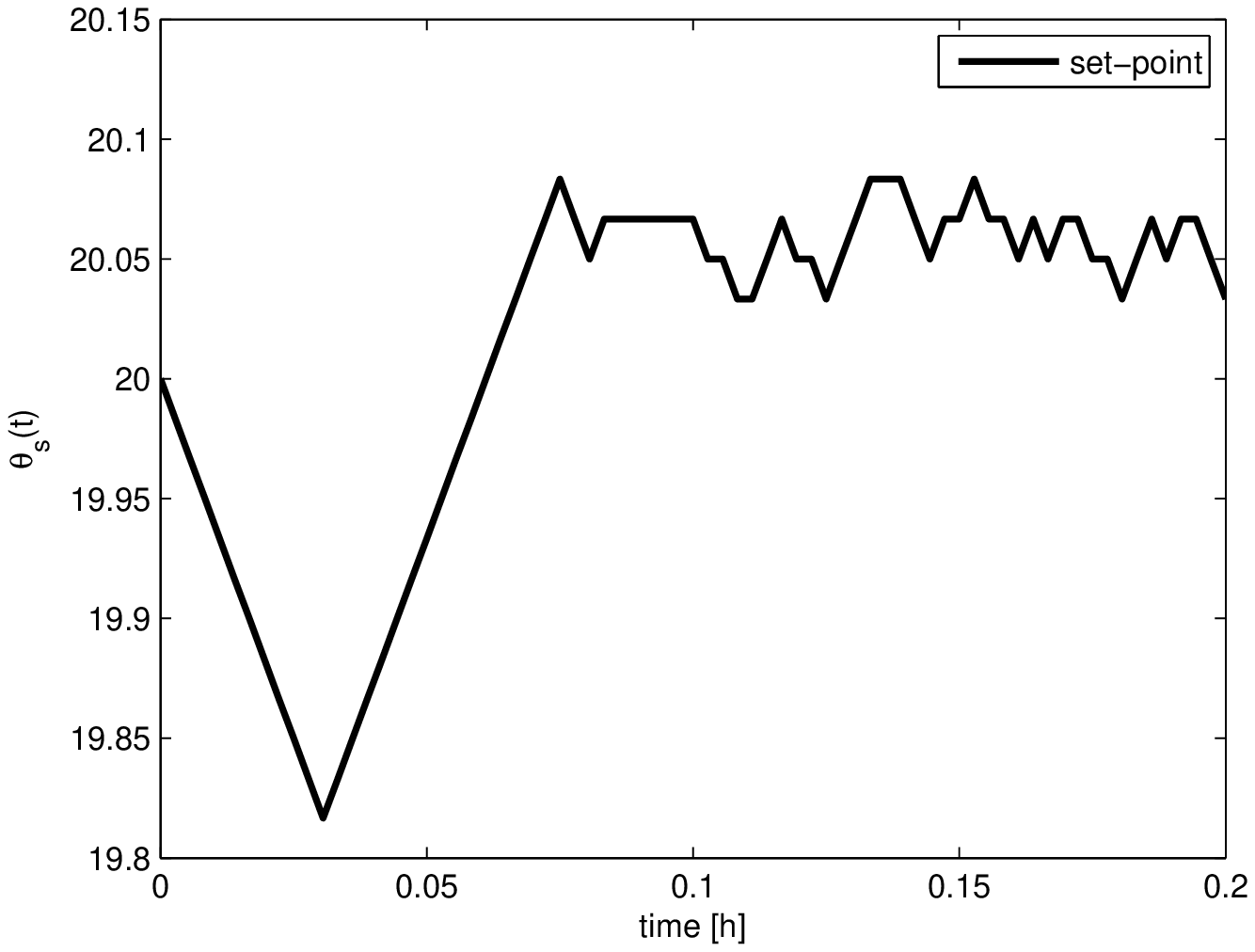}
}
\caption{Tracking of a constant reference signal (left) by set-point control (right) for a homogeneous population of TCL using the SMPC scheme.}
\label{fig:tracking_hom_3}
\end{figure}


\section{Conclusions and Future Work}

This work has put forward a formal approach for the abstraction of the dynamics of TCL and the aggregation of a population model. 
The approach starts by partitioning the state-space and constructing Markov chains for each single system.
Given the transition probability matrix of the Markov chains, it is possible to write down the state-space model of the population and further to aggregate it. 
The article has discussed approaches to deal with models heterogeneity and to perform controller synthesis over the aggregated model. 
It is worth mentioning that the error bound derived for autonomous populations can be extended to controlled populations. 


Looking forward, developing alternative approaches for the heterogeneous case, 
synthesizing new control schemes, 
and improving the error bounds are directions that are research-worthy in order to render the approach further applicable in practice.

\bibliographystyle{siam}

\begin{thebibliography}{10}

\bibitem{APKL10}
{\sc A.~Abate, J.-P. Katoen, J.~Lygeros, and M.~Prandini}, {\em Approximate
  model checking of stochastic hybrid systems}, European Journal of Control, 6
  (2010), pp.~624--641.

\bibitem{APLS08}
{\sc A.~Abate, M.~Prandini, J.~Lygeros, and S.~Sastry}, {\em Probabilistic
  reachability and safety for controlled discrete time stochastic hybrid
  systems}, Automatica, 44 (2008), pp.~2724--2734.

\bibitem{A05}
{\sc A.C. Antoulas}, {\em {Approximation of large-scale dynamical systems}},
  Society for Industrial Mathematics, 2005.

\bibitem{BK08}
{\sc C.~Baier and J.-P. Katoen}, {\em Principles of Model Checking}, MIT Press,
  2008.

\bibitem{BF11}
{\sc S.~Bashash and H.K. Fathy}, {\em Modeling and control insights into
  demand-side energy management through setpoint control of thermostatic
  loads}, in Proceedings of the 2011 American Control Conference, San
  Francisco, CA, June 2011, pp.~4546--4553.

\bibitem{D05}
{\sc P.~Billingsley}, {\em Probability and Measure - Third Edition}, Wiley
  Series in Probability and Mathematical Statistics, 1995.

\bibitem{HomTCL}
{\sc D.S. Callaway}, {\em Tapping the energy storage potential in electric
  loads to deliver load following and regulation, with application to wind
  energy}, Energy Conversion and Management, 50 (2009), pp.~1389--1400.

\bibitem{DLT08}
{\sc J.~Desharnais, F.~Laviolette, and M.~Tracol}, {\em Approximate analysis of
  probabilistic processes: logic, simulation and games}, in Proceedings of the
  International Conference on Quantitative Evaluation of SysTems (QEST 08),
  Sept. 2008, pp.~264--273.

\bibitem{D04}
{\sc R.~Durrett}, {\em Probability: Theory and Examples - Third Edition},
  Duxbury Press, 2004.

\bibitem{SA11}
{\sc S.~{Esmaeil Zadeh Soudjani} and A.~Abate}, {\em Adaptive gridding for
  abstraction and verification of stochastic hybrid systems}, in Proceedings of
  the 8th International Conference on Quantitative Evaluation of Systems,
  Aachen, DE, September 2011, pp.~59--69.

\bibitem{SAH12}
\leavevmode\vrule height 2pt depth -1.6pt width 23pt, {\em {H}igher-{O}rder
  {A}pproximations for {V}erification of {S}tochastic {H}ybrid {S}ystems}, in
  Automated Technology for Verification and Analysis, S.~Chakraborty and
  M.~Mukund, eds., vol.~7561 of Lecture Notes in Computer Science, Springer
  Verlag, Berlin Heidelberg, 2012, pp.~416--434.

\bibitem{SA13}
\leavevmode\vrule height 2pt depth -1.6pt width 23pt, {\em Adaptive and
  sequential gridding procedures for the abstraction and verification of
  stochastic processes}, SIAM Journal on Applied Dynamical Systems, 12 (2013),
  pp.~921--956.

\bibitem{SA13TCL}
\leavevmode\vrule height 2pt depth -1.6pt width 23pt, {\em Aggregation of
  thermostatically controlled loads by formal abstractions}, in European
  Control Conference, Zurich, Switzerland, July 2013, pp.~4232--4237.

\bibitem{HCL09}
{\sc Peter Hokayem, Debasish Chatterjee, and John Lygeros}, {\em On stochastic
  receding horizon control with bounded control inputs}, in Proceedings of the
  48th IEEE Conference on Decision and Control, Shanghai, PRC, December 2009,
  pp.~6359--6364.

\bibitem{JKB97}
{\sc N.~L. Johnson, S.~Kotz, and N.~Balakrishnan}, {\em Discrete Multivariate
  Distributions}, Wiley Series in Probability and Statistics, 1997.

\bibitem{HetTCL}
{\sc S.~Koch, J.L. Mathieu, and D.S. Callaway}, {\em Modeling and control of
  aggregated heterogeneous thermostatically controlled loads for ancillary
  services}, in Proceedings of the 17th Power Systems Computation Conference,
  Stockholm, Sweden, 2011.

\bibitem{KSBH11}
{\sc S.~Kundu, N.~Sinitsyn, S.~Backhaus, and I.~Hiskens}, {\em Modeling and
  control of thermostatically controlled loads}, in Proceedings of the 17th
  Power Systems Computation Conference, 2011.

\bibitem{MC85}
{\sc R.~{Malhame} and C.-Y. {Chong}}, {\em Electric load model synthesis by
  diffusion approximation of a high-order hybrid-state stochastic system}, IEEE
  Transactions on Automatic Control, 30 (1985), pp.~854--860.

\bibitem{HetTCLCallaway}
{\sc J.L. Mathieu and D.S. Callaway}, {\em State estimation and control of
  heterogeneous thermostatically controlled loads for load following}, in
  Hawaii International Conference on System Sciences, Hawaii, USA, 2012,
  pp.~2002--2011.

\bibitem{MKC12}
{\sc J.L. Mathieu, S.~Koch, and D.S. Callaway}, {\em State estimation and
  control of electric loads to manage real-time energy imbalance}, IEEE
  Transactions on Power Systems, 28 (2013), pp.~430--440.

\bibitem{PoissBin93}
{\sc Y.H. Wang}, {\em On the number of successes in independent trials},
  Statistica Sinica, 3 (1993), pp.~295--312.

\bibitem{WMBG12}
{\sc S.~Widergren, C.~Marinovici, T.~Berliner, and A.~Graves}, {\em Real-time
  pricing demand response in operations}, in Power and Energy Society General
  Meeting, 2012 IEEE, 2012, pp.~1--5.

\bibitem{TPS_13}
{\sc W.~Zhang, J.~Lian, {C.-Y.} Chang, and K.~Kalsi}, {\em Aggregated modeling
  and control of air conditioning loads for demand response}, To appear in IEEE
  Transactions on Power Systems,  (2013).

\end{thebibliography}


\clearpage
\section{Appendix}
\label{sec:appendix}

\subsection{Proofs of the Statements}

\begin{proof}[Proof of Theorem \ref{thm:pois_bino}]
Since the states of all Markov chains are known, the Markov chain $r$ jumps to the state $i$
with probability $P_{z_r(t)i}$ and fails to jump to the state $i$ with probability $(1-P_{z_r(t)i})$.
The definition of the variable $x_i$ implies that the conditional random variable $(x_i(t+1)| z(t))$ is the sum of $n_p$ independent Bernoulli trials with different success probabilities $P_{z_r(t)i}$. 
Then it follows the Poisson-binomial distribution (cf. Table \ref{tab:discrete_distributions}) with the specified mean and variance as in Table \ref{tab:discrete_distributions1}. 
\end{proof}

\smallskip

\begin{proof}[Proof of Theorem \ref{thm:shif_hevi}]
We fix the vector $\vectr X(t) = \left[X_1(t),X_2(t),\cdots,X_{2n}(t)\right]^T$ and $i$.
The sample space of the random variable $X_i(t+1)$ is the set
$\left\{0,\dfrac{1}{n_p},\dfrac{2}{n_p},\cdots,1\right\}$.
Define the probability masses of the random variable by
\begin{equation*}
\mathsf p_{jn_p} = \mathsf P\left(X_i(t+1)= \frac{j}{n_p}\bigg|\vectr X(t)\right),\quad j\in\mathbb{Z}_{n_p}.
\end{equation*}
Then $\sum_{j=0}^{n_p}\mathsf p_{jn_p} = 1$ for all $n_p\in \mathbb{N}$.
Denote the expected value of the random variable by $\mu = \sum_{r=1}^{2n}X_r(t)P_{ri}$, which is independent of $n_p$.
The variance can be computed as 
\begin{align*}
\sigma_{n_p}^2 =\sum_{j=0}^{n_p}\mathsf p_{jn_p}(\frac{j}{n_p} - \mu)^2 = \frac{1}{n_p}\sum_{r=1}^{2n}X_r(t)P_{ri}(1-P_{ri}),
\end{align*}
and converges to zero as $n_p$ goes to infinity.
Fix an $\varepsilon>0$ and consider the following implication: 
\begin{align*}
& \sum_{j=0}^{n_p}\mathsf p_{jn_p}(\frac{j}{n_p} - \mu)^2
\ge \sum_{j=0}^{j\le n_p(\mu-\varepsilon)}\mathsf p_{jn_p}(\frac{j}{n_p} - \mu)^2
\ge \varepsilon^2\sum_{j=0}^{j\le n_p(\mu-\varepsilon)}\mathsf p_{jn_p}\ge 0
\Rightarrow \lim\limits_{n_p\rightarrow\infty}\sum_{j=0}^{j\le n_p(\mu-\varepsilon)}\mathsf p_{jn_p} = 0.
\end{align*}
The same reasoning can be applied to the upper tail
\begin{align*}
& \sum_{j=0}^{n_p}p_{jn_p}(\frac{j}{n_p} - \mu)^2
\ge \sum_{j=0}^{j> n_p(\mu+\varepsilon)}\mathsf p_{jn_p}(\frac{j}{n_p} - \mu)^2
\ge \varepsilon^2\sum_{j=0}^{j> n_p(\mu+\varepsilon)}\mathsf p_{jn_p}\ge 0
\Rightarrow \lim\limits_{n_p\rightarrow\infty}\sum_{j=0}^{j> n_p(\mu+\varepsilon)}\mathsf p_{jn_p} = 0.
\end{align*}
Define the cumulative distribution function of the random variable $(X_i(t+1)|\vectr X(t))$ as 
\begin{equation*}
F_{n_p}(x) = \mathsf P(X_i(t+1)\le x|\vectr X(t)) = \sum_{j=0}^{j\le n_p x}\mathsf p_{jn_p}.
\end{equation*}
For all $x\le \mu-\varepsilon$ we have
\begin{equation*}
\lim\limits_{n_p\rightarrow\infty}F_{n_p}(x) = 
 \lim\limits_{n_p\rightarrow\infty}\sum_{j=0}^{j\le n_p x}\mathsf p_{jn_p}
 \le \lim\limits_{n_p\rightarrow\infty}\sum_{j=0}^{j\le n_p(\mu-\varepsilon)}\mathsf p_{jn_p} = 0.
\end{equation*}
Similarly, for all $x\ge \mu+\varepsilon$ we have
\begin{equation*}
\lim\limits_{n_p\rightarrow\infty}1-F_{n_p}(x) = 
\lim\limits_{n_p\rightarrow\infty}\sum_{j=0}^{j> n_p x}\mathsf p_{jn_p}
 \le \lim\limits_{n_p\rightarrow\infty}\sum_{j=0}^{j> n_p(\mu+\varepsilon)}\mathsf p_{jn_p} = 0.
\end{equation*}
Since the above reasoning holds for any $\varepsilon>0$, $F_{n_p}(x)$ converges to the Heaviside step function shifted at point $x=\mu$.
\end{proof}

\smallskip

\begin{lemma}
\label{lm:CLT}
(Lyapunov CLT, \cite{D05})
Let $\{y_j\}$ be a sequence of independent random variables, each having a finite expected value $\mu_j$ and variance $\sigma_j^2$. Define $s_{n_p}^2 = \sum_{j=1}^{n_p}\sigma_j^2$. If for some $\delta>0$, the Lyapunov's condition
\begin{equation*}
\lim\limits_{n_p\rightarrow\infty}\frac{1}{s_{n_p}^{2+\delta}}\sum_{j=1}^{n_p}\mathbb E\left[|y_j-\mu_j|^{2+\delta}\right] = 0,
\end{equation*}
is satisfied, then the variable $\sum_{j=1}^{n_p}(y_j - \mu_j)/s_{n_p}$ converges, in distribution, to a standard normal random variable, as $n_p$ goes to infinity:
\begin{equation*}
\frac{1}{s_{n_p}}\sum_{j=1}^{n_p}(y_j - \mu_j)\xrightarrow{d}\mathcal{N}(0,1).
\end{equation*}
\end{lemma}

\smallskip

\begin{proof}[Proof of Theorem \ref{thm:Gaus_appr}]
In order to prove that the random vector $(\vectr X(t+1)|\vectr X(t))$ converges to a multivariate normal random variable,
we show that every linear combination of its components converges to a normal random variable.
Consider any arbitrary vector $\nu = [\nu_1,\nu_2,\cdots,\nu_{2n}]^T\in\mathbb R^{2n}$.
The random variable $(\nu^T\vectr X(t+1)|\vectr X(t))$ can be seen as the sum of $n_p$ independent (normalized) categorical random variables $y_j$ over the sample space $\{\frac{\nu_1}{n_p},\frac{\nu_2}{n_p},\cdots,\frac{\nu_{2n}}{n_p}\}$,
where 
$x_1(t)$ of them have success probability $\vectr p_1 = [P_{11},P_{12},\cdots, P_{12n}]$,
$x_2(t)$ have success probability $\vectr p_2 = [P_{21},P_{22}\cdots, P_{22n}]$,
and so on.
Take $\delta = 1$:
\begin{align*}
\lim\limits_{n_p\rightarrow\infty}\frac{1}{s_{n_p}^3}\sum_{j=1}^{n_p}\mathbb E\left[|y_j-\mu_j|^3\right]
& = \lim\limits_{n_p\rightarrow\infty}\dfrac{n_p^{3/2}}{n_p^2}\dfrac{\sum_{r=1}^{2n}\sum_{i=1}^{2n}\vectr X_r(t)|\nu_i-\vectr p_r\nu|^3}{\left(\sum_{r=1}^{2n}\sum_{i=1}^{2n}\vectr X_r(t)|\nu_i-\vectr p_r\nu|^2\right)^{3/2}}
= 0.
\end{align*}
The limit is zero since both numerator and denominator of the second fraction are constant and independent of $n_p$.
On the other hand,
the mean and variance can be obtained based on the direct definition of $\nu^T\vectr X(t+1)$ and relation \eqref{eq:cov_mat}.
Then we are able to conclude that
\begin{equation*}
\left(\frac{\nu^T\vectr X(t+1)-\nu^TP^T\vectr X(t)}{\sqrt{\nu^T\varSigma(\vectr X(t))\nu}}\bigg|\vectr X(t)\right)\xrightarrow{d}\mathcal{N}(0,1).
\end{equation*}
Defining the variable $\omega_i(t)$ through $ X_i(t+1) = \sum_{r=1}^{2n} X_r(t)P_{ri} + \omega_i(t)$ leads to
\begin{align*}
\vectr W(t)\xrightarrow{d}\mathcal{N}(0,\varSigma(\vectr X(t))).
\end{align*}
\end{proof}

\smallskip

\begin{proof}[Proof of Theorem \ref{thm:posit_def}]
The matrix $\varSigma(\vectr X)$ can be written as
$\frac{1}{n_p}\sum_{r=1}^{2n} X_r\varPhi_r$,
where
\begin{equation*}
\varPhi_r = \left[
\begin{array}{ccccc}
P_{r1}(1-P_{r1}) & -P_{r1}P_{r2} & \cdots & -P_{r1}P_{r 2n}\\
-P_{r2}P_{r1} & P_{r2}(1-P_{r2}) & \cdots & -P_{r2}P_{r 2n}\\
\vdots & \vdots & \vdots & \vdots\\
-P_{r 2n}P_{r1} & -P_{r 2n}P_{r2} & \cdots & P_{r 2n}(1-P_{r 2n})
\end{array}
\right],
\end{equation*}
\begin{equation*}
\varPhi_r = \left[
\begin{array}{ccccc}
P_{r1} & 0 & \cdots & 0\\
0 & P_{r2} & \cdots & 0\\
0 & 0 & \ddots & 0\\
0 & 0 & \cdots & P_{r 2n}
\end{array}
\right]-\left[
\begin{array}{c}
P_{r1}\\ P_{r2}\\ \vdots\\ P_{r 2n}
\end{array}
\right]\left[
\begin{array}{c}
P_{r1}\\ P_{r2}\\ \cdots\\ P_{r 2n}
\end{array}
\right]^T.
\end{equation*}
The positive semi-definiteness of all $\varPhi_r$ implies the positive semi-definiteness of $\varSigma(\vectr X)$, for all $X_r\ge 0$.
Further, the above structure of matrix $\varPhi_r$ allows us to compute, for all $\nu\in \mathbb R^{2n}$,
\begin{equation*}
\nu^T\varPhi_r \nu = \sum\limits_{i=1}^{2n}P_{ri}\nu_i^2-\left( \sum\limits_{i=1}^{2n}P_{ri}\nu_i\right)^2.
\end{equation*}
We use the Cauchy-Schwartz inequality, 
$|\vectr{a}\cdot \vectr{b}|\le \Vert \vectr{a}\Vert_2\times \Vert \vectr{b}\Vert_2$,
to show that $\nu^T\varPhi_r \nu\ge 0$. Consider two vectors
\begin{align*}
& \vectr {a} = \left[
\begin{array}{llll}
\sqrt{P_{r1}} & \sqrt{P_{r2}} & \cdots & \sqrt{P_{r 2n}}
\end{array}
\right]^T,\\
& \vectr {b} = \left[
\begin{array}{llll}
\nu_1\sqrt{P_{r1}} & \nu_2\sqrt{P_{r2}} & \cdots & \nu_{2n}\sqrt{P_{r 2n}}
\end{array}
\right]^T.
\end{align*}
The 2-norm of the vector $\vectr{a}$ is clearly equal to one, then
\begin{align*}
\left(\sum\limits_{i=1}^{2n}P_{ri}\nu_i\right)^2 \le \sum\limits_{i=1}^{2n}P_{ri}\sum\limits_{i=1}^{2n}P_{ri}\nu_i^2\Rightarrow
\nu^T\varPhi_r \nu\ge 0,
\end{align*}
The equality holds at least for the vectors $\nu = c \mathfrak 1_{2n}^T$,
where
$c$ is an arbitrary constant.

In order to prove the second part of the theorem we define the random variable
$\omega = \mathfrak 1_{2n} \vectr W = \sum\limits_{r=1}^{2n}\omega_r$, which is a linear combination of multivariate normal random vector.
Then it is a univariate normal random variable characterized by
\begin{align*}
& \mathbb E[\omega] = \mathbb E[\mathfrak 1_{2n}\vectr W] = \mathfrak 1_{2n} \mathbb E[\vectr W] = 0\Rightarrow\\
& \sigma^2(\omega) = \mathbb E[\omega \omega^T] = \mathfrak 1_{2n} \varSigma(X)\mathfrak 1_{2n}^T 
= \sum_{r=1}^{2n} X_r \mathfrak 1_{2n}\varPhi_r \mathfrak 1_{2n}^T = 0.
\end{align*}
Then the random variable $\omega$ is in fact deterministic: $\omega = 0$.

The last part of the theorem is proven by taking the sum of all the equations of the dynamical system
and noticing that the matrix $P$ is a stochastic matrix:
\begin{equation*}
\sum_{r=1}^{2n} X_r(t+1) = \sum_{r=1}^{2n} X_r(t) + \sum_{r=1}^{2n} \omega(t) = \sum_{r=1}^{2n} X_r(t).
\end{equation*}
\end{proof}

\smallskip


%
\begin{proof}[Proof of Theorem \ref{thm:reach_continuity}]
We prove the statement for one of the continuity regions, namely $m=0$ and $\theta,\theta'\in (-\infty,\theta_+]$, 
the other regions being treated in the same way. 
Consider the following chain of inequalities: 
\begin{align*}
& |\mathcal V_k(m,\theta)-\mathcal V_k(m,\theta')|\\
& =\left| \int_{\mathbb R}\mathcal V_{k+1}(0,\bar\theta)t_w(\bar\theta-a\theta-(1-a)\theta_a)d\bar\theta
-\int_{\mathbb R}\mathcal V_{k+1}(0,\bar\theta)t_w(\bar\theta-a\theta'-(1-a)\theta_a)d\bar\theta\right|\\
& \le\int_{\mathbb R}\mathcal V_{k+1}(0,\bar\theta)\left| t_w(\bar\theta-a\theta-(1-a)\theta_a) - t_w(\bar\theta-a\theta'-(1-a)\theta_a)\right|d\bar\theta\\
& \le\int_{\mathbb R}\left| t_w(\bar\theta-a\theta-(1-a)\theta_a) - t_w(\bar\theta-a\theta'-(1-a)\theta_a)\right|d\bar\theta\\
& = \frac{1}{\sigma}\int_{\mathbb R}\left|\phi\left(\frac{\bar\theta-a\theta-(1-a)\theta_a}{\sigma}\right)
-\phi\left(\frac{\bar\theta-a\theta'-(1-a)\theta_a}{\sigma}\right)\right|d\bar\theta\\
& = \int_{\mathbb R}\left|\phi\left(u-\frac{a(\theta-\theta')}{2\sigma}\right) -\phi\left(u+\frac{a(\theta-\theta')}{2\sigma}\right)\right|d\bar\theta\\
& = 2\int_{0}^{\infty}\left[\phi\left(u-\frac{a|\theta-\theta'|}{2\sigma}\right) -\phi\left(u+\frac{a|\theta-\theta'|}{2\sigma}\right)\right]d\bar\theta\\
& = 2\int_{-a|\theta-\theta'|/2\sigma}^{\infty}\phi(v)d\bar\theta - 2\int_{a|\theta-\theta'|/2\sigma}^{\infty}\phi(v)d\bar\theta\\
& = 2\int_{-a|\theta-\theta'|/2\sigma}^{a|\theta-\theta'|/2\sigma}\phi(v)d\bar\theta
\le 2\left(\frac{a|\theta-\theta'|}{2\sigma}+\frac{a|\theta-\theta'|}{2\sigma}\right)\frac{1}{\sqrt{2\pi}}
 = \frac{2a}{\sigma\sqrt{2\pi}}|\theta-\theta'|.
\end{align*}
\end{proof}

\smallskip

\begin{proof}[Proof of Theorem \ref{thm:hetr_mean}]
As we discussed for the homogeneous case, 
$(x_i(t+1)|\vectr z(t))$ is the sum of $n_p$ Bernoulli trials -- however now they allow different success probabilities. Then
\begin{align*}
& \mathbb E[x_i(t+1)|\vectr z(t)] = \sum_{j=1}^{n_p}P_{z_j(t)i}(\alpha_j)\Rightarrow\\
& \sum_{\vectr z(t)\rightarrow\vectr x(t)}\mathbb E[x_i(t+1)| z(t)] 
= \sum_{\vectr z(t)\rightarrow\vectr x(t)}\sum_{j=1}^{n_p}P_{z_j(t)i}(\alpha_j)
= \sum_{j=1}^{n_p}\sum_{\vectr z(t)\rightarrow\vectr x(t)}P_{z_j(t)i}(\alpha_j)
= \sum_{j=1}^{n_p}\sum_{r=1}^{2n}\beta_{rj} P_{r i}(\alpha_j)
\end{align*}
By changing the order of the summation, we can replace 1) by 2) in the following:
\begin{enumerate}
\item
fix the state of all Markov chains,
compute the sum of all probabilities of jumping to bin $i$,
finally sum over the states of $\varXi$ that satisfy $\vectr z(t)\rightarrow\vectr x(t)$.
\item
fix the Markov chain $\mathcal M_{\alpha}$, sum the probabilities of its jump to bin $i$ for all combinations $\vectr z(t)\rightarrow\vectr x(t)$,
finally sum over all Markov chains.
\end{enumerate}
In the latter case the addend of the inner sum has only $2n$ possibilities and we only need to count how many times each probability appears in the summation. 
These quantities are denoted by $\beta_{rj}$ as the number of the appearance of $P_{r i}$ of $\mathcal M_{\alpha_j}$. 
This number can be quantified as follows.  
The total number of states $\vectr z(t)$ generating the label $\vectr x(t) = [j_1,j_2,\cdots,j_{2n}]^T$ is
$n_p!/(j_1!j_2!\cdots j_{2n}!)$.
We know that the Markov chain $\mathcal M_{\alpha_j}$ is in state $z_j(t) = r$ and is jumping to state $i$.
For the remaining Markov chains the state
\begin{equation*}
[z_1(t),\cdots,z_{j-1}(t),z_{j+1}(t),\cdots,z_{n_p}(t)]^T\rightarrow [j_1,\cdots,j_{r-1},j_r-1,j_{r+1},\cdots,j_{2n}]^T.
\end{equation*}
Then the number of possibilities is
\begin{equation*}
\beta_{rj} = \dfrac{(n_p-1)!}{j_1!\cdots j_{r-1}!(j_r-1)!\cdots j_{2n}!}.
\end{equation*}

Finally we have:
\begin{align*}
\mathbb E[x_i(t+1)& | \vectr x(t)]
=\frac{j_1!j_2!\cdots j_{2n}!}{n_p!}\sum_{j=1}^{n_p}\sum_{r=1}^{2n}\dfrac{(n_p-1)! P_{r i}(\alpha_j)}{j_1!\cdots j_{r-1}!(j_r-1)!\cdots j_{2n}!}\\
& = \sum_{j=1}^{n_p}\sum_{r=1}^{2n}\frac{j_r}{n_p}P_{r i}(\alpha_j) = \sum_{r=1}^{2n}j_r\frac{1}{n_p}\sum_{j=1}^{n_p}P_{r i}(\alpha_j)
= \sum_{r=1}^{2n}j_r \int P_{r i}(v)f_{\alpha}(v)d v = \sum_{r=1}^{2n}j_r \overline{P_{r i}},\\
\Rightarrow \mathbb E[X_i(t&+1)|\vectr X(t)] = \sum_{r=1}^{2n}X_r(t) \overline{P_{r i}}.
\end{align*}
Now we look at the second moment of $(x_i(t+1)|\vectr z(t))$:
\begin{align*}
& \mathbb E[x_i^2(t+1)|\vectr z(t)]
= \sigma^2(x_i(t+1)|\vectr z(t))+\left(\mathbb E[x_i(t+1)|\vectr z(t)]\right)^2\\
& = \sum\limits_{j=1}^{n_p} P_{z_j(t)i}(\alpha_j)(1-P_{z_j(t)i}(\alpha_j))+\left(\sum\limits_{j=1}^{n_p} P_{z_j(t)i}(\alpha_j) \right)^2.
\end{align*}
Taking the same steps as for the first term leads to
\begin{align*}
\dfrac{\sum\limits_{\vectr z(t)\rightarrow\vectr x(t)}
\sigma^2(x_i(t+1)|\vectr z(t))}{\#\left\{\vectr z(t)\rightarrow\vectr x(t)\right\}}
 = \sum\limits_{r=1}^{2n}j_r \mathbb E_\alpha[P_{ri}(\alpha)(1-P_{ri}(\alpha))].
\end{align*}
For the second term we take the following steps:
\begin{align*}
& \sum_{\vectr z(t)\rightarrow\vectr x(t)} \left(\sum\limits_{j=1}^{n_p} P_{z_j(t)i}(\alpha_j) \right)^2
= \sum_{\vectr z(t)\rightarrow\vectr x(t)}\sum\limits_{j=1}^{n_p}\sum\limits_{u=1}^{n_p} P_{z_j(t)i}(\alpha_j)P_{z_u(t)i}(\alpha_u) \\
& = \sum\limits_{j,u=1}^{n_p}\sum_{\vectr z(t)\rightarrow\vectr x(t)}P_{z_j(t)i}(\alpha_j)P_{z_u(t)i}(\alpha_u)
= \sum\limits_{j,u=1}^{n_p}\sum\limits_{r,s=1}^{2n}\gamma_{rs} P_{r i}(\alpha_j)P_{s i}(\alpha_u),
\end{align*}
where
\begin{equation*}
\gamma_{rs} = \left\{
\begin{array}{lc}
\dfrac{(n_p-2)!}{j_1!\cdots (j_r-1)!\cdots(j_s-1)!\cdots j_{2n}!} & r\ne s\\
\dfrac{(n_p-2)!}{j_1!\cdots (j_r-2)!\cdots j_{2n}!} & r = s. 
\end{array}
\right.
\end{equation*}
Then we have 
\begin{align*}
& \frac{\sum\limits_{\vectr z(t)\rightarrow\vectr x(t)}\left(\sum\limits_{j=1}^{n_p} P_{z_j(t)i}(\alpha_j) \right)^2}
{\#\left\{\vectr z(t)\rightarrow\vectr x(t)\right\}}\\
& = \sum\limits_{j,u=1}^{n_p}\sum\limits_{r=1}^{2n}\frac{j_r(j_r-1)}{n_p(n_p-1)} P_{r i}(\alpha_j)P_{r i}(\alpha_u)
+ \sum\limits_{j,u=1}^{n_p}\sum\limits_{r,s=1,r\ne s}^{2n}\frac{j_r j_s}{n_p(n_p-1)} P_{r i}(\alpha_j)P_{s i}(\alpha_u)\\
& = \sum\limits_{j,u=1}^{n_p}\sum\limits_{r,s=1}^{2n}\frac{j_r j_s}{n_p(n_p-1)} P_{r i}(\alpha_j)P_{s i}(\alpha_u)
- \sum\limits_{j,u=1}^{n_p}\sum\limits_{r=1}^{2n}\frac{j_r}{n_p(n_p-1)} P_{r i}(\alpha_j)P_{r i}(\alpha_u)\\
& = \sum\limits_{r,s=1}^{2n}\frac{j_r j_s}{n_p(n_p-1)}\sum\limits_{j=1}^{n_p} P_{r i}(\alpha_j)\sum\limits_{u=1}^{n_p} P_{s i}(\alpha_u)
- \sum\limits_{r=1}^{2n}\frac{j_r}{n_p(n_p-1)}\left(\sum\limits_{j=1}^{n_p}P_{r i}(\alpha_j)\right)^2\\
& = \frac{n_p}{n_p-1}\sum\limits_{r,s=1}^{2n}j_r j_s \mathbb E_\alpha[P_{r i}(\alpha)]\mathbb E_\alpha[P_{s i}(\alpha)]
- \frac{n_p}{n_p-1}\sum\limits_{r=1}^{2n}j_r\left(E_\alpha[P_{r i}(\alpha)]\right)^2\\
& = \left(\sum\limits_{r=1}^{2n}j_r \mathbb E_\alpha[P_{r i}(\alpha)]\right)^2
+ \frac{1}{n_p-1}\left(\sum\limits_{r=1}^{2n}j_r \mathbb E_\alpha[P_{r i}(\alpha)]\right)^2
- \frac{n_p}{n_p-1}\sum\limits_{r=1}^{2n}j_r\left(\mathbb E_\alpha[P_{r i}(\alpha)]\right)^2.
\end{align*}
Dividing both sides by $n_p^2$ gives:
\begin{align*}
\mathbb E[X_i^2(t+1)|\vectr X(t)]
& = \frac{1}{n_p}\sum\limits_{r=1}^{2n}X_r \mathbb E_\alpha[P_{ri}(\alpha)(1-P_{ri}(\alpha))]
+ \left(\sum\limits_{r=1}^{2n}X_r \overline{P_{r i}}\right)^2\\
& + \frac{1}{n_p-1}\left(\sum\limits_{r=1}^{2n}X_r \overline{P_{r i}}\right)^2
-\frac{1}{n_p-1}\sum\limits_{r=1}^{2n}X_r \overline{P_{r i}}^2.
\end{align*}
Subtracting the square of the mean value 
$\left(\sum\limits_{r=1}^{2n}X_r \overline{P_{r i}}\right)^2$,
from both sides will give the desired formula for the variance.
Similarly, we have
\begin{align*}
& \mathbb E[x_i(t+1) x_{i'}(t+1)|z(t)]
= \sum\limits_{j=1}^{n_p}\sum\limits_{u=1,u\ne j}^{n_p} P_{z_j(t) i}(\alpha_j)P_{z_u(t) i'}(\alpha_u)\\
& = \sum\limits_{j=1}^{n_p}\sum\limits_{u=1}^{n_p} P_{z_j(t) i}(\alpha_j)P_{z_u(t) i'}(\alpha_u)
- \sum\limits_{j=1}^{n_p} P_{z_j(t) i}(\alpha_j)P_{z_j(t) i'}(\alpha_j).
\end{align*}
The first term is treated like the above theorem and gives the following:
\begin{align*}
\frac{\sum\limits_{\vectr z(t)\rightarrow\vectr x(t)}\sum\limits_{j=1}^{n_p}\sum\limits_{u=1}^{n_p} P_{z_j(t) i}(\alpha_j)P_{z_u(t) i'}(\alpha_u)}{\#\left\{\vectr z(t)\rightarrow\vectr x(t)\right\}}
& = \frac{n_p}{n_p-1}\sum\limits_{r=1}^{2n}j_r \mathbb E_\alpha[P_{r i}(\alpha)]\sum\limits_{s=1}^{2n}j_s \mathbb E_\alpha[P_{s i'}(\alpha)]\\
& -\frac{n_p}{n_p-1}\sum\limits_{r=1}^{2n}j_r \mathbb E_\alpha[P_{r i}(\alpha)]\mathbb E_\alpha[P_{r i'}(\alpha)].
\end{align*}
The second term is also manipulated in a same way:
\begin{align*}
&\frac{\sum\limits_{\vectr z(t)\rightarrow\vectr x(t)}\sum\limits_{j=1}^{n_p} P_{z_j(t) i}(\alpha_j)P_{z_j(t) i'}(\alpha_j)}
{\#\left\{\vectr z(t)\rightarrow\vectr x(t)\right\}}
= \sum_{j=1}^{2n}j_r \mathbb E_\alpha[P_{r i}(\alpha)P_{r i'}(\alpha)].
\end{align*}
Adding these terms together, diving by $n_p^2$, and subtracting the expected value concludes the proof.
\end{proof}

\smallskip

\noindent\emph{Proof of the error bounds. }
We denote the tail of the Gaussian density function by
\begin{equation*}
Q(\gamma) = \int_{\gamma}^{+\infty}\phi(u)du,\quad \phi(u) = \frac{1}{\sqrt{2\pi}}e^{-u^2/2}, 
\end{equation*}
which can be bounded as follows \cite{D04}: 
\begin{equation*}
Q(\gamma)\le \frac{\phi(\gamma)}{\gamma} = \frac{1}{\gamma\sqrt{2\pi}}e^{-\gamma^2/2},\quad \forall\gamma\in\mathbb R^{>0}.
\end{equation*}
The above inequality provides a convergence rate for the limit $\lim\limits_{\gamma\rightarrow+\infty}Q(\gamma) = 0$.
In other words, for any $\epsilon_0>0$ there exists a $\gamma_0>0$ such that $Q(\gamma)<\gamma_0$ for any $\gamma>\gamma_0$.
For instance $Q(\gamma)\le 10^{-6}$ for $\gamma\ge 5$.
The function $Q(\gamma)$ is monotonically decreasing for all $\gamma$.

Consider the following dynamical system with i.i.d. Gaussian process noise $\eta(k)$:
\begin{equation*}
x(k+1) = ax(k)+b+\eta(k),\quad a>0,\quad \eta(k)\sim\mathcal N(0,\sigma).
\end{equation*}
Define a probabilistic safety problem for this Markov process \cite{APLS08} as 
\begin{equation*}
p_{x_0}(A) = \mathsf P_{x_0}\{x(k)\in A, \text{ for all } k\in\mathbb Z_N, x(0) = x_0\}.
\end{equation*}
The solution of this safety problem can be characterized by the value functions $V_k:\mathbb R\rightarrow [0,1]$,
initialized with $V_N(x) = \mathfrak 1_{A}(x)$, and satisfying the recursion
\begin{equation*}
V_{k}(x) = \mathfrak 1_{A}(x)\int_{\mathbb R}V_{k+1}(\bar x)\frac{1}{\sigma}\phi\left(\frac{\bar x-ax-b}{\sigma}\right)d\bar x,
\quad \forall x\in\mathbb R, k\in\mathbb Z_{N-1}.
\end{equation*}
Then $p_{x_0}(A) = V_0(x_0)$. We are interested in the asymptotic properties of the function $p_{x_0}(A)$.
\begin{lemma}
\label{lm:conv_safety_upper}
The solution of the probabilistic safety problem for the above Markov process with a given safe set $ A = [\mathsf a,+\infty)$
converges to $1$ for large values of the initial state:
$\lim_{x_0\rightarrow+\infty}p_{x_0}(A) = 1.$
Similarly for a safe set $ B = (-\infty,\mathsf b]$, we have $\lim_{x_0\rightarrow-\infty}p_{x_0}(B) = 1$.
All the value functions $V_{k}$ present the same limiting behavior. 
\end{lemma}

\smallskip

\begin{proof}[Proof of Lemma \ref{lm:conv_safety_upper}]
%
Fix an arbitrary positive parameter $\gamma$ and construct the sequence $\{\gamma_k\}_{k=0}^{N}$: 
\begin{align*}
\gamma_{k} = \max\{(\gamma_{k+1}+\gamma\sigma-b)/a, \mathsf a\},\quad \gamma_N = \mathsf a.
\end{align*}
We claim that $V_k(x)\ge 1-(N-k)Q(\gamma)$, for all $x\ge \gamma_k$, which is proved by induction.
The statement is true for $k=N$ since $V_N(x) = \mathfrak 1_{A}(x)$ and $A = [\mathsf a,+\infty)$.
Suppose the statement is true for $(k+1)$, we prove it for $k$.
Consider the variable $x\ge \gamma_{k}\ge \mathsf a$, then:
\begin{align*}
&V_{k}(x)  =  \int_{-\infty}^{+\infty}V_{k+1}(\bar x)\frac{1}{\sigma}\phi\left(\frac{\bar x-ax-b}{\sigma}\right)d\bar x\Rightarrow\\
&1 - V_{k}(x) =
\int_{-\infty}^{+\infty}(1-V_{k+1}(\bar x))\frac{1}{\sigma}\phi\left(\frac{\bar x-ax-b}{\sigma}\right)d\bar x\\
&  = \int_{-\infty}^{\gamma_{k+1}}(1-V_{k+1}(\bar x))\frac{1}{\sigma}\phi\left(\frac{\bar x-ax-b}{\sigma}\right)d\bar x
+ \int_{\gamma_{k+1}}^{+\infty}(1-V_{k+1}(\bar x))\frac{1}{\sigma}\phi\left(\frac{\bar x-ax-b}{\sigma}\right)d\bar x\\
& \le \int_{-\infty}^{\gamma_{k+1}}\frac{1}{\sigma}\phi\left(\frac{\bar x-ax-b}{\sigma}\right)d\bar x
+ \int_{\gamma_{k+1}}^{+\infty}(N-k-1)Q(\gamma)\frac{1}{\sigma}\phi\left(\frac{\bar x-ax-b}{\sigma}\right)d\bar x
\end{align*}
\begin{align*}
& = \int_{-\infty}^{(\gamma_{k+1}-ax-b)/\sigma}\phi\left(u\right)du
+ (N-k-1)Q(\gamma)\int_{(\gamma_{k+1}-ax-b)/\sigma}^{+\infty}\phi\left(u\right)du\\
& \le Q\left(\frac{ax+b-\gamma_{k+1}}{\sigma}\right)+(N-k-1)Q(\gamma)
\le Q\left(\frac{a\gamma_{k}+b-\gamma_{k+1}}{\sigma}\right)+(N-k-1)Q(\gamma)\\
& \le Q(\gamma)+(N-k-1)Q(\gamma) =(N-k)Q(\gamma).
\end{align*}
We have obtained that $V_0(x)\ge 1-NQ(\gamma)$ for all $x\ge \gamma_0$.
Taking a sufficiently large $\gamma$ proves the first part.
The second part can be similarly proved by constructing a sequence $\{\beta_k\}_{k=0}^{N}$ as 
\begin{align*}
\beta_{k} = \min\{(\beta_{k+1}-\gamma\sigma-b)/a, \mathsf b\},\quad \beta_N = \mathsf b.
\end{align*}
\end{proof}

\smallskip

\begin{proof}[Proof of Theorem \ref{thm:reach_asymp}]
We divide the problem into the computation of four bounds for $\theta_k$. 
The first bound is computed by studying the behavior of $\mathcal V_k(1,\theta_k)$ at $+\infty$:
\begin{align*}
\mathcal V_k(1,\theta_k)& = \mathsf P\{s(N)\in\mathcal A| m(k) = 1,\theta(k) = \theta_k\}\\
& \ge \mathsf P\{s(N)\in\{1\}\times[\theta_-,+\infty)| m(k) = 1,\theta(k) = \theta_k\}\\
& \ge \mathsf P\{s(i)\in\{1\}\times[\theta_-,+\infty),\text{ for all } i\in[k,N]| m(k) = 1,\theta(k) = \theta_k\}\\
&  = \mathsf P\{\theta(i)\in[\theta_-,+\infty),\text{ for all } i\in[k,N]| \theta(k) = \theta_k\},
\end{align*}
where $\theta(\cdot)$ satisfies the temperature dynamical equation in the ON mode.
This is exactly the safety problem studied in Lemma \ref{lm:conv_safety_upper}.
Then $\mathcal V_k(1,\theta_k)\ge 1-(N-k)Q(\gamma)$ for all $\theta_k\ge \gamma_k$,
where
\begin{equation}
\label{eq:seq1}
\gamma_{k} = \max\{(\gamma_{k+1}+\gamma\sigma-(1-a)(\theta_a-RP_{rate}))/a, \theta_-\},\quad \gamma_N = \theta_-.
\end{equation}

The second bound is computed by studying the behavior of $\mathcal V_k(0,\theta_k)$ at $-\infty$:
\begin{align*}
\mathcal V_k(0,\theta_k)
& = 1-\mathsf P\{s(N)\in \mathcal S\backslash\mathcal A| m(k) = 0,\theta(k) = \theta_k\}\\
& \le 1-\mathsf P\{s(N)\in\{0\}\times\mathbb R^{<\theta_{+}}| m(k) = 0,\theta(k) = \theta_k\}\\
& \le 1-\mathsf P\{s(i)\in\{0\}\times\mathbb R^{<\theta_{+}},\text{ for all } i\in[k,N]| m(k) = 0,\theta(k) = \theta_k\}\\
&  = 1-\mathsf P\{\theta(i)\in\mathbb R^{<\theta_{+}},\text{ for all } i\in[k,N]| \theta(k) = \theta_k\},
\end{align*}
where $\theta(\cdot)$ satisfies the temperature dynamical equation in the OFF mode.
This is the complement of the safety problem studied in Lemma \ref{lm:conv_safety_upper}.
Then $\mathcal V_k(0,\theta_k)\le (N-k)Q(\gamma)$ for all $\theta_k\le \beta_k$,
where
\begin{equation}
\label{eq:seq2}
\beta_{k} = \min\{(\beta_{k+1}-\gamma\sigma-(1-a)\theta_a)/a, \theta_+\},\quad \beta_N = \theta_+.
\end{equation}

The third bound is provided by the behavior of $\mathcal V_k(0,\theta_k)$ at $+\infty$.
Take the value $\theta_k\ge\theta_+$,
\begin{align*}
\mathcal V_k(0,\theta_k) & = \mathsf P\{s(N) \in \mathcal A| m(k) = 0,\theta(k) = \theta_k\}\\
& = \mathbb E_{s_{k+1}}\left[\mathsf P\{s(N)\in \mathcal A| s(k+1) = s_{k+1}\}|m(k) = 0,\theta(k) = \theta_k \right]\\
& = \mathbb E_{s_{k+1}}\left[\mathcal V_{k+1}(s_{k+1})|m(k) = 0,\theta(k) = \theta_k \right]\\
& = \int_{\mathbb R}\mathcal V_{k+1}(1,\theta_{k+1})t_w(\theta_{k+1}-a\theta_k-(1-a)\theta_a)d\theta_{k+1}.
\end{align*}
Then we have $\mathcal V_k(0,\theta_k)\ge 1-(N-k)Q(\gamma)$, for all $\theta_k\ge\bar\gamma_k$, where
\begin{equation}
\label{eq:seq3}
\bar\gamma_k = \max\{(\gamma_{k+1}+\gamma\sigma-(1-a)\theta_a)/a,\theta_+\}.
\end{equation}

Finally, we study the behavior of $\mathcal V_k(1,\theta_k)$ at $-\infty$.
Take the value $\theta_k\le\theta_-$,
\begin{align*}
\mathcal V_k(1,\theta_k) & = \mathsf P\{s(N)\in \mathcal A| m(k) = 1,\theta(k) = \theta_k\}\\
& = \mathbb E_{s_{k+1}}\left[\mathsf P\{s(N) \in \mathcal A| s(k+1) = s_{k+1}\}|m(k) = 1,\theta(k) = \theta_k \right]\\
& = \mathbb E_{s_{k+1}}\left[\mathcal V_{k+1}(s_{k+1})|m(k) = 1,\theta(k) = \theta_k \right]\\
& = \int_{\mathbb R}\mathcal V_{k+1}(0,\theta_{k+1})t_w(\theta_{k+1}-a\theta_k-(1-a)(\theta_a-RP_{rate}))d\theta_{k+1}.
\end{align*}
Then we have $\mathcal V_k(1,\theta_k)\le (N-k)Q(\gamma)$, for all $\theta_k\le\bar\beta_k$, where
\begin{equation}
\label{eq:seq4}
\bar\beta_k = \min\{(\beta_{k+1}-\gamma\sigma-(1-a)(\theta_a-RP_{rate}))/a,\theta_-\}.
\end{equation}
All these bounds result in $\mathcal V_k(1,\theta_k)\le (N-k)Q(\gamma)$.
Since $\gamma$ is an arbitrary positive parameter, the proof is complete.
\end{proof}

\smallskip



\begin{proof}[Proof of Theorem \ref{thm:VF_limit}]
Since the parameter $a = e^{-h/RC}\in(0,1)$
and $[\theta_-,\theta_+]\subset [\theta_a-RP_{rate}, \theta_a]$,
the sequences introduced in \eqref{eq:seq1}, \eqref{eq:seq2} are monotonic and satisfy the following linear difference equations
\begin{align*}
& \gamma_{k} = (\gamma_{k+1}+\gamma\sigma-(1-a)(\theta_a-RP_{rate}))/a,\quad \gamma_N = \theta_-,\\
& \beta_{k} = (\beta_{k+1}-\gamma\sigma-(1-a)\theta_a)/a,\quad \beta_N = \theta_+.
\end{align*}
The sequences introduced in \eqref{eq:seq3}, \eqref{eq:seq4} are also monotonic. 
To show the correctness of the statement it is sufficient to find a $\gamma$, 
such that $\gamma_0\le\theta_{\mathsf m}$ and
$\beta_0\ge\theta_{-\mathsf m}$.
Note that by such a selection the conditions $\bar\gamma_1\le\theta_{\mathsf m}$ and $\bar\beta_1\ge\theta_{-\mathsf m}$
are automatically satisfied. We have: 
\begin{align*}
& \gamma_{0}\le\theta_{\mathsf m}\Rightarrow
\gamma\le\frac{1-a}{\sigma}\left[\frac{a^N\theta_{\mathsf m}-\theta_-}{1-a^N}+\theta_a-RP_{rate} \right],\\
& \beta_{0}\ge\theta_{-\mathsf m}\Rightarrow
\gamma\le\frac{1-a}{\sigma}\left[\frac{\theta_+-a^N\theta_{-\mathsf m}}{1-a^N}-\theta_a\right].
\end{align*}
Taking the minimum of the right hand-sides leads to the formulation of $\gamma$ in the theorem.

\end{proof}

\smallskip

\begin{proof}[Proof of Theorem \ref{thm:error_1}]
Let the vector $\bar{\mathcal V}_k$ be the solution of problem \eqref{eq:reach_def} for the Markov chain. 
The entries of this vector contain the values of the piecewise constant function $\mathcal W_k$ at the corresponding partition set.
For the absorbing states we have in particular 
\begin{equation*}
\bar{\mathcal V}_k(1) =\bar{\mathcal V}_k(n+1) =  0,\quad
\bar{\mathcal V}_k(n) = \bar{\mathcal V}_k(2n) = 1,\quad\forall k\in \mathbb N_N.
\end{equation*}
Based on Theorem \ref{thm:VF_limit} we have that $|\mathcal V_k(m,\theta)-\mathcal W_k(m,\theta)|\le (N-k)\epsilon$,
for all $(m,\theta)$ belonging to the infinite length intervals.

Recall that the value functions $\mathcal V_{k}$ satisfy the recursion in \eqref{eq:recur_prob}.
We discuss this step recursion for $m = 0,\theta_+\le \theta\le \theta_{\mathsf m}$, 
the other four possibilities being the same. 
Suppose that $\theta\in\Theta_i$ with representative point $\bar\theta_i$:
\begin{align*}
|\mathcal V_{k}(0,\theta) - \mathcal W_{k}(0,\theta)|
& \le |\mathcal V_{k}(0,\theta) - \mathcal V_{k}(0,\bar{\theta}_i)|
+ |\mathcal V_{k}(0,\bar{\theta}_i) - \mathcal W_{k}(0,\bar{\theta}_i)|\\
& \le \frac{2a}{\sigma\sqrt{2\pi}}|\theta-\bar{\theta}_i|+
|\mathcal V_{k}(0,\bar{\theta}_i) - \mathcal W_{k}(0,\bar{\theta}_i)|
\end{align*}
\begin{align*}
\mathcal V_{k}(0,\bar{\theta}_i) & = \int_{\mathbb R}\mathcal V_{k+1}(1,\bar\theta)t_w(\bar\theta-a\bar{\theta}_i-(1-a)\theta_a)d\bar\theta
= \int_{-\infty}^{\theta_{-\mathsf m}}\hspace{-0.1in}\mathcal V_{k+1}(1,\bar\theta)t_w(\bar\theta-a\bar{\theta}_i-(1-a)\theta_a)d\bar\theta\\
& + \int_{\theta_{-\mathsf m}}^{\theta_{\mathsf m}}\hspace{-0.05in}\mathcal V_{k+1}(1,\bar\theta)t_w(\bar\theta-a\bar{\theta}_i-(1-a)\theta_a)d\bar\theta
+ \int_{\theta_{\mathsf m}}^{\infty}\hspace{-0.05in}\mathcal V_{k+1}(1,\bar\theta)t_w(\bar\theta-a\bar{\theta}_i-(1-a)\theta_a)d\bar\theta,
\end{align*}
\begin{align*}
\mathcal W_{k}(0,\bar{\theta}_i) & = \int_{\mathbb R}\mathcal W_{k+1}(1,\bar\theta)t_w(\bar\theta-a\bar{\theta}_i-(1-a)\theta_a)d\bar\theta
= \int_{-\infty}^{\theta_{-\mathsf m}}\hspace{-0.11in}\mathcal W_{k+1}(1,\bar\theta)t_w(\bar\theta-a\bar{\theta}_i-(1-a)\theta_a)d\bar\theta\\
& + \int_{\theta_{-\mathsf m}}^{\theta_{\mathsf m}}\hspace{-0.06in}\mathcal W_{k+1}(1,\bar\theta)t_w(\bar\theta-a\bar{\theta}_i-(1-a)\theta_a)d\bar\theta
+ \int_{\theta_{\mathsf m}}^{\infty}\hspace{-0.06in}\mathcal W_{k+1}(1,\bar\theta)t_w(\bar\theta-a\bar{\theta}_i-(1-a)\theta_a)d\bar\theta,
\end{align*}
\begin{align*}
\Rightarrow & |\mathcal V_{k}(0,\bar{\theta}_i)-\mathcal W_{k}(0,\bar{\theta}_i)| \le 
\int_{-\infty}^{\theta_{-\mathsf m}}|\mathcal V_{k+1}(1,\bar\theta)-0|t_w(\bar\theta-a\bar{\theta}_i-(1-a)\theta_a)d\bar\theta\\
& +\int_{\theta_{-\mathsf m}}^{\theta_{\mathsf m}}|\mathcal V_{k+1}(1,\bar\theta)-\mathcal W_{k+1}(1,\bar\theta)|t_w(\bar\theta-a\bar{\theta}_i-(1-a)\theta_a)d\bar\theta\\
& + \int_{\theta_{\mathsf m}}^{\infty}|\mathcal V_{k+1}(1,\bar\theta)-1|t_w(\bar\theta-a\bar{\theta}_i-(1-a)\theta_a)d\bar\theta
\end{align*}
\begin{align*}
& \le (N-k-1)\epsilon\int_{-\infty}^{\theta_{-\mathsf m}}t_w(\bar\theta-a\bar{\theta}_i-(1-a)\theta_a)d\bar\theta\\
& + E_{k+1}\int_{\theta_{-\mathsf m}}^{\theta_{\mathsf m}}t_w(\bar\theta-a\bar{\theta}_i-(1-a)\theta_a)d\bar\theta
+ (N-k-1)\epsilon\int_{\theta_{\mathsf m}}^{\infty}t_w(\bar\theta-a\bar{\theta}_i-(1-a)\theta_a)d\bar\theta\\
& \le (N-k-1)\epsilon+E_{k+1}
\end{align*}
\begin{align*}
& \Rightarrow E_{k} = \frac{2a}{\sigma\sqrt{2\pi}}\upsilon+ (N-k-1)\epsilon + E_{k+1},\quad E_N = 0,\\
& \Rightarrow E_{1} = \frac{(N-1)(N-2)}{2}\epsilon + (N-1)\frac{2a}{\sigma\sqrt{2\pi}}\upsilon,\quad \quad\forall (m_0,\theta_0)\in\mathbb Z_1\times[\theta_{-\mathsf m},\theta_{\mathsf m}].
\end{align*}
\end{proof}

\smallskip

\begin{proof}[Proof of Theorem \ref{thm:error_hom}]
The total power consumption is the sum of $n_p$ independent Bernoulli trials with different success probabilities:
\begin{align}
& \mathbb E[y(N)|\vectr m_0,\vectr{\theta}_0] = P_{rate,ON}\sum_{j=1}^{n_p}\mathbb E[m(N)|m_{0j},\theta_{0j}]\label{eq:sum_power}\\
& \mathbb E[y_{abs}(N)|\vectr X_0] = H (P^T)^N \vectr X_0 = P_{rate,ON}\sum_{i=1}^{2n}n_pX_{0i}\mathbb E[\bar m(N)|\bar{m}_{0i},\bar{\theta}_{0i}].\nonumber
\end{align}
Then we obtain 
\begin{align*}
& \Rightarrow |\mathbb E[y(N)|m_0,\theta_0] - \mathbb E[y_{abs}(N)|X_0]| \le\\
& \le P_{rate,ON}\sum_{j=1}^{n_p}|\mathbb E[m(N)|m_{0j},\theta_{0j}] - \mathbb E[\bar m(N)|\bar{m}_{0j},\bar{\theta}_{0j}]|
\le P_{rate,ON}n_p E_{1}.
\end{align*}
\end{proof}

\smallskip

\begin{proof}[Proof of Theorem \ref{thm:elim_redun}]
We apply the linear transformation $\tilde X = T X$ with
\begin{equation*}
T = \left[
\begin{array}{cc}
I_{2n-1} & 0_{2n-1}^T\\
\mathfrak 1_{2n-1} & 1
\end{array}
\right]
\Rightarrow
T^{-1} = \left[
\begin{array}{cc}
I_{2n-1} & 0_{2n-1}^T\\
-\mathfrak 1_{2n-1} & 1
\end{array}
\right].
\end{equation*}
The dynamical equation in \eqref{eq:dynamical_approx} becomes 
\begin{equation*}
\tilde X(t+1) = T P^T T^{-1}\tilde X(t)+T W(t),
\end{equation*}
then
\begin{align*}
T P^T T^{-1} & = \left[
\begin{array}{cc}
\Omega_{11}^T-\Omega_{21}^T\mathfrak 1_{2n-1} & \Omega_{21}^T\\
0_{2n-1} & 1
\end{array}
\right],\quad
T W(t) = \left[
\begin{array}{c}
\bar W(t)\\
0
\end{array}
\right].
\end{align*}
The last equation indicates that $\tilde X_{2n}$ is always equal to its initial value.
Since the sum of all state variables of \eqref{eq:dynamical_approx} are equal to one,
this is indeed one.
We then replace $\tilde X_{2n}$ by one in the first $(2n-1)$ equations and omit the last equation,
\begin{align*}
\bar X = \left[\mathbb I_{2n-1}, 0 \right]\tilde X
\Rightarrow
\bar X(t+1) = (\Omega_{11}^T-\Omega_{21}^T\mathfrak 1_{2n-1})\bar X(t) + \Omega_{21}^T + \bar W(t).
\end{align*}
Applying a same transformation to the output equation will result in the matrices $C,D$.
\end{proof}

\smallskip

\begin{proof}[Proof of Theorem \ref{thm:cluster_error}]
Relation \eqref{eq:sum_power} and Theorem \ref{thm:error_1} indicate that the first part of the error in \eqref{eq:cluster_error} is an upper-bound for the sum of abstraction error of each single TCL.

The second part of the error is proved by studying the sensitivity of the solution of the problem \eqref{eq:reach_def} against parameter $\alpha$. 
As we discussed before, the solution of this problem for the Markov chain over the time horizon $N$ is obtained by the recursion
$\bar{\mathcal V}_{k}(\alpha) = P(\alpha)\bar{\mathcal V}_{k+1}(\alpha)$, where $\bar{\mathcal V}_{N}(\alpha)$ is the indicator vector of the reach set, hence independent of $\alpha$.
Then we have
\begin{align*}
\|\bar{\mathcal V}_{k}(\alpha)-\bar{\mathcal V}_{k}(\alpha')\|_\infty & = 
\|P(\alpha)\bar{\mathcal V}_{k+1}(\alpha)-P(\alpha')\bar{\mathcal V}_{k+1}(\alpha')\|_\infty\\
& \le \|\left(P(\alpha)-P(\alpha')\right)\bar{\mathcal V}_{k+1}(\alpha)\|_\infty
+\|P(\alpha')\left(\bar{\mathcal V}_{k+1}(\alpha)-\bar{\mathcal V}_{k+1}(\alpha')\right)\|_\infty\\
& \le  \|P(\alpha)-P(\alpha')\|_\infty\|\bar{\mathcal V}_{k+1}(\alpha)\|_\infty
+ \|P(\alpha')\|_\infty\|\bar{\mathcal V}_{k+1}(\alpha)-\bar{\mathcal V}_{k+1}(\alpha')\|_\infty\\
& \le h_a \|\alpha-\alpha'\| + \|\bar{\mathcal V}_{k+1}(\alpha)-\bar{\mathcal V}_{k+1}(\alpha')\|_\infty,
\end{align*} 
which results in the inequality
\begin{align*}
\|\bar{\mathcal V}_{1}(\alpha)-\bar{\mathcal V}_{1}(\alpha')\|_\infty\le (N-1) h_a\|\alpha-\alpha'\|.
\end{align*}
Define function $\xi(\cdot)$ that assigns to each $\alpha$ the representative parameter of its cluster. Then 
\begin{align*}
\bigg|\sum_{\alpha\in\Gamma_a}&P_{rate,ON}(\alpha)\bar{\mathcal V}_{1}(\alpha)-\sum_i n_i P_{rate,ON}(\alpha_i)\bar{\mathcal V}_{1}(\alpha_i)\bigg|\\
& \le \sum_{\alpha\in\Gamma_a}\left|P_{rate,ON}(\alpha)\bar{\mathcal V}_{1}(\alpha)-P_{rate,ON}(\xi(\alpha))\bar{\mathcal V}_{1}(\xi(\alpha))\right|\\
& \le \sum_{\alpha\in\Gamma_a}\left|P_{rate,ON}(\alpha)-P_{rate,ON}(\xi(\alpha))\right|\bar{\mathcal V}_{1}(\alpha)
+ \sum_{\alpha\in\Gamma_a} P_{rate,ON}(\xi(\alpha))\left|\bar{\mathcal V}_{1}(\alpha)-\bar{\mathcal V}_{1}(\xi(\alpha))\right|\\
& \le n_p\upsilon_a  + (N-1)h_a\upsilon_a \sum_{\alpha\in\Gamma_a} n_i P_{rate,ON}(\alpha_i).
\end{align*}
\end{proof}

\smallskip

For the poof of Theorem \ref{thm:cost_function} we need the following lemma.
\begin{lemma}
\label{lem:quadratic}
The following equality holds: $\nu^T \varSigma(\vectr X)\nu = \frac{1}{n_p}\mathscr R(\nu^T,P^T)\vectr X$.
\end{lemma}

\smallskip

\begin{proof}[Proof of Lemma \ref{lem:quadratic}]
Using the notation of the proof of Theorem \ref{thm:posit_def} we have
\begin{align*}
\nu^T \varSigma(\vectr X)\nu = \frac{1}{n_p}\sum_{r=1}^{2n}\nu^T \varPhi_r \nu X_r,\quad
\nu^T \varPhi_r \nu = \nu^T diag(P_r)\nu - \nu^T P_r^T P_r \nu = P_r\nu^{\circ 2} - (P_r \nu)^2,
\end{align*}
where $P_r$ is the $r^{th}$-row of the probability matrix $P$. Then
\begin{align*}
\nu^T \varSigma(\vectr X)\nu
& = \frac{1}{n_p}\sum_{r=1}^{2n}P_r\nu^{\circ 2}X_r - \frac{1}{n_p}\sum_{r=1}^{2n}(P_r \nu)^2 X_r\\
& = \frac{1}{n_p}(P\nu^{\circ 2})^T\vectr X - \frac{1}{n_p}(\nu^T P^T)^{\circ 2} \vectr X
= \frac{1}{n_p}\mathscr R(\nu^T,P^T)\vectr X.
\end{align*}
\end{proof}

\smallskip

\begin{proof}[Proof of Theorem \ref{thm:cost_function}]
We prove \eqref{eq:cost_function} for all $t\le T$.
Define the backward recursion
\begin{align*}
J_{\bar\tau} = \mathbb E\left[\left[y_{abs}(\bar\tau+1)-y_{des}(\bar\tau+1)\right]^2+J_{\bar\tau+1}\big|\vectr X(\bar\tau)\right],\quad
J_T = \kappa^T\vectr X(T).
\end{align*}
Then $\Psi_{\sigma}(T,T) = \kappa^T$. Using the dynamics of the system we get
\begin{align*}
\mathbb E\left[\left[y_{abs}(\bar\tau+1)-y_{des}(\bar\tau+1)\right]^2\big|\vectr X(\bar\tau)\right]
& =\mathbb E\left[\left[HF_{\sigma(\bar\tau)}\vectr X(\bar\tau)+H\vectr W(\bar\tau)-y_{des}(\bar\tau+1)\right]^2\big|\vectr X(\bar\tau)\right]\\
& = \left[HF_{\sigma(\bar\tau)}\vectr X(\bar\tau)-y_{des}(\bar\tau+1)\right]^2 + H\varSigma(\vectr X(\bar\tau))H^T\\
& = \left[HF_{\sigma(\bar\tau)}\vectr X(\bar\tau)-y_{des}(\bar\tau+1)\right]^2 + \frac{1}{n_p}\mathscr R(H,F_{\sigma(\bar\tau)})X(\bar\tau).
\end{align*}
This leads to: 
\begin{align*}
\mathbb E \left[J_{\bar\tau+1}\big|\vectr X(\bar\tau)\right] 
& = \mathbb E\left[\sum_{\tau=\bar\tau+2}^{T}\left[H\Phi_{\sigma}(\tau,\bar\tau+1)\vectr X(\bar\tau+1)-y_{des}(\tau)\right]^2 + \Psi_{\sigma}(T,\bar\tau+1)\vectr X(\bar\tau+1)\big|\vectr X(\bar\tau)\right]\\
& = \sum_{\tau=\bar\tau+2}^{T}\left[H\Phi_{\sigma}(\tau,\bar\tau+1)F_{\sigma(\bar\tau)}\vectr X(\bar\tau)-y_{des}(\tau)\right]^2\\
& + \frac{1}{n_p}\sum_{\tau=\bar\tau+2}^{T}\mathscr R(H\Phi_{\sigma}(\tau,\bar\tau+1),F_{\sigma(\bar\tau)})X(\bar\tau)
 + \Psi_{\sigma}(T,\bar\tau+1)F_{\sigma(\bar\tau)}\vectr X(\bar\tau)
\end{align*}
Summing up the two terms and using the characteristics of the transition matrix $\Psi$ we get
\begin{align*}
J_{\bar\tau} = \sum_{\tau=\bar\tau+1}^{T}\left[H\Phi_{\sigma}(\tau,\bar\tau)\vectr X(\bar\tau)-y_{des}(\tau)\right]^2 + \Psi_{\sigma}(T,\bar\tau)\vectr X(\bar\tau),
\end{align*}
where
\begin{align*}
\Psi_{\sigma}(T,\bar\tau) = \Psi_{\sigma}(T,\bar\tau+1)F_{\sigma(\bar\tau)}  + 
\frac{1}{n_p}\sum_{\tau=\bar\tau+1}^{T}\mathscr R(H\Phi_{\sigma}(\tau,\bar\tau+1),F_{\sigma(\bar\tau)}).
\end{align*}
This recursion admits the following explicit solution
\begin{align*}
\Psi_{\sigma}(T,\bar\tau) = \kappa^T\Phi_{\sigma}(T,\bar\tau)
+\frac{1}{n_p}\sum_{\tau_1=\bar\tau}^{T}\sum_{\tau_2=\tau_1+1}^{T}\mathscr R(H\Phi_{\sigma}(\tau_2,\tau_1+1),F_{\sigma(\tau_1)})\Phi_{\sigma}(\tau_1,\bar\tau).
\end{align*}
\end{proof}

\end{document}